%% file: main.tex
\documentclass[12pt,letter]{article}

\usepackage{mathrsfs} 
\usepackage{amsfonts}
\usepackage{dsfont}
\usepackage{amssymb}

\usepackage{titling}
\usepackage{bibunits}
\usepackage{graphicx,psfrag,epsf}
\usepackage{enumerate}
\usepackage{natbib}

\usepackage{url} 

\usepackage{float}

\usepackage{setspace}
\usepackage{tikz-network}

\usepackage{amsmath}

\usepackage{mwe}
\usepackage{subfig}

\usepackage{threeparttable}
\renewcommand{\arraystretch}{0.65} 
\usepackage{caption}

\usepackage{booktabs}

\usepackage{tikz}
\usepackage{pgfplots}

\usepackage{pgfplotstable}
\usepgfplotslibrary{groupplots}

\usepackage{pdflscape}

\newtheorem{theorem}{Theorem}

\newtheorem{corollary}{Corollary}[section]

\newtheorem{definition}{Definition}[section]

\newtheorem{lemma}{Lemma}[section]

\newtheorem{proposition}{Proposition}

\newtheorem{assumption}{Assumption}

\newcounter{subassumption}[assumption]
\renewcommand{\thesubassumption}{(\textit{\roman{subassumption}})}
\makeatother
\newcommand{\subasu}{
  \refstepcounter{subassumption}%
  \thesubassumption~\ignorespaces}

\newenvironment{proof}[1][Proof]{\noindent\textbf{#1.} }{\ \rule{0.5em}{0.5em}}
\renewcommand{\baselinestretch}{1.3}
\makeatletter
\renewcommand{\thetheorem}{\arabic{theorem}}
\@addtoreset{lemma}{Appendix B} \makeatother

\numberwithin{equation}{section}
\makeatletter
\begin{document}
\begin{bibunit}[jpe]

\def\spacingset#1{\renewcommand{\baselinestretch}%
{#1}\small\normalsize} \spacingset{1}


  \title{\bf Estimating Social Effects with Randomized and Observational Network Data}
  \author{TszKin Julian Chan\\
    Bates White Economic Consulting\\
    Juan Estrada\\
    Analysis Group Economic Consulting\\
    Kim Huynh\thanks{We thank the Editor, Associate Editor, and two anonymous referees for very helpful corrections, comments, and suggestions that improved the overall readability and presentation of the article. The views expressed in this article are those of the authors. No responsibility for them should be attributed to the Bank of Canada. All remaining errors are the responsibility of the authors.}\hspace{.2cm}\\
    Currency Department, Bank of Canada\\
    David Jacho-Ch\'{a}vez\thanks{Corresponding Author}\\
    Department of Economics, Emory University\\
    Chungsang Tom Lam\\
    Department of Finance, Florida State University\\
    Leonardo S\'{a}nchez-Arag\'{o}n\\
    Facultad de Ciencias Sociales y Human\'{i}sticas, ESPOL University}
  \maketitle

\vspace{-0.2cm}

\begin{abstract}
This paper introduces an innovative approach to identifying and estimating the parameters of interest in the widely recognized linear-in-means regression model under conditions where the initial randomization of peers determines the observed network. We assert that peers who are initially randomized do not produce social effects. However, after randomization, agents can endogenously develop significant connections that potentially generate peer influences. We present a moment condition that compiles local heterogeneous identifying information for all agents within the population. Under the assumption of $\psi$-dependence in the endogenous network space, we propose a Generalized Method of Moments (GMM) estimator, which is proven to be consistent, asymptotically normally distributed, and straightforward to implement using commonly available statistical software due to its closed-form expression. Monte Carlo simulations demonstrate the GMM estimator's strong small-sample performance. An empirical analysis utilizing data from Hong Kong high school students reveals substantial positive spillover effects on math test scores among study partners in our sample, provided that their seatmates were exogenously assigned by their teachers.
\end{abstract}

\noindent%
\vfill

\newpage
\spacingset{1.9} 

\section{Introduction}\label{introduction}
\input{introduction.tex}

\section{Preliminaries}\label{prelim}
\input{preliminaries.tex}
    
\section{Peer Effects Model and Identification}\label{sid}
\input{identification.tex}
    
\section{Estimation}\label{estimation}
\input{estimation.tex}

\section{Monte Carlo Experiments}\label{mc}
\input{simulations.tex}

\section{Empirical Illustration}\label{emp}
\input{empirical.tex}

\section{Conclusion}\label{discussion}
\input{conclusion.tex}

\putbib[typ5.bib]
\end{bibunit}

\emptythanks
\clearpage
\setcounter{page}{1}

  \title{\bf Estimating Social Effects with Randomized and Observational Network Data\\ -- Supplemental Materials --}
  \author{TszKin Julian Chan\\
    Bates White Economic Consulting\\
    Juan Estrada\\
    Analysis Group Economic Consulting\\
    Kim Huynh\thanks{The views expressed in this article are those of the authors. No responsibility for them should be attributed to the Bank of Canada. All remaining errors are the responsibility of the authors.}\hspace{.2cm}\\
    Currency Department, Bank of Canada\\
    David Jacho-Ch\'{a}vez\\
    Department of Economics, Emory University\\
    Chungsang Tom Lam\\
    Department of Finance, Florida State University\\
    Leonardo S\'{a}nchez-Arag\'{o}n\\
    Facultad de Ciencias Sociales y Human\'{i}sticas, ESPOL University}
  \maketitle

\setcounter{section}{0}

\appendix\renewcommand\thesection{Appendix \Alph{section}}
\renewcommand{\theequation}{\Alph{section}-\arabic{equation}}
\renewcommand{\theenumi}{\alph{enumi}}\renewcommand{\labelenumi}%
{\emph{(\theenumi)}}
\renewcommand{\theenumii}{(\roman{enumii})}\renewcommand{\labelenumii
}{\theenumii}
\renewcommand{\thetheorem}{\Alph{section}.\arabic{theorem}.}\renewcommand
{\theassumption}{\Alph{section}.\arabic{assumption}}

\begin{bibunit}[jpe]

\section{Primitive Conditions for Relevance\label{Appendix_D}}
\renewcommand\thesection{\Alph{section}}
\input{supp_D.tex}

\section{Proofs of Main Results\label{Appendix_A}}
\input{supp_A.tex}

\section{Auxiliary Results\label{Appendix_B}}
\renewcommand\thesection{\Alph{section}}
\input{supp_B.tex}

\section{Empirical Application\label{Appendix_C}}
\renewcommand\thesection{\Alph{section}}
\input{supp_C}

\putbib[typ5.bib]
\end{bibunit}

\end{document}

%% file: introduction.tex
In many observational studies in economics and other social settings, a unit of observation's outcome (e.g. purchase of an item, health well-being of an individual, performance of a firm, or test scores of a student) depends not only on this unit's own observed characteristics (direct effect), but also on the outcome (peer effects) and characteristics (contextual effects) of other observations (peers) in the sample with which they have a link. The workhorse model in economics and other social sciences for this type of setting is the so-called \emph{linear-in-means} regression model (see \citeauthor{manski1993}, \citeyear{manski1993} and Section 3.1 in \citeauthor{Paula2017}, \citeyear{Paula2017}, pp. 275-289) where the outcome variable for observation $i$ (e.g., a person, a firm, or a country), $y_{n,i}$, is determined according to

\begin{equation}
\label{intro_lmm}
y_{n;i}=\beta_{0}\sum_{j \neq i}w_{n;i,j}y_{n;j}+\sum_{j \neq i}w_{n;i,j}\mathbf{x}_{n;j}^{\top}\boldsymbol{\delta}_{0}+ \widetilde{\mathbf{x}}_{n;i}^{\top}\boldsymbol{\gamma}_{0}+\varepsilon_{n;i}\text{,}
\end{equation}

\noindent where $i,j\in\{1,\dots,n\}$ are also known as \emph{nodes}, $\widetilde{\mathbf{x}}_{n;i}=[1,\mathbf{x}_{n;i}^\top]^\top$, and $\mathbf{x}_{n;i}$ is a vector of attributes that characterizes observations $i$, $w_{n;i,j} \in (0,1]$ if $j$ is connected to $i$ (a potentially weighted \emph{edge}), and 0 otherwise, $\varepsilon_{n;i}$ represents an unobserved latent error, and $n$ is the number of observations or nodes in the sample. The structure of the \emph{social} network is fully characterized by the square $n \times n$ matrix, $\mathbf{W}_{n}$, with the entry $(i,j)$ given by $w_{n;i,j}$, that is, the adjacency matrix. The structural parameters $\beta_0$ and $\boldsymbol{\delta}_0$ capture the peer and contextual effects, respectively, while $\boldsymbol{\gamma}_0$ captures the direct effects of the observation's own characteristics. They are jointly known as the social (or neighbor) effect parameters and the object of interest in empirical studies with network data; see, e.g., \cite{Sacerdote_QJE} in Economics, \cite{L_i_M_public_health} in Public Health, \cite{L_i_M_criminology} in Criminology, and \cite{L_i_M_sociology} in Sociology to mention just a few.

Sufficient conditions under which the social parameters in \eqref{intro_lmm} can be uniquely recovered from the estimating sample $\left\{y_{n;i},\mathbf{x}_{n;i}^{\top},\{w_{n;i,j}\}_{j=1,j\neq i}^n\right\}_{i=1}^n$ are well understood under the assumption that the adjacency matrix, $\mathbf{W}_{n}$, is exogenous; see, for example, \cite{Paula2017} and references therein. However, the endogenous case remains an active area of research among economists and social scientists due to simultaneity bias, measurement error in link information, or because there might be a natural correlation between covariates and the error term in \eqref{intro_lmm} (homophily); see, e.g., \cite{Johnsson2019} and references therein.

This article proposes the use of a type of multilayered (multidimensional) network data structure known as \emph{multiplex} networks \citep{Boccaletti2014,Kivela_multilayer_network_2014} to consistently estimate and perform correct inferences on the structural social parameters in \eqref{intro_lmm} with a potentially endogenous network structure, $\mathbf{W}_n$. In particular, we assume that the researcher observes another set of social ties among the original observations, $\{w_{n,0;i,j}\}_{j=1,j\neq i}^n$, in the form of $n\times n$ adjacency matrix $\mathbf{W}_{n,0}$ that are exogenous in the usual sense.

The assumption of network exogeneity in $\mathbf{W}_{n,0}$ is motivated from the literature on experimental settings with interactions; see, e.g., \cite{Athey2017}. Network randomization has been used in a wide range of applied settings, such as educational achievement \citep{Sacerdote_QJE}, entrepreneurship team performance \citep{Hasan2019}, and business performance \citep{Cai2018}, just to name a few. Furthermore, recent studies have emphasized the importance of endogenous choices of individuals to interact with or avoid their randomized peers as mediators of peer effects \citep{Hasan2019}. For example, \cite{Carrell2013} finds that after assigning cadets to squadrons (initial group assignment), i.e., $\mathbf{W}_{n,0}$, they endogenously sorted into homogeneous friendship groups based on their ability, i.e., $\mathbf{W}_n$. In the context of worker productivity, \cite{Mas2009}, \cite{Hjort2014}, and \cite{Kato2016} show that peer effects matter only when individuals interact frequently or share similar origins or ethnic divisions, even if they all belong to the same group. 

Our method offers a benefit over earlier approaches for identifying peer and contextual effects in endogenous networks because it does not depend on a specific description of the network formation process, as seen in works such as \cite{Goldsmith-Pinkham2013}, \cite{Qu2015}, \cite{Johnsson2019}, \cite{Qu2021}, and \cite{Auerbach2022}. Apart from bypassing the estimation of the parameters of a potentially complex nonlinear outcome, the flexibility of our method permits us to address a broader spectrum of potential endogeneity sources. A notably significant benefit is that our method can manage endogeneity arising from the concurrent determination of the outcome and the network of interest. Current methods that require explicit modeling of network formation assume that networks depend on certain endogenous variables and that the outcomes are established once the endogenous network is formed. Extending those models to include a simultaneous system of equations that include both outcomes and dyadic linking choices requires one to fundamentally change the assumptions in the model, which can drastically affect the identification results. Nevertheless, since our method bypasses the direct modeling of link formation, it eliminates the need to explicitly account for the simultaneous determination of outcomes and the network, thereby making our straightforward approach highly effective in handling various sources of endogeneity.

The paper uses the exogenous network, $\mathbf{W}_{n,0}$, to construct a set of valid moment conditions that point-identify parameters in \eqref{intro_lmm}. The resulting linear Generalized Method of Moments (GMM) estimator is simple to implement in existing statistical software like \texttt{Python}, \texttt{R} or \texttt{Stata}; see, e.g., \cite{netivreg}. Furthermore, the estimator is shown to be asymptotically normal at the standard root--$n$ rate of convergence, and a consistent estimator of the efficient asymptotic variance-covariance is proposed to perform asymptotically valid inference. 

Our proof method allows us to characterize the asymptotic variance-covariance that takes into account the presence of network dependence between individuals generated by the endogenous network $\mathbf{W}_{n}$. In particular, we assume that the triangular array for the random vector of observed and unobserved characteristics is $\psi$ dependent \citep{Doukhan1999}, and the levels of dependence between individuals decrease with their distance in the network space spanned by $\mathbf{W}_{n}$. Moreover, our asymptotic normality result requires that the levels of dependence decrease at a rate that is fast enough to compensate for any potential increase in the asymptotic levels of network density. Finally, our large sample variance-covariance characterization also connects the asymptotic levels of sparsity in the exogenous network, $\mathbf{W}_{n,0}$, with the precision of the network effects parameters. In particular, we show that the precision of the parameters increases with the number of nodes for which we can find indirect paths of distance of at least two in $\mathbf{W}_{n,0}$. Certainly, higher levels of network density correspond to a reduced probability of identifying indirect links between the individuals in that network, thereby affecting the accuracy of our suggested estimator.

The structure of this paper is as follows. Section \ref{prelim} provides various definitions and notation used throughout. Section \ref{sid} introduces the model and conditions for the parameters of the basic model to be uniquely recovered (point identification) from the estimation sample $\left\{y_i,\mathbf{x}_i^{\top},\{w_{i,j}\}_{j=1,j\neq i}^n,\{w_{0;i,j}\}_{j=1,j\neq i}^n\right\}_{i=1}^n$. Section \ref{estimation} describes the proposed linear GMM estimator, its asymptotic distribution, and how to calculate valid asymptotic standard errors. Section \ref{mc} provides the Monte Carlo exercises showing the small sample properties of the estimator, while an empirical illustration of the proposed methodology is discussed in Section \ref{emp}. Finally, Section \ref{discussion} concludes. The supplemental materials contain all mathematical proofs of the main results, as well as further details on the real data illustration.

%% file: preliminaries.tex
We assume the full observability of two types of networks over the same set of nodes. One is the endogenous network of interest that can create social externalities; the other is the instrumental network induced by random assignment. For $N \in \mathbb{N}_{+} \equiv \{1, 2,\dots\}$, let $\mathcal{I}_{N}$ represent the set of agents in an arbitrarily large population, which we call \textit{nodes} in the network. We denote by $\mathcal{G}_{N}=(\mathcal{I}_{N}, E)$ the population network of interest and by $\mathcal{G}_{N,0}=(\mathcal{I}_{N}, E_{0})$ the instrumental population network, where $E$ and $E_{0}$ are the corresponding sets of links (Edges). Similarly, as in \cite{Graham2020}, the observed networks of size $n<N$, are denoted as $\mathcal{G}_{n}$ and $\mathcal{G}_{n,0}$, respectively, and are assumed to coincide with the subgraphs induced by $\mathcal{I}_{n}$ nodes sampled from their corresponding large population networks.

In the population, we represent the two networks $\mathcal{G}_{N}$ and $\mathcal{G}_{N,0}$ with their respective adjacency matrices; that is, $\mathbf{W}_{N}=[w_{N;i,j}]$ and $\mathbf{W}_{N,0}=[w_{N,0;i,j}]$, where $w_{N;i,j},w_{N,0;i,j} \in (0,1]$ are weights representing the importance of the $(i,j)$ connection in each of the networks and $w_{N;i,j}=0$ if $i$ and $j$ are not connected in $\mathcal{G}_{N}$. We do the same for $w_{N,0;i,j}$ in $\mathcal{G}_{N,0}$. Our framework permits both weighted as well as unweighted edges in networks. In the latter scenario, $w_{N;i,j}$ and $w_{N,0;i,j}$ equal to one whenever there exists a connection between individuals $i$ and $j$ and equal zero otherwise. The theory developed here works in both cases.

We also define the $1 \times N$ vectors $\mathbf{w}_{N,i}=[w_{N;i,1},\dots,w_{N;i,N}] \in [0,1]^{N}$ and $\mathbf{w}_{N,0;i}=[w_{N,0;i,1},\dots,w_{N,0;i,N}] \in [0,1]^{N}$ to be the $i$th row of the adjacency matrices $\mathbf{W}_{N}$ and $\mathbf{W}_{N,0}$, respectively. Similarly $\mathbf{w}_{N,0;i}^p$ represents the $i$th row of powers $p\ge 1$ of the adjacency matrix $\mathbf{W}_{N,0}^p$. The adjacency matrices of the observed sample, $\mathbf{W}_{n}$ and $\mathbf{W}_{n,0}$, are defined and formed accordingly.

Section \ref{estimation} uses the concept of $\psi-$dependence to bound the dependence among individuals as a function of their distance in the network space. Following the literature on graph theory, we use the shortest path length as our measure of distance; i.e., we denote $d_{n}(i,j)$ as the minimum path length connecting individuals $i$ and $j$ in the \textit{endogenous} network of interest $\mathcal{G}_{n}$ induced by the sample size $n$. Let $A$ and $B$ be any two sets of individuals of size $a, b \in \mathbb{N}_{+}$. We define the distance between sets as
$d_{n}(A,B)= \min_{i\in A}\min_{j\in B}d_{n}(i,j)$.

Based on the definition of set distance, we define the following group of node sets used in the asymptotic results presented in Section \ref{estimation}: \subasu \label{ds1} $\mathcal{P}^{+}_{n}(a, b , d)=\{(A, B): A, B \subset \mathcal{I}_{n},|A|=a,|B|=b\text{, and } d_{n}(A, B) \geq d\}$ containing groups of nodes at a distance of at least $d$ from each other; \subasu \label{ds2} $\mathcal{P}^{-}_{n}(a, b , d)=\{(A, B): A, B \subset \mathcal{I}_{n},|A|=a,|B|=b\text{, and } d_{n}(A, B) \leq d\}$ containing groups of nodes at a distance of at most $d$ from each other; and \subasu \label{ds3} $\mathcal{P}_{n}(a, b , d)=\{(A, B): A, B \subset \mathcal{I}_{n},|A|=a,|B|=b\text{, and } d_{n}(A, B) = d\}$ the set associated with groups of nodes at distance $d$ from each other, where $\mathcal{I}_{n}$ is the set of sampled individuals of size $n$. The associated set that contains all nodes at a certain distance from node $i$ is $\mathcal{P}_{n}^{+}(i,d)=\{j \in \mathcal{I}_{n}: d_{n}(i,j)\geq d\}$, $\mathcal{P}_{n}(i,d)=\{j \in \mathcal{I}_{n}: d_{n}=d\}$, and $\mathcal{P}_{n}^{-}(i,d)=\{j \in \mathcal{I}_{n}: d_{n}\leq d\}$.

For any random vector $\mathbf{r}_{N;i} \in \mathbb{R}^{L}$ for some $L \in \mathbb{N}_{+}$, we endow $\mathbb{R}^{L\times a}$ with the distance measure $\boldsymbol{d}_{a}(\mathbf{x}, \mathbf{y})=\sum_{l=1}^{a}\left\|x_{l}-y_{l}\right\|_{2}$, for $a\in \mathbb{N}_{+}$, and where $\|\cdot\|_{2}$ denotes the Euclidean norm and $(\mathbf{x}, \mathbf{y}) \in \mathbb{R}^{L\times a}$. We let $\mathscr{L}_{L, a}$ denote the collection of bounded Lipschitz real functions that map values from $\mathbb{R}^{L \times a}$  to $\mathbb{R}$. For any set of individuals $A$, let $\mathbf{r}_{N,A} = \left(\mathbf{r}_{N; i}\right)_{i \in A}$. In the following, we write triangular arrays simply as $\{\mathbf{r}_{N;i}\}$, and sequences such as $\{\lambda_n\}_{n\ge 1}$ as $\{\lambda_n\}$. As in \cite{Doukhan1999} and \cite{Kojevnikov2020}, we define $\psi-$dependence as follows.

\begin{definition}[$\psi$-dependence]
\label{depdef}
A triangular array $\{\mathbf{r}_{n;i}\}$, $n\geq 1$, $\mathbf{r}_{n;i} \in \mathbb{R}^{L}$ is $\psi$-dependent if for each ${n \in \mathbb{N}_{+}}$ there exist a sequence $\{\lambda_{n}\} \equiv\left\{\lambda_{n, d}\right\}_{d \geq 0}, \lambda_{n, 0}=1$ and a collection of non-random functions $\left(\psi_{a, b}\right)_{a, b \in \mathbb{N}}, \psi_{a, b}: \mathscr{L}_{v, a} \times \mathscr{L}_{v, b} \to [0, \infty)$, such that for all $A,B \in \mathcal{P}^{+}_{N}(a, b , d)$ for $d>0$ and all $f \in \mathscr{L}_{L, a}$ and $g \in \mathscr{L}_{L, b}$, 

$$\left|\operatorname{cov}\left(f\left(\mathbf{r}_{n,A}\right), g\left(\mathbf{r}_{n,B}\right)\right)\right| \leq \psi_{a, b}(f, g) \lambda_{n, d}.$$
\end{definition}

The sequence $\{\lambda_{n}\}$ is called the \emph{dependence coefficients} of $\mathbf{r}_{n;i}$. The covariance of the nonlinear functions of the random vectors $\mathbf{r}_{n,A}$ and $\mathbf{r}_{n,B}$ are bounded by the dependence coefficients $\lambda_{n, d}$ and a functional $\psi_{a, b}(f, g)$, which depends on the size of the sets $A$ and $B$, and the aggregating nonlinear functions $f$ and $g$. 

%% file: identification.tex
There is an arbitrarily large population of agents in the set $\mathcal{I}_{N}$. Each agent $i\in\mathcal{I}_{N}$ is characterized by a set of $K$ observable characteristics $\mathbf{x}_{N;i}$, and an unobserved idiosyncratic shock (error) $\varepsilon_{N;i}$. The agents in the population are connected by two types of networks, that is, $\mathcal{G}_{N}$ and $\mathcal{G}_{N,0}$. The data generation process characterizing the two observed networks adheres to a framework where the network $\mathcal{G}_{N}$ contains connections formed by two agents making endogenous decisions to participate in social or professional relationships, while $\mathcal{G}_{N,0}$ represents the network formed through the random (or quasirandom) allocation of individuals into groups. The endogenous network formation determining $\mathcal{G}_{N}$ induces a potential correlation between the individuals' decisions to connect and their observed and unobserved characteristics \citep[see, e.g., ][]{Hasan2019}. Conversely, by the properties of randomization, the network $\mathcal{G}_{N,0}$ is strictly exogenous \citep[see, e.g., ][]{Athey2017}.

Let $\mathscr{G}$ be the discrete set of possible network configurations of size $N$, such that, for any possible realizations of the network of interest and the exogenous network, $\mathbf{g}$ and $\mathbf{g}_{0}$, it follows that $\left\{\mathbf{g},\mathbf{g}_{0}\right\} \in \mathscr{G}$. Let $\mathbf{X}_{N}=[\mathbf{x}_{N;1},\cdots,\mathbf{x}_{N;N}]^\top\in\mathscr{X}$. Here, $\mathscr{X}$ is the space that represents the support for the matrix of regressors, where we assume $\mathscr{X}\subseteq \mathbb{R}^{N\times K}$, without loss of generality. Finally, let $\boldsymbol{\varepsilon}_{N}=[\varepsilon_{N;1},\cdots,\varepsilon_{N;N}]^\top\in\mathbb{R}^N$. We denote the realizations of $\mathbf{X}_{N}$ and $\boldsymbol{\varepsilon}_{N}$ by $\mathbf{X}\in\mathscr{X}$ and $\boldsymbol{\varepsilon}\in\mathbb{R}^N$, respectively.

Formally, we assume that there is a population joint probability distribution function that determines the dependence patterns between the regressors, the networks, and the errors, that is, the joint distribution is defined as $f_{\mathbf{X}_{N}, \mathcal{G}_{N},\mathcal{G}_{N,0},\boldsymbol{\varepsilon}_{N}}(\mathbf{X},\mathbf{g},\mathbf{g}_{0},\boldsymbol{\varepsilon})= \Pr(\mathcal{G}_{N}=\mathbf{g},\mathcal{G}_{N,0}=\mathbf{g}_{0}\mid\mathbf{X}, \boldsymbol{\varepsilon})f_{\mathbf{X}_{N},\boldsymbol{\varepsilon}_{N}}(\mathbf{X}, \boldsymbol{\varepsilon})$, where $\Pr(\mathcal{G}_{N}=\mathbf{g},\mathcal{G}_{N,0}=\mathbf{g}_{0}\mid\mathbf{X}_{N}, \boldsymbol{\varepsilon}_{N})$ is the probability that $\mathcal{G}_{N}$ and $\mathcal{G}_{N,0}$ take on particular structures $\left\{\mathbf{g},\mathbf{g}_{0}\right\} \in \mathscr{G}$ conditional on the regressors and the errors, while $f_{\mathbf{X}_{N},\boldsymbol{\varepsilon}_{N}}$ represents the joint probability distribution function for $\mathbf{X}_{N}$ and $\boldsymbol{\varepsilon}_{N}$. Note that the joint distribution, $f_{\mathbf{X}_{N}, \mathcal{G}_{N},\mathcal{G}_{N,0},\boldsymbol{\varepsilon}_{N}}(\mathbf{X},\mathbf{g},\mathbf{g}_{0},\boldsymbol{\varepsilon})$, is characterized by two main features: First, \textit{validity}, the network $\mathcal{G}_{N,0}$ is independent of the agents' observed and unobserved characteristics because of randomization. Second, \textit{relevance}, agents randomly assigned to the same group in $\mathcal{G}_{N,0}$ are more likely to form connections in $\mathcal{G}_{N}$ \citep[see, e.g., ][]{Granovetter1973, Gargiulo2000, Kim2006, Goldsmith-Pinkham2013}. 

The following assumption imposes the primary validity condition that we call ex-ante exogeneity.

\begin{assumption}[Ex-Ante Exogeneity]
\label{id}
Let $f_{\mathbf{X}_{N}, \mathcal{G}_{N,0},\boldsymbol{\varepsilon}_{N}}(\mathbf{X}, \mathbf{g}_{0}, \boldsymbol{\varepsilon})\equiv\sum_{\mathbf{g}\in\mathscr{G}}f_{\mathbf{X}_{N}, \mathcal{G}_{N},\mathcal{G}_{N,0},\boldsymbol{\varepsilon}_{N}}\allowbreak(\mathbf{X},\mathbf{g},\mathbf{g}_{0},\boldsymbol{\varepsilon})$ be the resulting probability distribution from integrating $f_{\mathbf{X}_{N}, \mathcal{G}_{N},\mathcal{G}_{N,0},\boldsymbol{\varepsilon}_{N}}$ out with respect to the endogenous network $\mathcal{G}_{N}$. Thus, the probability distribution $f_{\mathbf{X}_{N}, \mathcal{G}_{N,0},\boldsymbol{\varepsilon}_{N}}$ is such that $\mathbb{E}[\mathbf{x}_{N;i}\varepsilon_{N;i}]=\boldsymbol{0}_K$ and $\mathbb{E}[(\mathbf{w}_{N,0;i}^{p}\mathbf{X}_{N})^{\top}\varepsilon_{N;i}]=\boldsymbol{0}_K,$  $\forall i\in\mathcal{I}_{N}$, and any $p\geq 1$, where $\boldsymbol{0}_K$ is a $K\times 1$ vector of zeros. Moreover, we assume $\mathbb{E}[\varepsilon_{N;i}]=0$, $\forall i\in\mathcal{I}_{N}$, where the expectation is taken with respect to the marginal distribution of $\boldsymbol{\varepsilon}_{N}$. 
\end{assumption}

This assumption imposes the restriction of zero correlation between the observed and unobserved characteristics of any individual $i\in\mathcal{I}_{N}$. For simplicity, we impose the noncorrelation assumption across the entire vector of observed characteristics $\mathbf{x}_{N;i}$; however, our results also hold under the non-correlation premise for a minimum of one variable within $\mathbf{x}_{N;i}$.\footnote{The ex-ante exogeneity assumption can also accommodate correlation between the network $G_{N,0}$ and individuals' characteristics. If the network $\mathcal{G}_{N, 0}$ is determined by \emph{observed} characteristics, one can control for them in the outcome equation and the ex-ante exogeneity assumption becomes a \textit{conditional} ex-ante endogeneity assumption.} Furthermore, the constraint $\mathbb{E}[(\mathbf{w}_{N,0;i}^{p}\mathbf{X}_{N})^{\top}\varepsilon_{N;i}]=\boldsymbol{0}_K$ implies that the unobserved characteristics of the individual $i\in\mathcal{I}_{N}$ remain uncorrelated with the observed characteristics of any other $j\in\mathcal{I}_{N}$, selected exogenously from the population of individuals. Analogous formulations of this assumption, employing randomized networks, have previously been used to support the causal identification of peer effects \citep[see, e.g.,][]{Sacerdote_QJE, Carrell2009, Cai2018, Hasan2019}.

Note that the Assumption \ref{id} does not impose restrictions on the correlation between the vector $(\mathbf{w}_{N;i}\mathbf{X}_{N})$ and the unobserved characteristics $\varepsilon_{N;i}$. Thus, we allow for correlation between the observed and unobserved characteristics of any two individuals engaging in endogenous linking formation within the network $\mathcal{G}_{N}$. The assumption \ref{id} links the endogenous network formation process with the mechanism that induces correlation between the vectors of regressors and errors. 

Under ex-ante exogeneity, a pair of agents $(i,j)$, exogenously allocated in an initial network $\mathcal{G}_{N,0}$, establish connections in the network $\mathcal{G}_{N}$ based on observed and unobserved characteristics. This behavior creates correlation between $\mathbf{x}_{N,i}$ and $\varepsilon_{N, j}$ through observed and unobserved homophily. For example, in the context of \cite{Carrell2013}, students with similar abilities may be more inclined to form connections, and observable characteristics such as race or gender could correlate with both link formation and individuals' abilities. Moreover, as previously mentioned, the source of endogeneity can arise not only from unobserved factors, but also because of an unspecified potential simultaneous determination of the outcome and the network of interest. In the same context of educational achievement, for example, students may care about the educational achievement of potential connections, which can generate an issue of simultaneous determination in the outcome equation.

In addition to ex-ante exogeneity, the following assumption imposes a linear model of peer effects, which has been shown to have a structural interpretation as the best response function of a Bayesian game of social interactions; see, i.e., \cite{Blume2015}. 

\begin{assumption}[Linear Model]
\label{excrest}
The optimal choice (outcome), $y_{N;i}$, for agent $i$ is characterized by

\begin{equation}
\label{lmm}
y_{N;i}=\beta_{0}\sum_{j \neq i}w_{N;i,j}y_{N;j}+\sum_{j \neq i}w_{N;i,j}\mathbf{x}_{N;j}^{\top}\boldsymbol{\delta}_{0}+ \widetilde{\mathbf{x}}_{N;i}^{\top}\boldsymbol{\gamma}_{0}+\varepsilon_{N;i}\text{,}
\end{equation}

\noindent where $\widetilde{\mathbf{x}}_{N;i}=[1,\mathbf{x}_{N;i}^\top]^\top$, $w_{N;i,j}$ is the $ij$\emph{th} position in the adjacency matrix representing the endogenous network, $\mathbf{W}_{N}$, and $\boldsymbol{\theta}_0\equiv(\beta_{0}, \boldsymbol{\delta}_{0}^\top, \boldsymbol{\gamma}_{0}^\top)^\top$ belongs to the interior of the parameter space $\boldsymbol{\Theta}\subset\mathbb{R}^{2K+2}$, which is assumed to be compact.\footnote{Our Linear Model assumption includes both the case when the adjacency matrices are row normalized and when they are not. Our results hold for both the local aggregate and the local average models \citep{liu2014}.}
\end{assumption}

Assumption \ref{excrest} prescribes a linear model for social effects, effectively imposing an exclusion restriction on the network $\mathcal{G}_{N,0}$. The model posits that only optimally formed connections by agents can induce peer effects. Specifically, we contend that randomly grouped individuals are unlikely to generate peer effects, but, after randomization, agents can form endogenous connections that influence their behavior; see, for example, \cite{Hasan2019}. Similarly, \cite{Carrell2013} shows that groups designed to improve academic performance can produce negative effects due to the role of endogenous study partners and the formation of friendship bonds after the initial allocation of optimally designed to improve academic performance. Furthermore, \cite{Hjort2014} and \cite{Kato2016} find that the effects of peers on worker productivity manifest only between individuals of the same ethnic divisions and social origins, which could serve as proxies for closer social interactions. This evidence supports the plausibility of the exclusion restriction assumption of the initial exogenous network in a peer effects model.

Here, we implicitly impose a network \emph{exclusion} restriction in the sense that potential peer effects from the exogenous network $\mathcal{G}_{N,0}$ are precisely zero in \eqref{lmm}. Alternatively, it is possible to relax the exact exclusion restriction in the exogenous network by incorporating prior information where the effect of $\mathcal{G}_{N,0}$ on $y_{N;i}$ is proximate but not precisely zero, following the approach of \cite{Conley2012}.\footnote{We thank the Associate Editor and one of the referees for pointing this out.}

Finally, the assumption that the parameters $\boldsymbol{\theta}_0$ are in the interior of the parameter space is particularly relevant for the coefficient $\beta_{0}$, because \eqref{lmm} has a solution in terms of $w_{N;i}$, $\mathbf{X}_{N}$, and $\varepsilon_{N;i}$ only when $\beta_{0}<1/\lambda_{\max}$, where $\lambda_{\max}$ is the largest eigenvalue of $\mathbf{W}_{N}$. Assuming $K=1$ and that there is no constant for the sake of illustration, Assumption \ref{excrest} implies that the peer effects regressor can be written as 

\begin{equation}
\mathbf{W}_{N} \mathbf{y}_{N}= \gamma_{0} \mathbf{W}_{N} \mathbf{x}_{N}+(\gamma_{0} \beta_{0}+\delta_{0}) \sum_{p=0}^{\infty} \beta_{0}^{p} \mathbf{W}_{N}^{p+2} \mathbf{x}_{N} +\sum_{p=0}^{\infty} \beta_{0}^{p}\mathbf{W}_{N}^{p+1}\varepsilon_{N},
\end{equation}

\noindent which under the condition that $\gamma_{0} \beta_{0}+\delta_{0}\neq0$ shows that, in principle, the powers of the adjacency matrix $\mathbf{W}_{N}$ could be used to instrument $\mathbf{W}_{N} \mathbf{y}_{N}$ \citep{Bramoulle2009, Degiorgi2010, manta2021}. This approach is not possible here because the network $\mathcal{G}_{N}$ is allowed to be endogenous. However, note that from Assumptions \ref{id}, the powers of the adjacency matrix $\mathbf{W}_{N,0}$ are natural candidates to replace $\mathbf{W}_{N}$ in this approach. 

We propose to use the random assignment embodied in $\mathbf{W}_{N,0}$ to identify the parameters of a linear model defined on the network space spanned by $\mathbf{W}_{N}$ in Assumption \ref{excrest}. Formally, we define $\mathbf{D}_{N}=[\mathbf{W}_{N}\mathbf{y}_{N}, \mathbf{W}_{N}\mathbf{X}_{N}, \widetilde{\mathbf{X}}_{N}]$ as the matrix of regressors in the matrix notation counterpart of \eqref{lmm} and $\mathbf{Z}_{N}=[\mathbf{W}_{N,0}^{p}\mathbf{X}_{N},\mathbf{W}_{N,0}^{p-1}\mathbf{X}_{N},\dots,\mathbf{W}_{N,0}\mathbf{X}_{N}, \widetilde{\mathbf{X}}_{N}]$ as the matrix that produces the moment conditions formed based on Assumption \ref{id}, where $p>1$ is a constant parameter representing the powers of the adjacency matrix used as instruments. This framework allows for the option of using the characteristics of the so-called \textit{connections of connections}' as instruments by letting $p=2$. The flexibility of Assumption \ref{id} also allows the use of the characteristics of more indirect connections that are at distance $p>2$. 

An important aspect differentiating the use of ex-ante exogeneity in Assumption \ref{id} and the standard validity of an instrumental variable is that ex-ante exogeneity allows us to form a large number of instruments. In particular, as long as $\mathbf{I}_{N}$, $\mathbf{W}_{N,0}$,$\mathbf{W}_{N,0}^{2}$,$\dots$, $\mathbf{W}_{N,0}^{p-1}$, and $\mathbf{W}_{N,0}^{p}$ are linearly independent, where $\mathbf{I}_{N}$ is the identity matrix of order $N$, we can form up to $K\cdot p$ different instruments using the $K$ ex-ante exogenous variables in $\mathbf{x}_{N,i}$. Note that as discussed in Section \ref{prelim}, the networks underlying the adjacency matrices $\mathbf{W}_{N}$ and $\mathbf{W}_{N,0}$ in $\mathbf{D}_{N}$ and $\mathbf{Z}_{N}$ can be weighted or unweighted. Considering that both adjacency matrices undergo post-multiplication by the matrix $\mathbf{X}$, we can think of the weighted and unweighted scenarios as computing averages and sums of the attributes of directly and indirectly connected individuals, respectively. As demonstrated by \cite{liu2014}, different weighting methodologies correspond to different behavioral mechanisms in social interactions, resulting in either local average or local aggregate effects. Although behavioral interpretations may differ, the identification argument here holds for both interpretations of peer effects. Another consideration when choosing between weighted or unweighted networks is that achieving the linear independence condition of the powers of adjacency matrices may be more feasible in the former case, see, i.e., \cite{Blume2015}.

The use of the two networks $\mathcal{G}_{N,0}$ and $\mathcal{G}_{N}$ for identification requires a \textit{relevance} condition that guarantees the two networks have enough overlap. The moment characterizing the correlation between the regressors and the instruments is the population average of the expected values of the random matrices $(\mathbf{z}_{N;i}\mathbf{d}_{N;i}^{\top})$, which can be written as $\mathbb{E}[N^{-1}\sum_{i \in \mathcal{I}_{N}}\mathbf{z}_{N;i}\mathbf{d}_{N;i}^{\top}]$, where $\mathbf{d}_{N;i}$ and $\mathbf{z}_{N;i}$ contain the $i$th rows of the matrices $\mathbf{D}_N$ and $\mathbf{Z}_N$,  respectively. Importantly, we do not assume that $\mathbb{E}[\mathbf{z}_{N;i}\mathbf{d}_{N;i}^{\top}]$ are equal for all $i$ and work directly with the population average of these expectations. The following assumption imposes a rank condition related to the strength of the correlation between $\mathbf{z}_{N;i}$ and $\mathbf{d}_{N;i}$. 

\begin{assumption}[Relevance]
\label{relevance}
The matrix $\mathbb{E}[N^{-1}\sum_{i \in \mathcal{I}_{N}}\mathbf{z}_{N;i}\mathbf{d}_{N;i}^{\top}]<\infty$ has full column rank. 
\end{assumption}

Assumption \ref{relevance} guarantees that the population average of expectations $\mathbb{E}[\mathbf{z}_{N;i}\mathbf{d}_{N;i}^{\top}]$ is finite, and it provides the necessary conditions to ensure that a unique parameter value solves the identifying moment conditions. Unlike standard IV estimation, the relevance condition here imposes restrictions on the more primitive network architectures. \ref{Appendix_D} provides such a set of primitive conditions. They are empirically testable restrictions on the regressors and networks. Importantly, unlike previous results, identification in our model does not require linear independence between the matrices $\mathbf{I}_{N}$, $\mathbf{W}_{N}$, and $\mathbf{W}_{N}^{2}$. Instead, we show that for Assumption \ref{relevance} to hold, there should exist two \emph{ different} numbers $(r,s)\in\mathbb{N}_{+}\times\mathbb{N}_{+}$ such that $\mathbf{I}_{N}$, $\mathbf{W}_{N}^{r}$, and $\mathbf{W}_{N}^{s}$ are linearly independent. Thus, our primitive conditions for the network of interest are weaker in the sense that any two powers of the adjacency matrix have to be linearly independent, and not only the first and second powers. 

Furthermore, we show that, as mentioned above, relevance requires that the matrices $\mathbf{I}_{N}$, $\mathbf{W}_{N,0}$,$\mathbf{W}_{N,0}^{2}$,$\dots$,  $\mathbf{W}_{N,0}^{p-1}$, and $\mathbf{W}_{N,0}^{p}$ be linearly independent for some $p\in\mathbb{N}_{+}$ and $(\gamma_{0,k} \beta_{0}+\delta_{0, k})\neq 0$ for any $k\in \{1, \dots, K\}$. The first condition resembles the restrictions on the network structure imposed in previous literature, with the difference that in our model the linear independence condition is imposed on the excluded network. The second condition on the structural parameters of the linear model has also been used in the previous literature, see, i.e., \cite{Bramoulle2009}. It excludes the possibility that peer and contextual effects cancel each other out. The following Theorem formalizes the identification result.

\begin{theorem}[Identification]
\label{idtheorem}
Let Assumptions \ref{id}, \ref{excrest}, and \ref{relevance} hold, then $\mathbb{E}[\boldsymbol{m}_{N}(\boldsymbol{\theta})]=\boldsymbol{0}_{K}$ if and only if $\boldsymbol{\theta}=\boldsymbol{\theta}_{0}$, where $\boldsymbol{m}_{N}(\boldsymbol{\theta})\equiv N^{-1}\sum_{i\in\mathcal{I}_{N}}\mathbf{z}_{N;i}(y_{N;i}-\mathbf{d}_{N;i}^{\top}\boldsymbol{\theta})$.
\end{theorem}

Appendix \ref{Appendix_A} presents the proof of Theorem \ref{idtheorem}. This Theorem shows that identification is possible in a context where the network of interest $\mathcal{G}_{N}$ is formed endogenously by taking advantage of the randomization and exclusion restrictions on the exogenously imposed network $\mathcal{G}_{N,0}$. This approach allows us to attach a causal interpretation to the estimated parameters of a linear model of peer effects, which uses observational network data that emerge after an initial randomization. This method can be used to address research designs with randomized peers, as in \cite{Carrell2013}.

%% file: estimation.tex
We propose a GMM estimator based on the identifying moment condition in Theorem \ref{idtheorem}. We assume that the analyst observes a sample of size $n<N$ from the population described in the previous section. In our sample scheme, $n$ agents are chosen at random without replacement and their observed characteristics, outcome, and connections in $\mathcal{G}_{N}$ and $\mathcal{G}_{N,0}$ are recorded.\footnote{In our empirical illustration we observe the complete graph for the schools' finite populations of students. However, as in \cite{Graham2020}, we consider the sampling process as a thought experiment that is useful in characterizing limiting distributions.} Therefore, the random sample consists of observations $\left\{y_i,\mathbf{x}_i^{\top},\{w_{i,j}\}_{j=1,j\neq i}^n,\{w_{0;i,j}\}_{j=1,j\neq i}^n\right\}_{i=1}^n$ from which it is possible to calculate the $n\times (2K+2)$ matrix of regressors $\mathbf{D}_{n}$ and the $n \times ((p-1)K+2K+1)$ matrix of \textit{instruments} $\mathbf{Z}_{n}$ (depending on the value of $K$, the system can be \emph{just}- or \emph{over}-identified). The population's GMM objective function is given by $J_{N}(\boldsymbol{\theta})=\mathbb{E}[\boldsymbol{m}_{N}(\boldsymbol{\theta})]^{\top}\mathbf{A}_{N}\mathbb{E}[\boldsymbol{m}_{N}(\boldsymbol{\theta})]$, where $\mathbf{A}_{N}$ is a constant full rank weighting matrix. The GMM estimator of $\boldsymbol{\theta}_0$ is defined as $\widehat{\boldsymbol{\theta}}_{\text{GMM}}=\arg\min_{\boldsymbol{\theta}\in\boldsymbol{\Theta}}J_{n}(\boldsymbol{\theta})$, where $J_{n}(\boldsymbol{\theta})\equiv[n^{-1}\sum_{i\in \mathcal{I}_{n}}\mathbf{z}_{n;i}(y_{n;i}-\mathbf{d}_{n;i}^\top\boldsymbol{\theta})]^{\top}\mathbf{A}_{n}[n^{-1}\sum_{i\in \mathcal{I}_{n}}\mathbf{z}_{n;i}(y_{n;i}-\mathbf{d}_{n;i}^\top\boldsymbol{\theta})]$, the $((p-1)K+2K+1)\times ((p-1)K+2K+1)$ full rank weighting matrix $\mathbf{A}_{n}$ is assumed to converge in probability to $\mathbf{A}_{N}$. The linearity in \eqref{lmm} guarantees that the GMM estimator has a closed form solution given by

\begin{equation}
    \widehat{\boldsymbol{\theta}}_{\text{GMM}} =  [\mathbf{D}_{n}^{\top}\mathbf{Z}_{n}\mathbf{A}_{n} \mathbf{Z}_{n}^{\top}\mathbf{D}_{n}]^{-1} [\mathbf{D}_{n}^{\top}\mathbf{Z}_{n}\mathbf{A}_{n}\mathbf{Z}_{n}^{\top}\mathbf{y}_{n}]\text{.}\label{theta_GMM}
\end{equation}

To allow the possibility that the observed and unobserved characteristics of individuals are correlated in the joint distribution of the population $f_{\mathbf{X}_{N}, \mathcal{G}_{N},\mathcal{G}_{N,0},\boldsymbol{\varepsilon}_{N}}$, we use the concept of $\psi-$dependence in definition \ref{depdef} above. As mentioned there, we bound the correlation between nonlinear functions of random variables with the \textit{dependence coefficients}, which are decreasing functions of the network distance. We rule out a direct dependence structure based on the exogenous network $\mathcal{G}_{N,0}$, i.e., dependence is only generated through the endogenous network $\mathcal{G}_{N}$. Intuitively, this is justified here because individuals endogenously form connections in $\mathcal{G}_{N}$ based on observed and unobserved characteristics, and therefore we would expect relatively high levels of dependence between individuals close to each other in the network space spanned by $\mathcal{G}_{N}$. For example, in our empirical illustration, we expect the observed and unobserved characteristics of students who study together or are indirectly connected by study partners to be more correlated than those of students who are not study partners and are not indirectly connected. Given that the network $\mathcal{G}_{N}$ generates the dependence structure, when we talk about the distance in the network space, we refer to the distances in the network space $\mathcal{G}_{N}$ hereafter. 

Let $\mathbf{r}_{N ; i} \equiv\left[\mathbf{x}_{N ; i}^{\top}, \varepsilon_{N ; i}\right]^{\top} \in \mathbb{R}^{K+1}$ be the vector that encompasses the observed and unobserved characteristics of the individual $i$. By choosing appropriate values for the functions $f$ and $g$ in Definition \ref{depdef}, the $\psi-$dependence framework allows us to bound the dependence between observed and unobserved characteristics among any set of individuals. Specifically, we impose the following assumption on the conditional population distribution $f_{\mathbf{X}_N, \varepsilon_N \mid \mathcal{G}_N}$.

\begin{assumption}[Weak Dependence]
\label{depas}
Consider the set $\mathscr{G}$ of all possible realizations of $\mathcal{G}_{N}$. $\forall\mathcal{G}_{N} \in \mathscr{G}$, and $\mathcal{N}$ denoting either $N\in \mathbb{N}_{+}$ or $n\in \mathbb{N}_{+}$, assume that the conditional distribution $f_{\mathbf{X}_{N}, \boldsymbol{\varepsilon}_{N}\mid \mathcal{G}_{N}}$ is such that:\medskip

\noindent \subasu \label{ea1} $\{\mathbf{r}_{\mathcal{N};i}\}$ is $\psi$-dependent with dependence coefficient $\lambda_{\mathcal{N}}$; \smallskip

\noindent \subasu \label{ea2} For a generic constant $C>0$, $\psi_{a, b}(f, g) \leq C \times a b\left(\|f\|_{\infty}+\operatorname{Lip}(f)\right)\left(\|g\|_{\infty}+\operatorname{Lip}(g)\right)$; \smallskip

\noindent \subasu \label{ea3} and for each ${\mathcal{N} \in \mathbb{N}_{+}}$, $\max_{d \geq 1} \lambda_{\mathcal{N}, d}<\infty$ and $\lim_{d\to\infty}\lambda_{\mathcal{N}, d}=0$.
\end{assumption}

We impose the Assumption \ref{depas} for both the population with conditional distribution $f_{\mathbf{X}_{N}, \boldsymbol{\varepsilon}_{N}\mid \mathcal{G}_{N}}$ and for the triangular array ${\mathbf{r}_{n;i}}$, where $n\geq 1$. This array is formed by randomly sampling networks $\mathcal{G}_{N}$ and the observed and unobserved characteristics that define ${\mathbf{r}_{n;i}}$. Condition \ref{ea2} bounds the functional $\psi_{a, b}(f, g)$ by an arbitrary constant $C$, the cardinality of the sets $A$ and $B$,  and the sup-norm and Lipschitz constants of the aggregating functions $f$ and $g$. Intuitively, if the Lipschitz constants $\operatorname{Lip}(f)$ and $\operatorname{Lip}(g)$ increase, the values of the functions $f$ and $g$ can be higher for some values of $\mathbf{r}_{n,A}$ and $\mathbf{r}_{n,B}$, which requires larger constants to bound the covariance. This intuition is similar to the sup-norm. Finally, condition \ref{ea3} requires that the dependence coefficients be finite for any value of $d$ and that they dissipate to zero for a sufficiently large network distance between the random vectors $\mathbf{r}_{n,A}$ and $\mathbf{r}_{n,B}$. 

The use of the $\psi-$dependence framework in modeling network dependence has the advantage that it does not impose functional form restrictions on the errors and it allows for correlation between indirectly connected nodes. However, transformations of $\psi$-dependent random variables are not necessarily $\psi$-dependent. Therefore, in order to analyze the asymptotic behavior of $\widehat{\boldsymbol{\theta}}_{\text{GMM}}$, we now impose bounds to covariances of the form $\operatorname{cov}(r_{n; i, q} r_{n; j, \ell}, r_{n; h, q^{\prime}} r_{n; s, \ell^{\prime}})$, where $(i, j, h, s) \in \mathcal{I}_{n}$, $q, q^{\prime}, \ell$, and $\ell^{\prime}$ are components of the vector $\mathbf{r}_{n; i}$. These include covariances such as $\operatorname{cov}(\varepsilon_{n; i} \varepsilon_{n; j}, \varepsilon_{n; h} \varepsilon_{n; s})$ or $\operatorname{cov}(x_{n; i, q}x_{n; j, \ell}, \varepsilon_{n; h} \varepsilon_{n; s})$, for example.

\begin{assumption}[Bound Covariances]
\label{moments}
Define the functions $f_{q,\ell}$ and $g_{q^{\prime},\ell^{\prime}}$ mapping $\mathbb{R}^{(K+1)\times 2}$ to $\mathbb{R}$ to be such that $f_{q,\ell}(\mathbf{r}_{n;\{i,j\}}) = r_{n;i,q}r_{n;j,\ell}$ and $g_{q^{\prime},\ell^{\prime}}(\mathbf{r}_{n;\{h,s\}}) = r_{n;h,{q^{\prime}}}r_{n;s,{\ell^{\prime}}}$ for $(i,j,h,s) \in \mathcal{I}_{n}$, $i\neq j$, $h\neq s$, $q\neq \ell$ and $q^{\prime}\neq \ell^{\prime}$. The norms $\|f_{q,\ell}(\mathbf{r}_{n;\{i,j\}})\|_{p_{f}^{\ast}}+\|g_{q^{\prime},\ell^{\prime}}(\mathbf{r}_{n;\{h,s\}})\|_{p_{g}^{\ast}}<\infty$ for all $q,\ell$ where $p_{f}^{\ast}=\max\{p_{f,i},p_{f,j}\}$ (analogous for $p_{g}^{\ast}$) and $1/p_{f,i}+1/p_{f,j}+1/p_{g,h}+1/p_{g,s}<1$.
\end{assumption}

Assumption \ref{moments} provides sufficient conditions for the functions $f_{q,\ell}$ and $g_{q^{\prime},\ell^{\prime}}$ of $\psi$-dependent random variables to have bounded covariances. The weak dependence in Assumption \ref{depas} guarantees that the dependence coefficients vanish to zero when the network distance increases. However, the network distance $d_{n}(i,j)$ between any two individuals $i$ and $j$ is also a function of the sample size. Therefore, the asymptotic behavior of the dependence coefficients $\lambda_{n, d}$ depends on the asymptotic behavior of the network features determining the distance between nodes. In particular, the density of the network is explicitly related to the geodesic distance. When the density of the network is arbitrarily large, the geodesic distance is always one for any pair of nodes. Therefore, as noted by \cite{Kojevnikov2020}, there is a trade-off between network density and the rate of convergence of the dependence coefficients. Networks with higher density would require the dependence to decrease faster (and vice versa). The following assumption provides a necessary condition on the dependence coefficients for a Law of Large Numbers to apply. 

\begin{assumption}[Dependence Rate of Decay]
\label{decay}

Let $\bar{D}_{n}(d)\equiv n^{-1}\sum_{i\in \mathcal{I}_{n}}|\mathcal{P}_{n}(i,d)|$ be the average number of distance-$d$ connections in the network $\mathcal{G}_{n}$. We assume that, for any realizations of the networks $\mathcal{G}_{n}$, for all $n$, it follows that $n^{-1}\sum_{d\geq 1}\bar{D}_{n}(d)\allowbreak \lambda_{n,d}{\longrightarrow}0$ as $n\longrightarrow\infty$.
\end{assumption}

Assumption \ref{decay} is similar to Assumption 3.2 in \cite{Kojevnikov2020}, using the notation in our paper. The key distinction lies in our use of the unconditional version of $\psi$-dependence, which results in the dependence coefficients in the sequence ${\lambda_{n,d}}$ not being random variables. Furthermore, we apply Assumption \ref{decay} conditionally to any realizations of the networks $\mathcal{G}_{n}$, ensuring that $n^{-1}\sum_{d\geq 1}\bar{D}_{n}(d)\allowbreak \lambda_{n,d}$ is nonrandom. As emphasized by \cite{Kojevnikov2020}, this assumption is implied by restrictions on dependence coefficients ($\lambda_{n, d} \leq \theta_{n, d}^{1-4 / p}$, for $p>4$) and the number of distance-$d$ connections on $\mathcal{G}_{n}$ (${D}_{n}(d)\leq 4c_{n}(d, m , k)$, where $c_{n}(d, m , k)$ is defined in \eqref{dense}). The following assumption imposes the existence of moments for products of $\psi$-dependent random variables. 

\begin{assumption}[Existence of Moments]
\label{epmoments}
$\exists\epsilon>0$ such that $\sup _{n \geq 1} \max _{i \in \mathcal{I}_{n}}\|R_{n; i, j}\|_{1+\epsilon}<\infty$, where $R_{n;i,j}\equiv r_{n;i,q}r_{n;j,\ell}$, and $\left\|R_{n; i,j}\right\|_{p}\equiv (\mathbb{E}[|R_{n;i,j}|^{p} ])^{1 / p}$.
\end{assumption}

The previous assumptions are sufficient to guarantee that a Law of Large Numbers applies to products of $\psi$-dependent random variables. To show asymptotic normality, we again use the Central Limit Theorem result in \cite{Kojevnikov2020}. As mentioned above, for the asymptotic moments of network-dependent random variables to be well defined, we need to control the level of asymptotic density. In particular, following \cite{Kojevnikov2020}, we define a measure of the average neighborhood size as $\bar{D}_{n}(d , k)=n^{-1} \sum_{i \in \mathcal{I}_{n}}\left|\mathcal{P}_{n}(i , d)\right|^{k}$ and a measure of the average neighborhood shell size as $\bar{D}_{n}(d, m , k)^{-}=n^{-1} \sum_{i \in \mathcal{I}_{n}} \max _{j \in \mathcal{P}_{n}(i,d)}\left|\mathcal{P}_{n}^{-}(i,m) \setminus \mathcal{P}_{n}^{-}(j , d-1)\right|^{k}$, where $\mathcal{P}_{n}^{-}(j , d-1)=\left\{\emptyset\right\}$ when $d=0$. With these two measures of average density, construct the combined quantity, 

\begin{equation}
\label{dense}
c_{n}(d, m , k)=\inf _{\alpha>1}\left[\bar{D}_{n}(d, m , k \alpha)^{-}\right]^{\frac{1}{\alpha}}\left[\bar{D}_{n}\left(d, \frac{\alpha}{\alpha-1}\right)\right]^{1-\frac{1}{\alpha}}.
\end{equation}

For some arbitrary position $q$ in the matrix $\mathbf{Z}_{n;i}$, let $S_{n}=\sum_{i\in\mathcal{I}_{n}}z_{n;i,q}\varepsilon_{n;i}$. Defining $\sigma_{n,q}^{2}\equiv\text{var}(S_{n})$, the following assumption guarantees the existence of higher-order moments, imposes asymptotic sparsity, and bounds the long-run variance.

\begin{assumption}[Average Sparsity]
\label{av_sparsity}
For all network realizations $\mathcal{G}_{n} \in \mathscr{G}$, \subasu \label{clt1} for some $p>4$, $\sup _{n \geq 1} \max _{i \in \mathcal{I}_{n}}\left\|z_{n;i,q}\varepsilon_{n;i}\right\|_{p}<\infty$. There exists a sequence $m_{n}\to\infty$, such that for $k=1,2$, \subasu \label{clt2} , $\frac{n}{\sigma_{n,q}^{2+k}}\sum_{d \geq 0} c_{n}\left(d, m_{n} , k\right) \lambda_{n, d}^{1-\frac{2+k}{p}} {\longrightarrow} 0$ as $n\to \infty$, \subasu \label{clt3} $\frac{n^{2} \lambda_{n, m_{n}}^{1-(1 / p)}}{\sigma_{n,q}} {\longrightarrow} 0$ as $n\to \infty$.
\end{assumption}

These conditions impose a convergence rate of the dependence coefficients $\lambda_{n, d}$ that is related to the density of the network. There is a trade-off in which a higher density requires a higher speed in dependence-decreasing patterns. The previous assumptions are sufficient to show that our GMM estimator is consistent and asymptotically normal. Formally, let $\mathbf{\Omega}_{n}=\text{var}(\mathbf{Z}_{n}^{\top}\boldsymbol{\varepsilon}_{n})$ be a variance defined over the set of sampled individuals. It converges (see Lemma \ref{finitevar} in the supplemental material) to the finite population variance, 

\begin{equation}
\label{omega}
\mathbf{\Omega}_{N} = \lim_{n \to\infty} n^{-1}\left[\sum_{i=1}^{n} \text{var}(\mathbf{z}_{n;i}\varepsilon_{n;i}) + \sum_{i \neq j}\text{cov}(\mathbf{z}_{n;i}\varepsilon_{n;i}, \mathbf{z}_{n;j}\varepsilon_{n;j})\right]\equiv N^{-1}\sum_{d\geq 0} \mathbf{\Gamma}_{N}(d)<\infty,
\end{equation}

\noindent where $\mathbf{\Gamma}_{N}(d)=\sum_{i \in \mathcal{I}_{N}}\sum_{j \in \mathcal{P}_{N}(i,d)} \mathbb{E}[\mathbf{z}_{N;i}\varepsilon_{N;i}\varepsilon_{N;j}\mathbf{z}_{N;j}^{\top}]$ are the covariances between random variables of individuals at distance $d$. Therefore, the variance-covariance matrix $\mathbf{\Omega}_{N}$ can be calculated by summing the covariances for all possible distances $d\geq 0$. After characterizing $\mathbf{\Omega}_{N}$, Theorem \ref{t2} provides the asymptotic behavior of \eqref{theta_GMM}.

\begin{theorem}
\label{t2}
Let Assumptions \ref{id}--\ref{av_sparsity} hold, then as $n\to\infty$, $\widehat{\boldsymbol{\theta}}_{\text{\emph{GMM}}} = \boldsymbol{\theta}+o_{p}(1)$ and $\sqrt{n}(\widehat{\boldsymbol{\theta}}_{\text{\emph{GMM}}} - \boldsymbol{\theta}) \overset{d}{\to} \mathcal{N}(\boldsymbol{0}, \mathbf{\Sigma}_{N})$, where $\mathbf{\Sigma}_N\equiv(\mathbb{E}[N^{-1}\sum_{i \in \mathcal{I}_{N}}\mathbf{z}_{N;i}\mathbf{d}_{N;i}^{\top}]^{\top} \mathbf{A}_{N}\mathbb{E}[N^{-1}\sum_{i \in \mathcal{I}_{N}}\mathbf{z}_{N;i}\mathbf{d}_{N;i}^{\top}])^{-1}\times(\mathbb{E}[N^{-1}\sum_{i \in \mathcal{I}_{N}}\mathbf{z}_{N;i}\mathbf{d}_{N;i}^{\top}]^{\top}\allowbreak\mathbf{A}_{N}\mathbf{\Omega}_N\mathbf{A}_{N}\times\mathbb{E}[N^{-1}\sum_{i \in \mathcal{I}_{N}}\mathbf{z}_{N;i}\mathbf{d}_{N;i}^{\top}])(\mathbb{E}[N^{-1}\sum_{i \in \mathcal{I}_{N}}\mathbf{z}_{N;i}\mathbf{d}_{N;i}^{\top}]^{\top}
\mathbf{A}_{N}$

\noindent $\mathbb{E}[N^{-1}\sum_{i \in \mathcal{I}_{N}}\mathbf{z}_{N;i}\mathbf{d}_{N;i}^{\top}])^{-1}$, and when $\mathbf{A}_{N}=\mathbf{\Omega}_N^{-1}$, then 
\begin{equation}
    \mathbf{\Sigma}_N=(\mathbb{E}[N^{-1}\Sigma_{i \in \mathcal{I}_{N}}\mathbf{z}_{N;i}\mathbf{d}_{N;i}^{\top}]^{\top}\mathbf{\Omega}_{N}^{-1}\mathbb{E}[N^{-1}\Sigma_{i \in \mathcal{I}_{N}}\mathbf{z}_{N;i}\mathbf{d}_{N;i}^{\top}])^{-1}\text{.}\label{eff_Sigma}
\end{equation}
\end{theorem}

We present the proof for Theorem \ref{t2} in Appendix \ref{Appendix_A}. The following subsections discuss the relationship between the precision of the estimator in \eqref{theta_GMM} and the levels of sparsity of the endogenous population network of interest, $\mathcal{G}_{N}$.

\subsection{Precision and Sparsity}
Some individuals in the population may not have any connections at distance $p$, and may affect the identifying moment condition in Theorem \ref{idtheorem}. To consider the effect of changes in identifying information on the asymptotic variance-covariance matrix in \eqref{eff_Sigma}, define $\eta_{N,0;i}^{p}$ to be a random variable equal to one if individual $i$ has at least one connection at distance $p$ and zero otherwise. Let $\kappa_{N,0;i}^{p}=\mathbb{E}[\eta_{N,0;i}^{p}]$ be the unconditional probability that the individual $i$ has at least one connection at distance $p$. Define $\boldsymbol{H}_{N,0;i}\equiv\text{diag}(\eta_{N,0;i}^{p}, \dots, \eta_{N,0;i}, 1, \dots, 1)$ to be a $[(p+1)K+1]\times[(p+1)K+1]$ matrix where the first $K$ elements contain the random variables that determine whether the individual $i$ has at least one connection at distance $p$ and the second $K$ elements are the random variables that determine whether or not the individual $i$ has at least one connection at distance $p-1$, etc., until the last $K$ elements associated with $\mathbf{W}_{N,0}\mathbf{X}_{N}$, where $\eta_{N;i}$ is the random variable that determines whether $i$ is isolated in the network $\mathcal{G}_{N,0}$. Finally, the last $K+1$ elements in the lower right submatrix, which coincide with the non-network regressors $\widetilde{\mathbf{x}}_{N;i}$, are ones. Define the $[(p+1)K+1]\times[(p+1)K+1]$-matrix $\boldsymbol{K}_{N,0;i}\equiv\text{diag}(\kappa_{N,0;i}^{p}, \dots, \kappa_{N,0;i}, 1, \dots, 1)$ to be $\mathbb{E}[\boldsymbol{H}_{N,0;i}]$. 

Note that when $\eta_{N,0;i}^{p}=0$, the first $K$ elements of $\mathbf{z}_{N;i}$ equal zero. Similarly, when $\eta_{N,0;i}^{p-1}=0$, the second $K$ elements of $\mathbf{z}_{N;i}$ equal zero. The same argument repeats until the $K$ elements associated with $\mathbf{W}_{N,0}\mathbf{X}_{N}$, where if $\eta_{N;i}=0$ (individual $i$ is isolated), the elements of $\mathbf{z}_{N;i}$ associated with the component $(\mathbf{w}_{N,0;i}\mathbf{X}_{N})^{\top}$ are equal to zero. Therefore, by the law of total expectation, $\mathbb{E}[\sum_{i \in \mathcal{I}_{N}}\mathbf{z}_{N;i}\mathbf{d}_{N;i}^{\top}]$ can be written as $\boldsymbol{K}_{N,0;i}\mathbb{E}[\mathbf{z}_{N;i}\mathbf{d}_{N;i}^{\top}\mid \boldsymbol{H}_{N,0;i}^{\ast}\neq \boldsymbol{O}_{pK}]$, where $\boldsymbol{H}_{N,0;i}^{\ast}$ contains the left top $(pK\times pK)$-upper matrix of $\boldsymbol{H}_{N,0;i}$ and $\boldsymbol{O}_{pK}$ is the $pK \times pK$ zero matrix. Therefore, \eqref{eff_Sigma} depends on the population probabilities that an individual provides identification information. For low values of these probabilities, the upper right submatrix of $\boldsymbol{K}_{N;0,i}$ approaches the zero matrix, and the variance-covariance matrix could grow arbitrarily large. In the extreme case of nonidentification, \eqref{eff_Sigma} diverges to infinity. Theorem \ref{t2} exposes a relationship between the precision of network parameters and the sparsity of the network.

\subsection{Efficient Weight Matrix Estimation}
To construct an efficient version of the proposed GMM estimator, we need a consistent estimator of $\mathbf{\Omega}_{N}$. Here we use \citeauthor{Kojevnikov2020}'s \citeyearpar{Kojevnikov2020} network heteroskedasticity and autocorrelation-consistent (HAC) variance estimator. Let $D_{n}$ represent a bandwidth after which the dependence between individuals vanishes. For example, \cite{Kojevnikov2020} proposes $D_{n}=C \times [\log (\text{average degree} \vee(1+0.05))]^{-1} \times \log n$, and this rule of thumb is used in our Monte Carlo simulations and in our empirical study with $C=1.8$ according to their suggestion. The proposed variance-covariance matrix estimator is then given by  

\begin{equation}
\widetilde{\mathbf{\Omega}}_{n} =\sum_{d\geq 0}K(d/D_{n}) \frac{1}{n}\sum_{i=1}^{n}\sum_{j \in \mathcal{P}_{n}(i,d)}\mathbf{z}_{n;i}\widetilde{\varepsilon}_{n;i}\widetilde{\varepsilon}_{n;j}\mathbf{z}_{n;j}^{\top},\label{Omega_tilde}
\end{equation}

\noindent where $\widetilde{\varepsilon}_{n;i} = y_{n;i} - \mathbf{d}_{n;i}^{\top}\widetilde{\boldsymbol{\theta}}_{\text{GMM}}$; $K(\cdot)$ is a kernel (weighting) function such that $K(0)=1$ and $K(u)=0$ for $u>1$; and $\widetilde{\boldsymbol{\theta}}_{\text{GMM}}$ is a preliminary consistent estimator. In \eqref{theta_GMM} with $\mathbf{A}_n$ equal to the identity matrix or $n^{-1}\mathbf{Z}_{n}^{\top}\mathbf{Z}_{n}$, the latter was chosen in Monte Carlo exercises and the empirical study. In the second step, the feasible efficient GMM estimator is defined with $\mathbf{A}_n=\widetilde{\mathbf{\Omega}}_{n}^{-1}$ in \eqref{theta_GMM}, which we call $\widehat{\boldsymbol{\theta}}_{\text{GMM}}^{\star}$.

\subsection{Standard Error Calculation}

It follows that the efficient variance-covariance matrix \eqref{eff_Sigma} can be estimated by 

\begin{equation}
\left[n^{-1}\mathbf{D}_{n}^{\top}\mathbf{Z}_{n}\widehat{\mathbf{\Omega}}_{n}^{\star-1}n^{-1}\mathbf{Z}_{n}^{\top}\mathbf{D}_{n}\right]^{-1}\text{,}\label{feasible_var_cov}
\end{equation}

\noindent where $\widehat{\mathbf{\Omega}}_{n}^{\star}$ is calculated as in \eqref{Omega_tilde}, but using $\widehat{\boldsymbol{\theta}}_{\text{GMM}}^\ast$ instead. Standard errors can then be calculated taking the square root of the main diagonal elements of \eqref{feasible_var_cov} after dividing them by $n$.

Note that akin to any other IV procedure, potentially weak instruments can be a concern here. For example, within our specific framework, if both observed networks $\mathcal{G}_{n,0}$ and $\mathcal{G}_{n}$ exhibit a high sparsity with minimal overlap between connections, the identification power of the IV generated by the exogenous network would be inherently weak. Consequently, it is empirically recommended to validate the substantial overlap between the two networks and to confirm that the characteristics in $\mathbf{Z}_{n}$ significantly predict the components of $\mathbf{D}_{n}$ in an empirical application (see, e.g., Section \ref{validate_test} in Appendix \ref{Appendix_C}). The development of a theory addressing weak IVs in the network setting is beyond the scope of our paper and therefore left for future research.

%% file: simulations.tex
To showcase the versatility of the proposed estimator in this paper, this section documents its performance in two different data-generating processes (hereafter DGPs) where the endogeneity is generated by a simultaneous determination of network formation and outcomes--also known as unobserved homophily--(Design 1) and measurement error in the connections (Design 2). A total of 1,500 data sets $\left\{y_{n;i},x_{n;i},\{w_{n;i,j}\}_{j=1,j\neq i}^n,\{w_{n,0;i,j}\}_{j=1,j\neq i}^n\right\}_{i=1}^n$; with $n\in\left\{50,100,200\right\}$, are generated from \eqref{lmm} by setting $k=1$ and drawing $\left\{x_{n,i}\right\}_{i=1}^n$ as a random sample from a normal distribution with a mean of zero and variance of 3. We set the true vector of the parameters at $\boldsymbol{\theta}_{0}=[\beta_{0}, \delta_{0}, \gamma_{0}]^\top=[0.7,1,1]^\top$. The other data components are constructed using the following rules.

\subsubsection*{Design 1: Unobserved Characteristics with Homophily\label{d3}}
 
In this design, the outcome variable of the individual $i$, $y_{n,i}$, and the connections $\left\{w_{n;i,j}\right\}_{j=1,j\neq i}^n$ are jointly determined by a common idiosyncratic homophily-related unobserved variable $\varepsilon_{n,1;i}^{\ast}$. First, an exogenous adjacency matrix $\mathbf{W}_{n,0}=[w_{n,0;i,j}]$ from an \citeauthor{Erdos1959}'s \citeyearpar{Erdos1959} random graph with a density of 0.01 is generated along with a $n\times 1$ vector $\boldsymbol{\varepsilon}_{n,1}^{\ast}=[\varepsilon_{n,1;1}^{\ast},\ldots,\varepsilon_{n,1;n}^{\ast}]^\top$ from a multivariate standard normal distribution.\footnote{The density of a network is the ratio between the total numbers of actual ties and of potential ties.} The elements of the endogenous adjacency matrix $\mathbf{W}_{n}=[w_{n;i,j}]$ are then calculated as

\begin{equation*}
w_{n;i,j}=
\begin{cases}
\mathds{1}(|\varepsilon_{n,1;i}^{\ast}-\varepsilon_{n,1;j}^{\ast}|<\widehat{F}_{\varepsilon_{n,1}^\ast}^{-1}(0.95))\times (1-w_{n,0;i,j}) + w_{n,0;i,j} & \text{, if $\varepsilon_{n,1;i}^{\ast}>\Phi^{-1}(0.95)$;} \\
\mathds{1}(|\varepsilon_{n,1;i}^{\ast}-\varepsilon_{n,1;j}^{\ast}|<\widehat{F}_{\varepsilon_{n,1}^\ast}^{-1}(0.95)) \times w_{n,0;i,j} & \text{ if $\varepsilon_{n,1;i}^{\ast}<\Phi^{-1}(0.05)$;} \\
w_{n,0;i,j} & \text{, otherwise};
\end{cases}
\end{equation*}

\noindent where $\widehat{F}_{\varepsilon_{n,1}^{\ast}}^{-1}(0.95)$, which represents the 95\% empirical quantile of the elements of $\boldsymbol{\varepsilon}_{n,1}^{\ast}$, $\varepsilon_{n,1;k}^{\ast}$ represents its $k$th element, and $\Phi^{-1}(\cdot)$ represents the inverse of the cumulative distribution function of a standard normal random variable. The $n\times 1$ vector of outcomes, $\mathbf{y}_{n}$, is then constructed from \eqref{lmm} by setting $\boldsymbol{\varepsilon}_{n}=m \times \boldsymbol{\varepsilon}_{n,1}+\boldsymbol{\varepsilon}_{n,2}$, where $m\in\left\{1,3\right\}$, $\boldsymbol{\varepsilon_{n,2}}$ is drawn from a multivariate standard normal distribution. The elements of $\boldsymbol{\varepsilon}_{n,1}$ are defined as

\begin{equation*}
\varepsilon_{n,1;i}=
\begin{cases}
\varepsilon_{n,1;i}^{\ast} & \text{, if $\varepsilon_{n,1;i}^{\ast}<\Phi^{-1}(0.05)$ or $\varepsilon_{n,1;i}^{\ast}>\Phi^{-1}(0.95)$;} \\
0 & \text{, otherwise}.
\end{cases}
\end{equation*}

This design captures the homophily idea; i.e., agents endowed with a large value of $\varepsilon_{n,1}$ will tend to create / maintain connections with those also endowed with large values of $\varepsilon_{n,1}$ and sever them with those with low values of this unusual unobserved characteristic.

\subsubsection*{Design 2: Misclassified Links\label{d2}}

This is a modified version of \citeauthor{Lewbel_Qu_Tang}'s \citeyearpar{Lewbel_Qu_Tang} Monte Carlo design. While the true DGP involves an unobserved adjacency matrix $\mathbf{W}_{n,0}^{\ast}=[w_{n,0;i,j}^{\ast}]$ generated from a standard \citeauthor{Erdos1959} \citeyearpar{Erdos1959} random network model with a density of 0.01 for size $n$, it is assumed that the empiricist only has access to an adjacency matrix $\mathbf{W}_{n}=[w_{n;i,j}]$, with randomly misclassified links; that is, $w_{n;i, j}=w_{n,0;i,j}^{\ast}e_{n,1;i,j}+(1-w_{n,0;i,j}^{\ast}) e_{n,2;i,j}$ for $i \neq j$ and an exogenous adjacency matrix $\mathbf{W}_{n,0}=[w_{n,0;i,j}]$, where $w_{n,0;i,j}=w_{n,0;i,j}^{\ast}b_{n,1;i,j}+(1-w_{n,0;i,j}^{\ast}) b_{n,2;i,j}$  for $i \neq j$. The $e_{n,1;i,j}$, $e_{n,2;i,j}$, $b_{n,1;i,j}$, and $b_{n,2;i,j}$ are Bernoulli random variables drawn independently from each other $\forall i\neq j$ with parameters $0.5$, $0$, $1-\tau$, and $0.002$, respectively. The design parameter $\tau\in\left\{0.01,0.05\right\}$ controls the probability of misclassification in $\mathbf{W}_{n,0}$. Notice that, as in \cite{Lewbel_Qu_Tang}, nonexisting links are never misclassified in $\mathbf{W}_n$, but misclassification of these nonexisting links is allowed in $\mathbf{W}_{n,0}$ with a very small probability of 0.2\%. However, this design makes the vector of individual outcomes an explicit function of the proportion of misclassification in $\mathbf{W}_n$ for each $i$; that is, the $n\times 1$ vector $\mathbf{y}$ is constructed following Equation \eqref{lmm}, where $\boldsymbol{\varepsilon}_{n}=\boldsymbol{\varepsilon}_{n,1}+\boldsymbol{\varepsilon}_{n,2}$, $\varepsilon_{n,1;i}=1/n\sum_{j=1}^{n}w_{n,0;i,j}^{\ast}e_{n,1;i,j}$, and $\boldsymbol{\varepsilon}_{n,2}$ is drawn from a multivariate standard normal distribution independently of everything else. 

\subsubsection*{Results\label{MC_results}}

\noindent Figures \ref{boxplot} and \ref{qqplot} show the results in terms of box plots and Q-Q plots of the Monte Carlo replications. Apart from implementing the proposed efficient GMM estimator described in Theorem \ref{t2} for $p\in{2,3}$,  the performance of the standard Ordinary Least Squares (OLS) estimator and the Generalized Two Stage Least Squares (G2SLS) estimator are also included. All adjacency matrices in all designs are row normalized prior to estimation \citep{liu2014}. The calculation of the efficient GMM requires an estimator of the variance-covariance matrix $\mathbf{\Omega}_{n}$. We use the standard two-stage GMM procedure to calculate the efficient weighting matrix. In the first step, we calculate the GMM estimator for $\boldsymbol{\theta}$ setting $\mathbf{A}_{n}=(\mathbf{Z}_{n}^{\top}\mathbf{Z}_{n})^{-1}$. We then use the estimated coefficients in the first step to calculate the network HAC variance estimator in \eqref{Omega_tilde}, where we choose $K(\cdot)$ to be the Parzen kernel, we set the bandwidth $D_{n}=1.8\times [\log (\text { average degree } \vee(1+0.05))]^{-1} \times \log n$ as in \cite{Kojevnikov2020}.

Each panel in Figure \ref{boxplot} displays the performance of the three estimators when the state of a design (Des.) changes by changing the relevant design parameter $m$ or $\tau$. The box plots are based on Monte Carlo replications of OLS (black), G2SLS (dark gray), and the two proposed GMM estimators for $p=2$ (gray) and $p=3$ (light gray) of the social effects. Peer effects ($\beta$), contextual effects ($\delta$), and direct effects ($\gamma$) in \eqref{lmm} are shown with whiskers that show the empirical Monte Carlo quantiles 5\% and 95\%. Across the board, for all parameters, designs, and sample sizes, the proposed GMM estimators perform better than the OLS and G2SLS estimators in terms of bias and sampling variability. As expected, the estimation variability decreases when going from $p=2$ to $p=3$. In contrast, these results also show that naive OLS and G2SLS could potentially lead to estimates with substantial biases in the presence of an endogenous network in a linear-in-means model. On average, the G2SLS underestimates the real value of the peer effects coefficient for the case of misclassified links (Design 2).

Similarly, Figure \ref{qqplot} displays the corresponding Q-Q plots for the GMM based on the standardized version of the Monte Carlo replications of the GMM estimator of the same social effects for sample sizes $n=50$ (light gray), $n=100$ (gray), and $n=200$ (black). The blue dashed line depicts the 45-degree line. This plot shows that the asymptotic normal approximation in Theorem \ref{t2} works well even with a sample as small as 50 observations. Furthermore, as the sample size increases, the approximation improves for all parameters and designs.

\begin{figure}
    \caption{Box Plots of the OLS, G2SLS, and GMM Estimators of Social Effects}
        \input{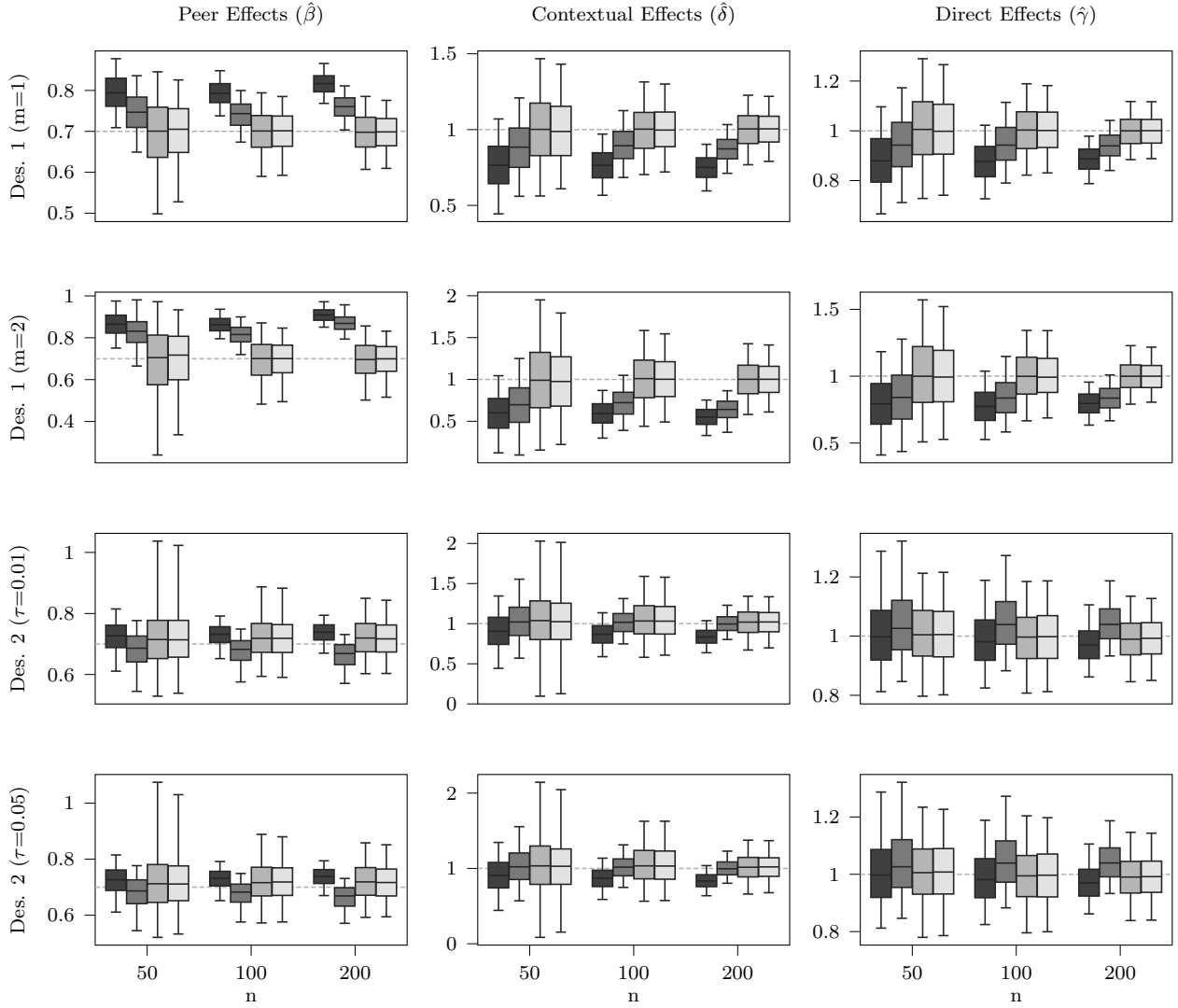}
        
    \vspace{0.4cm}
    \begin{minipage}{1\textwidth}
    \footnotesize
    Note: Box plots depict the Monte Carlo performance of OLS (black), G2SLS (dark gray) and the proposed efficient GMM estimator for $p=2$ (gray) and $p=3$ (light gray). The box plots are based on 1,500 replications of Design 1 (Des.1 ) and Design 2 (Des. 2) for sample sizes $n\in\left\{50,100,200\right\}$. The whiskers display the 5\% and 95\% empirical quantiles. Parameters $m$ and $\tau$ control the level of endogeneity and the probability of misclassification in $\mathbf{W}_{n}$, respectively.
    \end{minipage}
    \label{boxplot}
\end{figure}

\begin{figure}
    \caption{Q-Q Plots for the GMM Estimator of Social Effects}
        \input{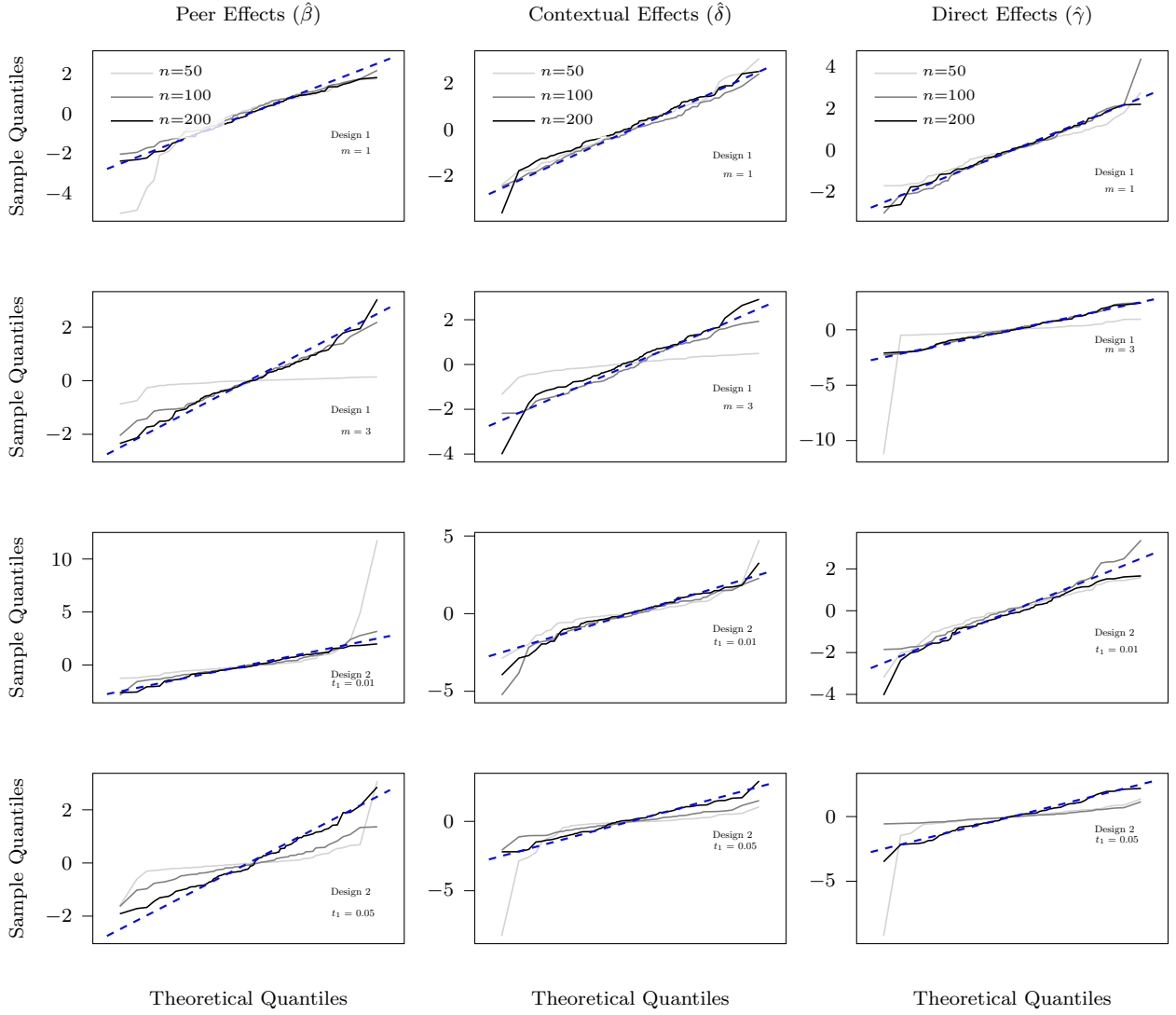}
        
    \vspace{0.4cm}
    \begin{minipage}{1\textwidth}
    \footnotesize
    Note: Q-Q plots are based on the standardized sample of 1,500 Monte Carlo replications of the proposed GMM estimator of the parameters in \eqref{lmm} for Design 1 (Des.1 ) and Design 2 (Des. 2) for sample sizes $n=50$ (light gray), $n=100$ (gray), and $n=200$ (black); $p=3$. The blue dashed line shows the 45 degree line. Parameters $m$ and $\tau$ control the level of endogeneity and the probability of misclassification in $\mathbf{W}_{n}$, respectively.
    \end{minipage}
    \label{qqplot}
\end{figure}

%% file: empirical.tex
We provide an empirical illustration of the proposed methods in the framework of estimating peer effects on academic performance among high school students. The data set was collected between March and May 2011 as part of the Hong Kong Secondary Education Survey in Hong Kong (SESHK). The survey was carried out in the second semester before the final exams and involved three secondary schools with 868 students participating. The sample includes students in the seventh grade of all three schools and students in the eighth and ninth grades of one school (g $\in\{7,8,9\}$). Each grade within a school is made up of five different sections (cl $\in\{1,\ldots,5\}$). Table \ref{tables:table1} in the supplemental material shows the summary statistics for the variables we use. Additional details about these variables can be found in Section \ref{descrip} of the supplemental material. 

In the survey, students were asked to report lists of up to ten peers within their grade with whom they discussed schoolwork issues and who were sitting nearby in class during the first semester. We used this information to build the \textit{study partner} and \textit{seatmate} networks, respectively. Classroom seat assignments undergo multiple changes throughout the semester, and the class teacher determines these adjustments. Unlike study partners, the seatmate network is based on proximity and is enforced by the school. Consequently, the seatmate network emerges as a compelling choice to serve as the instrumental network,  $\mathbf{W}_{n,0}$, in our analysis. That is, we treat the seatmate network as ex-ante exogenous in this analysis. On the contrary, students have the autonomy to select their study partners, and this choice can be influenced by unobservable characteristics that also affect exam performance, or the study partner choices can depend on other exams results, making the study partner network $\mathbf{W}_{n}$, potentially endogenous.  Table \ref{tables:table2} in the supplemental material reports the summary statistics of the network among all students by school.

We fit the Linear-in-Means model for social effects in \eqref{intro_lmm}. Specifically, for the individual $i$, we fit the following version of the model.
\begin{eqnarray} \label{emp_eq1}
\texttt{math}_{i,\text{s}\times\text{g}\times\text{cl}} &=& \alpha + \beta \sum_{j\neq i}^n w_{n;i,j}\texttt{math}_{j,\text{s}\times\text{g}\times\text{cl}} \\ \nonumber
       & & {} +  \sum_{j\neq i}^n w_{n;i,j}\texttt{characteristics}_{j,\text{s}\times\text{g}\times\text{cl}}^{\prime}\delta_{\texttt{characteristics}} \\ \nonumber
       & & {} +  \sum_{j\neq i}^n w_{n;i,j}\texttt{personality}_{j,\text{s}\times\text{g}\times\text{cl}}^{\prime}\delta_{\texttt{personality}} \\ \nonumber
       & & {} + \texttt{characteristics}_{i,\text{s}\times\text{g}\times\text{cl}}^{\prime}\gamma_{\texttt{characteristics}}
       +\texttt{personality}_{i,\text{s}\times\text{g}\times\text{cl}}^{\prime}\gamma_{\texttt{personality}} \\ \nonumber
       & & {} + \sum_{\text{s}=1}^{3}\sum_{\text{g}=7}^{9}\sum_{\text{cl}=1}^{5}f_{\text{s}\times\text{g}\times\text{cl}}\times\mathbb{I}\{i\in{\text{s}\times\text{g}\times\text{cl}}\} + \epsilon_i\text{,} 
\end{eqnarray}

\noindent where $\texttt{math}_{i,\text{s}\times\text{g}\times\text{cl}}$ is the natural logarithm of the first math test score of student $i$ in class cl, grade g, and school s; that is, $\mathbb{I}\{i\in{\text{s}\times\text{g}\times\text{cl}}\}=1$; $\texttt{personality}_{i,\text{s}\times\text{g}\times\text{cl}}$ includes the natural logarithm of cognitive ability, agreeableness, conscientiousness, extraversion, neuroticism, and openness test scores. Similarly, $\texttt{characteristics}_{i,\text{s}\times\text{g}\times\text{cl}}$ includes variables such as height, weight, indicator variables such as help from siblings, parents help, whether they commute to school by car or taxi; whether they play music; and whether the student is a man. We also include interaction terms between the male indicator and the test scores. Parameters $f_{\text{s}\times\text{g}\times\text{cl}}$ are jointly estimated with the social effects after setting $f_{3\times 9\times 5}=0$. The adjacency matrices are row-normalized before estimating the model as permitted by our theory. The effect of having more peers is captured by including the degree of the network (number of study partners). We also control for the fact that some students study alone (isolated).  We set $\delta_{\texttt{personality}} = -\gamma_{\texttt{personality}} $ for all estimation routines to avoid potential collinearity problems. Behaviorally, this restriction implies that only deviations from the students' own personality characteristics from the average of their peers affect the students' tests scores; see, for example, \cite{liu2014}.

The model \eqref{emp_eq1} is estimated using simple Ordinary Least Squares (OLS) with standard errors clustered at the $\text{s}\times\text{g}\times\text{cl}$ level, the Generalized Two-Stage Least Squares (G2SLS) of \cite{Kelejian1998,Kelejian_Prucha_1999_ER}, \cite{Lee2003}, and \cite{Bramoulle2009} with clustered standard errors as in the OLS estimator, and our proposed efficient GMM estimator with $p=5$, $C=1.8$ with the Tukey-Hanning kernel.

\begin{table}[H]
\vspace{-5em}
\def\arraystretch{1.00}
\setlength{\tabcolsep}{15pt} 
\centering
    \begin{threeparttable}
        \caption{Estimations results} 
        \vspace{0.1em}
        \footnotesize
        \input{tables//tables_rev/table3}
        \label{tables:table3}
        \begin{tablenotes}[para,flushleft]
            \footnotesize
            \raggedright
            \item{Note:} (i) \text{*} \(p<0.10\), \text{**}   \(p<0.05\), \text{***}  \(p<0.01\); (ii) Standard errors are in parentheses. (iii) $\dagger$ These regressors are measured as the deviation of students' personality from their peers' average.
        \end{tablenotes}
    \end{threeparttable}
\end{table}

Tables \ref{tables:table3} show the results for a subset of all regressors included in \eqref{emp_eq1} (Table \ref{tables:table4} in the supplemental material contains the results for the remaining set of regressors). All estimators show positive spillover effects; however, both OLS and G2SLS are significantly smaller than what the proposed GMM estimator uncovers. Direct-effect coefficients are relatively consistent between different estimators, both in terms of size and direction. Importantly, the signs of most coefficients match previous findings in the literature. For example, cognitive ability and noncognitive characteristics, such as conscientiousness, have a positive effect on math grade achievement \citep{Heckman2001}. There is no gender gap between the grades of men and women. Interestingly, the estimated coefficient for the degree variable suggests that having a larger number of study partners has a positive and significant effect on achievement and that result is robust across different estimators.

Our results suggest that estimators that do not control network endogeneity tend to underestimate peer effects. The observation of a negative bias in estimators that do not account for endogeneity aligns with findings in the literature, which highlight a similar negative bias in standard maximum likelihood estimators of spatial autoregressive models that overlook the reflection issue \citep{Mizruchi2008, Neuman2010}. This trend of negative bias, as observed between OLS and G2SLS, also extends to our GMM estimator, which incorporates control for network endogeneity. 

All results are qualitatively robust to different choices of $p$ and $D_n$, see, i.e. Section \ref{supp_estimates_output} in the supplemental material. We perform empirical tests to validate the exclusion restriction in Assumption \ref{excrest} and the relevance condition in Assumption \ref{relevance}. For the exclusion restriction, we perform a Least Squares (LS) regression that includes all variables in \eqref{emp_eq1}, but also includes our proposed instruments. The idea is to measure to what extent the instruments are good predictors of our variable of interest. We interpret the lack of predictability as evidence in favor of our exclusion restriction assumption. Table \ref{tables:table14} in the supplemental material. We find that all our instruments are not predictive of the outcome equation. For the relevance assumption, we perform a series of LS regressions in which the different outcomes are all different endogenous variables. We include all our instruments as regressors and calculate the $F$-statistic for each of the regressions. Consistent with the relevance condition, we reject the hypothesis that our instruments are jointly significant at the 1\% significance level for all our endogenous variables (see Section \ref{validate_test} in the supplemental material).  

%% file: tables/tables_rev/table3.tex
\begin{tabular}{lccc}
\toprule
                  Variables &        OLS &      G2SLS &        GMM \\
\midrule
       \textbf{Peer effect} &            &            &            \\
              ln(Math Test) &  0.2455*** &  0.4534*** &  0.6065*** \\
                            &   (0.0864) &   (0.1027) &   (0.1755) \\
                   \midrule 
\textbf{Contextual effects} &            &            &            \\
                       Male &    -0.0232 &    -0.0153 & -0.2455*** \\
                            &   (0.0311) &   (0.0276) &   (0.0873) \\
                 ln(Height) &     0.0794 &     0.0651 &  2.1682*** \\
                            &   (0.2229) &   (0.1910) &   (0.7400) \\
                 ln(Weight) &     0.1012 &     0.0399 &     0.0329 \\
                            &   (0.0642) &   (0.0600) &   (0.1942) \\
              Siblings Help &     0.0263 &     0.0305 &     0.0294 \\
                            &   (0.0261) &   (0.0217) &   (0.0491) \\
               Parents Help &     0.0345 &     0.0064 &    -0.0106 \\
                            &   (0.0287) &   (0.0189) &   (0.0517) \\
        Commute by Car/Taxi &     0.0277 &     0.0250 &   0.1255** \\
                            &   (0.0266) &   (0.0224) &   (0.0636) \\
                      Music &    -0.0212 &    -0.0090 &   0.1187** \\
                            &   (0.0166) &   (0.0160) &   (0.0480) \\
                   \midrule 
   \textbf{$\dagger$} &            &            &            \\            
              ln(Cognitive) &     0.0549 &   0.0734** &  0.1086*** \\
                            &   (0.0357) &   (0.0369) &   (0.0286) \\
          ln(Agreeableness) & -0.1145*** & -0.1275*** &    -0.0712 \\
                            &   (0.0300) &   (0.0356) &   (0.0457) \\
      ln(Conscientiousness) &     0.0445 &     0.0537 &   0.0799** \\
                            &   (0.0371) &   (0.0396) &   (0.0406) \\
           ln(Extraversion) &  -0.0998** &  -0.0982** & -0.1164*** \\
                            &   (0.0387) &   (0.0384) &   (0.0405) \\
            ln(Neuroticism) &    -0.0386 &    -0.0409 &    -0.0113 \\
                            &   (0.0378) &   (0.0373) &   (0.0257) \\
               ln(Openness) &     0.0461 &     0.0358 &     0.0272 \\
                            &   (0.0485) &   (0.0489) &   (0.0414) \\
            \midrule
            \textbf{Direct effects} &            &            &            \\
                       Male &   -0.7922* &  -0.8352** &    -0.1648 \\
                            &   (0.4596) &   (0.4224) &   (0.5417) \\
                 ln(Height) &  -0.3366** &  -0.2635** & -0.9904*** \\
                            &   (0.1342) &   (0.1210) &   (0.2248) \\
                 ln(Weight) &    -0.0088 &    -0.0106 &    -0.0494 \\
                            &   (0.0293) &   (0.0298) &   (0.0453) \\
              Siblings Help &    -0.0099 &    -0.0172 &   -0.0395* \\
                            &   (0.0176) &   (0.0165) &   (0.0204) \\
               Parents Help &    0.0278* &     0.0225 &     0.0167 \\
                            &   (0.0159) &   (0.0142) &   (0.0139) \\
        Commute by Car/Taxi &    -0.0079 &    -0.0162 &   -0.0251* \\
                            &   (0.0117) &   (0.0111) &   (0.0136) \\
                     Degree &  0.0292*** &  0.0264*** &  0.0248*** \\
                            &   (0.0041) &   (0.0039) &   (0.0034) \\
           Isolate Students &    2.2787* &    0.4720* &     0.1425 \\
                            &   (1.2954) &   (0.2813) &   (0.2935) \\
                   \midrule 
                        $n$ &        868 &        868 &        868 \\
             Adjusted $R^2$ &     0.3372 &     0.3936 &     0.2675 \\
                       RMSE &     0.1854 &     0.1716 &     0.2002 \\
\bottomrule
\end{tabular}

%% file: conclusion.tex
This research adds to the literature on the identification and estimation of social effects with observational network data that often contain endogenous or mismeasured connections. Unlike current approaches, such as those in \cite{Johnsson2019} and \cite{Auerbach2022}, our method does not require the specification and estimation of a model characterizing how connections are created or misclassified. Our method circumvents the imposition of these modeling requirements (along with its potential misspecification issues) by showing how a fully observed set of exogenous connections can be used as an \emph{instrumental} network to uniquely identify and estimate parameters of interest in a widely used linear model of social interactions. Therefore, our approach is semiparametric in nature and hence avoids the usual drawbacks of strong modeling assumptions in this literature.

Another contribution of this research is technical in nature. A byproduct of acknowledging potential network mismeasurement or endogeneity is that it explicitly permits the observed and unobserved characteristics of individuals to be correlated; i.e., creating network dependence across observations in the sample. Our asymptotic results utilize the idea that dependence among observations decreases as a function of their distance in the network; that is, $\psi$-dependence. We show that the resulting estimator can be easily implemented utilizing standard linear GMM estimation routines in popular software like Python, R, or Stata, see, e.g., \cite{netivreg}. The estimator is consistent and asymptotically normal distributed at the standard parametric convergence rate. We characterize the form of the asymptotic variance-covariance matrix that accounts for the network dependence and illustrate how standard errors can be calculated in an empirical application.

Empirically, an important aspect of the proposed methodology is that it recognizes that exogenously imposed connections on individuals do not necessarily cause social effects. However, they can generate new types of freely formed connections that do so; i.e., resorting. The correlation between these two networks is at the heart of our identification and estimation strategy. In this sense, our approach provides an explicit solution to the resorting issue in random network identification strategies, such as in \cite{Moffitt2000}, by distinguishing what type of network creates peer effects (with whom you study, for example) and what other type simply influences these connections, but are otherwise exogenous to the model (for instance, to whom you are randomly assigned to share a physical space). Our empirical and Monte Carlo results show that ignoring the potential network endogeneity can severely bias the network effects estimators. We find significant positive network effects of math test scores among high schoolers from study partners in Hong Kong. These results are in line with previous literature in that they show the existence of strong positive network effects, but suggest that the magnitude of the effects can be larger. We postulate that this could be due to the fact that we focus directly on networks that endogenously emerge after an initial exogenous network assignment.

Finally, the idea of using the initial random assignment of network connections as an instrument differs from using other sources of exogenous variation in two critical aspects. First, under restrictions on exogenous network density, the shape of the network structure together with $K$ regressors can be used to form at least the $K+1$ instruments required to identify endogenous peer effects and contextual effects in the linear-in-means model fitted above. Without the exogenous variation of the randomized network, a researcher would have a difficult task finding $K+1$ different instrumental variables. Second, the standard relevance IV assumption imposes restrictions on the shape of the exogenous network and the process that determines the formation of the network of interest. In particular, we have shown that relevance requires that a number of $p$ powers of the adjacency matrix of the exogenous network to be linearly independent, and connections in the exogenous network need to have an effect on the decision of forming a connection on the endogenous network of interest. The need for linear independence imposes restrictions on the potential randomization of links, which researchers have to follow when designing their experiments. For example, researchers cannot randomly select individuals into groups of the same size if they want to estimate the effects of the network using our method. These are among the important empirical considerations left for future research.

%% file: supp_D.tex
This section provides a set of primitive conditions that imply the relevance condition in Assumption \ref{relevance}. This is formally proven in Proposition \ref{rank_prop} below.

\begin{assumption}
\label{da1}
$\mathbf{X}$ is full column rank for any realization of the matrix of regressors $\mathbf{X}\in\mathscr{X}$ with positive probability in $f_{\mathbf{X}_{N}, \mathcal{G}_{N},\mathcal{G}_{N,0},\boldsymbol{\varepsilon}_{N}}(\mathbf{X},\mathbf{g},\mathbf{g}_{0},\boldsymbol{\varepsilon})$.
\end{assumption}

\begin{assumption}
\label{da2}
For any two regressors $k$ and $\ell$ and a number $p \geq 2$, the expectation $\mathbb{E}[\mathbf{w}_{N,0;i}^{p}\mathbf{x}_{N;k}\mathbf{w}_{N;i}\mathbf{x}_{N;\ell}]$ exists for all $i$.
\end{assumption}

\begin{assumption}
\label{da3}
There exists two \emph{different} numbers $(r, s) \in \mathbb{N}_{+}\times\mathbb{N}_{+}$ such that $\mathbf{I}_{N}$, $\mathbf{W}^{r}$ and $\mathbf{W}^{s}$ are linearly independent for any realization of the network of interest $\mathbf{g}\in \mathscr{G}$ with a positive probability in $f_{\mathbf{X}_{N}, \mathcal{G}_{N},\mathcal{G}_{N,0},\boldsymbol{\varepsilon}_{N}}(\mathbf{X},\mathbf{g},\mathbf{g}_{0},\boldsymbol{\varepsilon})$. Furthermore, $(\gamma_{0,k} \beta_{0}+\delta_{0,k})\neq 0$ for all $k \in \{1,\dots, K\}$.
\end{assumption}

\begin{assumption}
\label{da5}
The equation $x_{i, k}\neq \mathbf{w}_{i}^{m}\mathbf{x}_{k}$ holds for some number $m\in\mathbb{N}_{+}$, for all individuals $i$, any regressor $k$, and any realization of the matrix of regressors $\mathbf{X}\in\mathscr{X}$ and the network of interest $\mathbf{g}\in \mathscr{G}$ with positive probability in $f_{\mathbf{X}_{N}, \mathcal{G}_{N},\mathcal{G}_{N,0},\boldsymbol{\varepsilon}_{N}}(\mathbf{X},\mathbf{g},\mathbf{g}_{0},\boldsymbol{\varepsilon})$.
\end{assumption}

\begin{assumption}
\label{da4}
There exist a number $p\geq 2$ such that $\mathbf{I}_{N}$, $\mathbf{W}_{0}$, $\mathbf{W}_{0}^{2}$, $\dots$, $\mathbf{W}_{0}^{p}$ are linearly independent for any realization of the exogenous network $\mathbf{g}_{0}\in \mathscr{G}_{0}$ with positive probability in $f_{\mathbf{X}_{N}, \mathcal{G}_{N},\mathcal{G}_{N,0},\boldsymbol{\varepsilon}_{N}}(\mathbf{X},\mathbf{g},\mathbf{g}_{0},\boldsymbol{\varepsilon})$.
\end{assumption}

\begin{assumption}
\label{da6}
The joint probability distribution $\Pr(\mathcal{G}_{N}=\mathbf{g},\mathcal{G}_{N,0}=\mathbf{g}_{0})$ is such that $\mathbb{E}[\mathbf{w}_{N,0;i}^{p}\mathbf{x}_{\ell}\mathbf{w}_{i}\mathbf{x}_{N,k}]\neq 0$ for at least $p=2$, and any realizations of $\mathbf{x}_{\ell}$ and $\mathbf{x}_{k}$ with positive probability in $f_{\mathbf{X}_{N}, \mathcal{G}_{N},\mathcal{G}_{N,0},\boldsymbol{\varepsilon}_{N}}(\mathbf{X},\mathbf{g},\mathbf{g}_{0},\boldsymbol{\varepsilon})$ and all $i$.
\end{assumption}

\begin{assumption}
\label{da7}
$\mathbb{E}[\sum_{r=0}^{\infty} \beta_{0}^{r}\mathbf{w}_{N;i}\mathbf{W}_{N}^{r}\varepsilon_{N}\mid \mathbf{W}_{N, 0},\mathbf{X}_{N,k}]=0$ for any $k\in \{1,\dots, K\}$.
\end{assumption}

\begin{proposition}
\label{rank_prop}
Let Assumptions \ref{da1}, \ref{da2}, \ref{da3}, \ref{da4}, \ref{da5}, \ref{da6} and \ref{da7} hold. It follows that the matrix $\mathbb{E}[N^{-1}\sum_{i \in \mathcal{I}_{N}}\mathbf{z}_{N;i}\mathbf{d}_{N;i}^{\top}]$ has full column rank.  
\end{proposition}

\begin{proof}
Let $\mathbf{w}_{N;i}$ and $\mathbf{w}_{N,0;i}$ be the $i$th row of the adjacency matrices representing the population network of interest and the exogenous network. Similarly, let $\mathbf{x}_{N,k}$ be the $N\times 1$ vector of the regressor $k$ for all $N$ individuals in the population and let $x_{N, k; i}$ be the value of the regressor $k$ for the individual $i$. Therefore, it follows that for the individual $i$, $\mathbf{z}_{N;i}\mathbf{d}_{N;i}^{\top}$ equals

\begin{equation}
\label{matrix}
\begin{bmatrix}
			\mathbf{w}_{N,0;i}^{p}\mathbf{x}_{N,1}\mathbf{w}_{N;i}\mathbf{y}_{N} & \mathbf{w}_{N,0;i}^{p}\mathbf{x}_{N,1}\mathbf{w}_{N;i}\mathbf{x}_{N,1} & \dots & \mathbf{w}_{N,0;i}^{p}\mathbf{x}_{N, 1}\mathbf{w}_{N;i}\mathbf{x}_{N,K} & \dots & \mathbf{w}_{N,0;i}^{p}\mathbf{x}_{N,1}x_{N,K;i}\\
			\mathbf{w}_{N,0;i}^{p}\mathbf{x}_{N,2}\mathbf{w}_{N;i}\mathbf{y}_{N} & \mathbf{w}_{N,0;i}^{p}\mathbf{x}_{N,2}\mathbf{w}_{N;i}\mathbf{x}_{N,1} & \dots & \mathbf{w}_{N,0;i}^{p}\mathbf{x}_{N,2}\mathbf{w}_{N;i}\mathbf{x}_{N,K} & \dots & \mathbf{w}_{N,0;i}^{p}\mathbf{x}_{N,2}x_{N,K;i}\\
			\vdots &  \vdots &  & \vdots &  & \vdots \\
			\mathbf{w}_{N,0;i}^{p-1}\mathbf{x}_{N,K}\mathbf{w}_{N;i}\mathbf{y}_{N} & \mathbf{w}_{N,0;i}^{p-1}\mathbf{x}_{N,K}\mathbf{w}_{N;i}\mathbf{x}_{N,1} & \dots & \mathbf{w}_{N,0;i}^{p-1}\mathbf{x}_{N,K}\mathbf{w}_{N;i}\mathbf{x}_{N,K} & \dots & \mathbf{w}_{N,0;i}^{p-1}\mathbf{x}_{N,K}x_{N,K;i}\\
			\vdots &  \vdots &  & \vdots &  & \vdots \\
			\mathbf{w}_{N,0;i}\mathbf{x}_{N,K}\mathbf{w}_{N;i}\mathbf{y}_{N} & \mathbf{w}_{N,0;i}\mathbf{x}_{N,K}\mathbf{w}_{N;i}\mathbf{x}_{N,1} & \dots & \mathbf{w}_{N,0;i}\mathbf{x}_{N,K}\mathbf{w}_{N;i}\mathbf{x}_{N,K} & \dots & \mathbf{w}_{N,0;i}\mathbf{x}_{N,K}x_{N,K;i}\\
			\vdots &  \vdots &  & \vdots &  & \vdots \\
			x_{N,K;i}\mathbf{w}_{N;i}\mathbf{y}_{N} & x_{N,K; i}\mathbf{w}_{N;i}\mathbf{x}_{N,1} & \dots & x_{N,K; i}\mathbf{w}_{N;i}\mathbf{x}_{N,K} & \dots & x_{N,K;i}^{2}
\end{bmatrix}.
\end{equation}

\noindent If, in expectation, the columns of the matrix $\mathbf{z}_{N;i}\mathbf{d}_{N;i}^{\top}$ are linearly independent for all $i$, it follows that $\mathbb{E}[N^{-1}\sum_{i \in \mathcal{I}_{N}}\mathbf{z}_{N;i}\mathbf{d}_{N;i}^{\top}]$ has full column rank. We check the linear independence across columns for a generic row from the matrix in \eqref{matrix}. First, note that for all the components that involve the outcome variable $\mathbf{y}_{N}$, the linearity assumption in \ref{excrest} implies that, for $K=1$,

\begin{equation*}
\mathbf{W}_{N} \mathbf{y}_{N}= \gamma_{0} \mathbf{W}_{N} \mathbf{x}_{N}+\pi_{0} \sum_{r=0}^{\infty} \beta_{0}^{r} \mathbf{W}_{N}^{r+2} \mathbf{x}_{N} +\sum_{r=0}^{\infty} \beta_{0}^{r}\mathbf{W}_{N}^{r+1}\varepsilon_{N},
\end{equation*}

\noindent where $\pi_{0}=(\gamma_{0} \beta_{0}+\delta_{0})$. Thus, all elements in the first column involve infinite powers of the endogenous adjacency matrix $\mathbf{W}_{N}$. In particular, it follows that for an arbitrary individual $i$ and any number of regressors $K$, 

\begin{equation}
\label{outcome}
\begin{split}
\mathbf{w}_{N;i} \mathbf{y}_{N}  =& \gamma_{0,1} \mathbf{w}_{N;i}\mathbf{x}_{N,1}+\dots+\gamma_{0,K} \mathbf{w}_{N;i}\mathbf{x}_{N,K}+\pi_{0, 1} \sum_{r=0}^{\infty} \beta_{0}^{r} \mathbf{w}_{N;i}\mathbf{W}_{N}^{r+1} \mathbf{x}_{N, 1} + \dots \\
& + \pi_{0, K} \sum_{r=0}^{\infty} \beta_{0}^{r} \mathbf{w}_{N;i}\mathbf{W}_{N}^{r+1} \mathbf{x}_{N, K} + e_{i},    
\end{split}
\end{equation}

\noindent where $e_{i}=\sum_{r=0}^{\infty} \beta_{0}^{r}\mathbf{w}_{N;i}\mathbf{W}_{N}^{r}\varepsilon_{N}$ and $\pi_{0, k}=(\gamma_{0, k} \beta_{0}+\delta_{0, k})$ for some $k\in \{1,\dots, K\}$. Choose an arbitrary row $k$ from the first $K$ rows in equation \eqref{matrix}. Taking expectations with respect to the joint distribution $f_{\mathbf{X}_{N}, \mathcal{G}_{N},\mathcal{G}_{N,0},\boldsymbol{\varepsilon}_{N}}(\mathbf{X},\mathbf{g},\mathbf{g}_{0},\boldsymbol{\varepsilon})$, the expected value for row $k$ is given by the vector 

$$[\mathbb{E}[\mathbf{w}_{N,0;i}^{p}\mathbf{x}_{N,k}\mathbf{w}_{N;i}\mathbf{y}_{N}], \dots, \mathbb{E}[\mathbf{w}_{N,0;i}^{p}\mathbf{x}_{N,k}\mathbf{w}_{N;i}\mathbf{x}_{N,K}], \dots, \mathbb{E}[\mathbf{w}_{N,0;i}^{p}\mathbf{x}_{N,k}x_{N,K;i}]].$$

\noindent By definition of conditional expectations and considering the discrete nature of random networks, for any two aggressors $k$ and $\ell$, we can write 

\begin{align}
\nonumber
\mathbb{E}[\mathbf{w}_{N,0;i}^{p}\mathbf{x}_{N,k}\mathbf{w}_{N;i}\mathbf{x}_{N,\ell}] = \sum_{\mathbf{w}_{0}:\mathbf{g}_{0}\in\mathscr{G}_{0}}\int_{\mathbf{x}_{k}:\mathbf{X}\in\mathscr{X}}&\mathbf{w}_{0}^{p}\mathbf{x}_{k}\mathbb{E}[\mathbf{w}_{N;i}\mathbf{x}_{N,\ell}\mid \mathbf{w}_{N,0;i}=\mathbf{w}_{0}, \mathbf{x}_{N,k} = \mathbf{x}_{k}]\\
\label{expect}
&f_{\mathbf{x}_{N,k}}(\mathbf{x}_{k})\Pr(\mathcal{G}_{N,0}=\mathbf{g}_{0})d\mathbf{x}_{k},
\end{align}

\noindent where $f_{\mathbf{x}_{N,k}}(\mathbf{x}_{k})\Pr(\mathcal{G}_{N,0}=\mathbf{g}_{0})$ represents the product of the marginal distributions of the $k$th regressors and the exogenous network. We can represent the distribution of the regressors and the exogenous network by the products of the marginals because of the independence guaranteed by the properties of randomization. From Assumption \ref{da2}, $\mathbb{E}[\mathbf{w}_{N,0;i}^{p}\mathbf{x}_{N,k}\mathbf{w}_{N;i}\mathbf{x}_{N,\ell}]$ exists, which implies that the conditional expectations defined in \eqref{expect} also exist. Choose arbitrary values of $\mathbf{x}_{k}$ and $\mathbf{g}_{0}$ such that $\mathbb{E}[\mathbf{w}_{N;i}\mathbf{x}_{N,\ell}\mid \mathbf{w}_{0;i}=\mathbf{w}_{0}, \mathbf{x}_{k} = \mathbf{x}_{k}]\neq0$, and that occur with positive probability in $f_{\mathbf{x}_{N,k}}(\mathbf{x}_{k})$ and $\Pr(\mathcal{G}_{N,0}=\mathbf{g}_{0})$. We can collect all the values related to the same arbitrary regressor $\mathbf{x}_{\ell}$ from the expectation vector $[\mathbb{E}[\mathbf{w}_{N,0;i}^{p}\mathbf{x}_{N,k}\mathbf{w}_{N;i}\mathbf{y}_{N}], \dots, \mathbb{E}[\mathbf{w}_{N,0;i}^{p}\mathbf{x}_{N,k}\mathbf{w}_{N;i}\mathbf{x}_{N,K}], \dots, \mathbb{E}[\mathbf{w}_{N,0;i}^{p}\mathbf{x}_{N,k}x_{N,K;i}]]$, which, after replacing $\mathbf{w}_{N;i}\mathbf{y}_{N}$ with the expression in \eqref{outcome} and considering the assumption \ref{da7}, is given by 

\begin{align}
\label{vector}
[&\gamma_{0, \ell}\mathbf{w}_{0}^{p}\mathbf{x}_{k}\mathbb{E}[\mathbf{w}_{N;i}\mathbf{x}_{N,\ell}\mid \mathbf{w}_{0},\mathbf{x}_{k}]+\pi_{0, \ell}\mathbf{w}_{0}^{p}\mathbf{x}_{k}\sum_{r=0}^{\infty}\beta^{r}\mathbb{E}[\mathbf{w}_{N;i}\mathbf{W}_{N}^{r+1}\mathbf{x}_{N,\ell}\mid \mathbf{w}_{0},\mathbf{x}_{k}],\\
\nonumber
&\mathbf{w}_{0}^{p}\mathbf{x}_{k}\mathbb{E}[\mathbf{w}_{N;i}\mathbf{x}_{N,\ell}\mid \mathbf{w}_{0},\mathbf{x}_{k}], \mathbf{w}_{0}^{p}\mathbf{x}_{k}\mathbb{E}[x_{N,\ell;i}\mid \mathbf{w}_{0},\mathbf{x}_{k}]].  
\end{align}

\noindent The three components of the vector in \eqref{vector} are linearly dependent if and only if there exist three constants $a$, $b$ and $c$ different from zero such that 

\begin{align}
\label{linind}
&\mathbb{E}[(a\gamma_{0, \ell}+b)\mathbf{w}_{N;i}\mathbf{x}_{N,\ell}+ c x_{N,\ell;i}+a\pi_{0, \ell}\beta\mathbf{w}_{N;i}\mathbf{W}_{N}\mathbf{x}_{N,\ell}+\\
\nonumber
&a\pi_{0, \ell}\beta^{2}\mathbf{w}_{N;i}\mathbf{W}_{N}^{2}\mathbf{x}_{N,\ell}+\dots \mid \mathbf{w}_{0},\mathbf{x}_{k}]=0,
\end{align}

\noindent where the dots represent the infinite sum on $r$. Under the assumption that $\pi_{0, \ell}\neq0$, the only way in which equation \eqref{linind} holds for constant $a$, $b$, and $c$ different from zero is if the matrices $\mathbf{I}_{N}, \mathbf{W}_{N}, \mathbf{W}_{N}^{2},\dots$ are linearly dependent and there exists $\mathbf{W}_{N}^{r}$ such that $x_{i, \ell}= \mathbf{w}_{N;i}^{r}\mathbf{x}_{\ell}$. If $\pi_{0, \ell}=0$, clearly the first and second components of the vector are linearly dependent. Therefore, Assumptions \ref{da3} and \ref{da5} imply that the components of the vector in \eqref{vector} are linearly independent. We arbitrarily chose the regressors $k$ and $\ell$. Then, under Assumption \ref{da1} all the regressors are linearly independent, which implies that all the components of the vector 
$$[\mathbb{E}[\mathbf{w}_{N,0;i}^{p}\mathbf{x}_{N,k}\mathbf{w}_{N;i}\mathbf{y}_{N}], \dots, \mathbb{E}[\mathbf{w}_{N,0;i}^{p}\mathbf{x}_{N,k}\mathbf{w}_{N;i}\mathbf{x}_{N,K}], \dots, \mathbb{E}[\mathbf{w}_{N,0;i}^{p}\mathbf{x}_{N,k}x_{N,K;i}]]$$ 

\noindent are linearly independent. The same result follows for the vectors of the type 

$$[\mathbb{E}[x_{N,k;i}\mathbf{w}_{N;i}\mathbf{y}_{N}], \dots, \mathbb{E}[x_{N,k;i}\mathbf{w}_{N;i}\mathbf{x}_{N,k}], \dots, \mathbb{E}[x_{N,k;i}^{2}]]$$ 

\noindent by conditioning on the arbitrary regressor $x_{N,k; i}$ for nonzero rows. To show that the rows are linearly independent, we can use an analogous approach considering arbitrary values of $\mathbf{w}_{i}$ and $\mathbf{x}_{K}$. It is straightforward to see that under the column rank assumption on all matrices of regressors, Assumption \ref{da4} implies the result. Finally, given that we showed linear independence conditions for nonzero rows, we need to show that there are at least $2K+1$ rows different from zero. The rank condition on the matrix of regressors implies that we only need to focus on combinations of connections in $\mathbf{W}_{N,0}$ and $\mathbf{W}_{N}$ that can make $\mathbb{E}[\mathbf{w}_{0,N,i}^{p}\mathbf{x}_{N,\ell}\mathbf{w}_{N,i}\mathbf{x}_{N,k}]=0$ for any value of $\mathbf{x}_{\ell}$ and $\mathbf{x}_{k}$ and some value of $p$. First, note that 

\begin{equation}
\mathbb{E}[\mathbf{w}_{N,0;i}^{p}\mathbf{x}_{N,\ell}\mathbf{w}_{N;i}\mathbf{x}_{N,k}] = \int_{\mathbf{x}_{k}:\mathbf{X}\in\mathscr{X}}\int_{\mathbf{x}_{\ell}:\mathbf{X}\in\mathscr{X}}\mathbb{E}[\mathbf{w}_{N,0;i}^{p}\mathbf{x}_{\ell}\mathbf{w}_{i}\mathbf{x}_{N}]f_{\mathbf{x}_{N,k}\mathbf{x}_{N,\ell}}(\mathbf{x}_{k},\mathbf{x}_{\ell})d\mathbf{x}_{k}d\mathbf{x}_{\ell},
\end{equation}

\noindent where $f_{\mathbf{x}_{N,k}\mathbf{x}_{N,\ell}}$ is the joint probability of $\mathbf{x}_{N,k}$ and $\mathbf{x}_{N,\ell}$. Take some arbitrary values $\mathbf{x}_{k}\neq 0$ and $\mathbf{x}_{\ell}\neq0$ with positive probability in $f_{\mathbf{x}_{N,k}\mathbf{x}_{N,\ell}}$. It follows that   

\begin{equation}
\label{expectation}
\mathbb{E}[\mathbf{w}_{N,0;i}^{p}\mathbf{x}_{\ell}\mathbf{w}_{i}\mathbf{x}_{N,k}]=\sum_{\mathbf{w}_{0;i}:\mathbf{g}_{0}\in\mathscr{G}_{0}}\sum_{\mathbf{w}_{i}:\mathbf{g}\in\mathscr{G}}\mathbf{w}_{0;i}^{p}\mathbf{x}_{\ell}\mathbf{w}_{i}\mathbf{x}_{k}\Pr(\mathcal{G}_{N}=\mathbf{g},\mathcal{G}_{N,0}=\mathbf{g}_{0}).    
\end{equation}

\noindent The only way in which the equation \eqref{expectation} can equal zero for different values of $p$, even when $\mathbf{x}_{k}\neq 0$ and $\mathbf{x}_{\ell}\neq0$, is if the linear combination of $\mathbf{w}_{0;i}^{p}\mathbf{x}_{\ell}\mathbf{w}_{i}\mathbf{x}_{k}$ for different values of $\mathbf{w}_{0;i}^{p}$ and $\mathbf{w}_{i}$ weighted by their respective probabilities equals zero. Therefore, Assumption \ref{da6} guarantees the existence of at least $3K$ rows different from zero. \hfill\end{proof}

%% file: supp_A.tex
\begin{proof}[Proof of Theorem \ref{idtheorem}]
First, note that Assumption \ref{excrest} guarantees that the solution for model \eqref{lmm} exists. Assumption \ref{relevance} guarantees that the system of equations $\mathbb{E}[\boldsymbol{m}_{N}(\boldsymbol{\theta})]=\boldsymbol{0}_{K}$ are not trivially satisfied by making all individuals $i\in \mathcal{I}_{N}$ isolated. We show that the moment condition equation has a unique root at $\boldsymbol{\theta}_{0}= (\alpha_{0},\beta_{0}, \boldsymbol{\delta}_{0}^\top, \boldsymbol{\gamma}_{0}^\top)^\top$. In particular, we show that there cannot be any other $\boldsymbol{\theta}\in \boldsymbol{\Theta}$ different from $\boldsymbol{\theta}_{0}$ for which the moment condition is satisfied. Choose an arbitrary vector of parameters $\boldsymbol{\theta} \in \boldsymbol{\Theta}$, such that $\mathbb{E}[\mathbf{m}(\boldsymbol{\theta})] = 0$. Assumptions \ref{id} and \ref{excrest} imply that $\mathbb{E}[N^{-1}\sum_{i \in \mathcal{I}_{N}}\mathbf{z}_{N;i}(y_{N;i} - \mathbf{d}_{N;i}^{\top}\boldsymbol{\theta})] = \boldsymbol{0}_{K} $. It follows that $\mathbb{E}[\sum_{i \in \mathcal{I}_{N}}\mathbf{z}_{N;i}\mathbf{d}_{N;i}^{\top}](\boldsymbol{\theta}_{0} - \boldsymbol{\theta})+ \mathbb{E}[\sum_{i \in \mathcal{I}_{N}}\mathbf{z}_{N;i}\varepsilon_{N;i}] = \boldsymbol{0}_{K}$ and $\mathbb{E}[\sum_{i \in \mathcal{I}_{N}}\mathbf{z}_{N;i}\mathbf{d}_{N;i}^{\top}]\left(\boldsymbol{\theta}_{0} - \boldsymbol{\theta}\right) = \boldsymbol{0}_{K}$, given that $N$ is arbitrarily large, but finite. Under Assumption \ref{relevance}, it follows that $\mathbb{E}[\boldsymbol{m}(\boldsymbol{\theta})] = \boldsymbol{0}_{K}$ if and only if $\boldsymbol{\theta}_{0} = \boldsymbol{\theta}$. \hfill\end{proof}

\medskip

\begin{proof}[Proof of Theorem \ref{t2}]
The GMM estimator in \eqref{theta_GMM} in the main text can be written as 

\begin{small}
\begin{equation}
\label{t2est}
\widehat{\boldsymbol{\theta}}_{\text{GMM}} = \boldsymbol{\theta} + (n^{-1}\mathbf{D}_{n}^{\top}\mathbf{Z}_{n}\mathbf{A}_{n} n^{-1}\mathbf{Z}_{n}^{\top}\mathbf{D}_{n})^{-1}n^{-1}\mathbf{D}_{n}^{\top}\mathbf{Z}_{n}\mathbf{A}_{n}n^{-1}\mathbf{Z}_{n}^{\top}\boldsymbol{\varepsilon}_{n}\text{.}
\end{equation}
\end{small}

\noindent By construction, it is assumed that the matrix $\mathbf{A}_{n}$ converges to the full-rank matrix $\mathbf{A}_{N}$ as $n\to \infty$. From Corollary \ref{int_reg}, $n^{-1}\mathbf{Z}_{n}^{\top}\mathbf{D}_{n}$ converges to the population quantity $\mathbb{E}[N^{-1}\sum_{i \in \mathcal{I}_{N}}\mathbf{z}_{N;i}\mathbf{d}_{N;i}^{\top}]$, which is finite given Assumption \ref{relevance}. Finally, Corollary \ref{int_err} shows that $n^{-1}\mathbf{Z}_{n}^{\top}\boldsymbol{\varepsilon}_{n}(\boldsymbol{\theta})$ converges to $\mathbb{E}[N^{-1}\sum_{i \in \mathcal{I}_{N}}\mathbf{z}_{N;i}\varepsilon_{N;i}(\boldsymbol{\theta})]=0$. It then follows that $\widehat{\boldsymbol{\theta}}_{\text{GMM}} = \boldsymbol{\theta}+o_{p}(1)$ as $n\to\infty$. For asymptotic normality, note that from \eqref{t2est}

\begin{equation*}
\sqrt{n}(\widehat{\boldsymbol{\theta}}_{\text{GMM}}- \boldsymbol{\theta}) = (n^{-1}\mathbf{D}_{n}^{\top}\mathbf{Z}_{n}\mathbf{A}_{n}n^{-1} \mathbf{Z}_{n}^{\top}\mathbf{D}_{n})^{-1}n^{-1}\mathbf{D}_{n}^{\top}\mathbf{Z}_{n}\mathbf{A}_{n}\times n^{-1/2}\mathbf{Z}_{n}^{\top}\boldsymbol{\varepsilon}_{n}.
\end{equation*}

\noindent Let $\mathbf{Q}_{zx}=\mathbb{E}[N^{-1}\sum_{i \in \mathcal{I}_{N}}\mathbf{z}_{N;i}\mathbf{d}_{N;i}^{\top}]$. Then from Corollary \ref{int_reg} and Lemma \ref{clt}, it follows that

$$
\sqrt{n} \left(\widehat{\boldsymbol{\theta}}_{\text{GMM}}-\boldsymbol{\theta}\right) \stackrel{d}{\rightarrow}\left[\mathbf{Q}_{zx}^{\top}\mathbf{A}_{N} \mathbf{Q}_{zx}\right]^{-1} \mathbf{Q}_{zx}^{\top}\mathbf{A}_{N}\times \mathcal{N}(\mathbf{0}, \boldsymbol{\Omega}_{N}).
$$

\noindent The result then follows. The efficient variance-covariance matrix in \eqref{eff_Sigma} is derived from standard matrix algebra calculations.\hfill\end{proof}

%% file: supp_B.tex
\noindent All the results in this section, and consequently in section \ref{Appendix_A}, are derived conditional on the sequence of networks $\{\mathcal{G}_{n}\}$. For simplicity in notation, we have omitted explicit conditioning on expectations, but it is important to note that all expectations are taken with respect to the conditional distribution $f_{\mathbf{X}_{N}, \boldsymbol{\varepsilon}_{N}\mid \mathcal{G}_{N}}$.

\begin{lemma}
\label{bounds}
Let Assumption \ref{depas} hold for $\{\mathbf{r}_{n;i}\}_{n\geq 1}$, $i\in\mathcal{I}_{n}$ and define $R_{n;i,j}=f_{q,\ell}(\mathbf{r}_{n,\{i,j\}})\equiv r_{n;i,q}r_{n;j,\ell}$ and $R_{n;h,s}=g_{q^{\prime},\ell^{\prime}}(\mathbf{r}_{n,\{h,s\}})\equiv r_{n;h,q^{\prime}}r_{n;s,\ell^{\prime}}$ for $i,j,h,s \in \mathcal{I}_{n}$, where $q$, $q^{\prime}$, $\ell$, and $\ell^{\prime}$ are components of the vector $\mathbf{r}_{n;i}$. Let Assumption \ref{moments} hold for $R_{i,j}$ and $R_{h,s}$; then

\begin{equation}
\label{covbd}
\left|\operatorname{cov}\left(R_{n;i,j}, R_{n;h,s}\right)\right| \leq 2 \bar{\lambda}_{n, d}(C+16) \times 4\left(\pi_{1}+\tilde{\gamma}_{1}\right)\left(\pi_{2}+\tilde{\gamma}_{2}\right) \underline{\lambda}_{n, d}^{1-p_{f}-p_{g}},
\end{equation}

\noindent where $\underline{\lambda}_{n, d}=\lambda_{n, d} \wedge 1$, $\bar{\lambda}_{n, d}=\lambda_{n, d} \vee 1$, $\pi_{1}=\|\mathbf{r}_{n;i}\|_{p_{f,i}}\|\mathbf{r}_{n;j}\|_{p_{f,j}}$, $\pi_{2}=\|\mathbf{r}_{n;h}\|_{p_{f,h}}\|\mathbf{r}_{n;s}\|_{p_{f,s}}$, $\widetilde{\gamma}_{1} = \max\{\|\mathbf{r}_{n;i}\|_{p_{f,i}+p_{f,j}}, \|\mathbf{r}_{n;j}\|_{p_{f,i}+p_{f,j}}\}$; $\widetilde{\gamma}_{2} = \max\{\|\mathbf{r}_{n;h}\|_{p_{f}}, \|\mathbf{r}_{n;s}\|_{p_{g}}\}$, where $p_{f}=1/p_{f,i}+1/p_{f,j}$ and $p_{g}=1/p_{g,h}+1/p_{g,s}$, where the constant $C$ is the same as in Assumption \ref{depas}. The indexes $i,j,h,s$ and the components $q$, $q^{\prime}$, $\ell$, $\ell^{\prime}$ may or may not be the same. 
\end{lemma}

\begin{proof}
Define the increasing continuous functions $h_{1}(x)$ and $h_{2}(x)$ as in Theorem A.2 in \citet[][Appendix A, pp. 899-907]{Kojevnikov2020} as $h_{1}(x)=h_{2}(x)=x$. Note that the functions $f_{q,\ell}$ and $g_{q^{\prime},\ell^{\prime}}$ are continuous, and their truncated version of the form $\varphi_{K_{1}} \circ f \circ \varphi_{h_{1}}\left(K_{2}\right)$ and $\varphi_{K_{1}} \circ g \circ \varphi_{h_{1}}\left(K_{2}\right)$ for all $K \in(0, \infty)^{2}$ is in $\mathscr{L}_{Q+1,2}$. Assumption \ref{moments} guarantees the existence of the moments defining $\widetilde{\gamma}_{1}$ and $\widetilde{\gamma}_{2}$. Then, Theorem A.2 in \citet[][Appendix A, pp. 899-907]{Kojevnikov2020} applies to this setting \citep[see also Corollary A.2. in Appendix A in ][pp. 899-907]{Kojevnikov2020}. \hfill\end{proof}

\begin{lemma}[LLN for Products of $\psi$-dependent Random Variables]
\label{lln_main}
Let Assumptions \ref{depas} -- \ref{epmoments} hold, define $R_{n;i,j}\equiv r_{n;i,q}r_{n;j,\ell}$, and let $w_{i,j}^{\ast}$ be weights between zero and one. Form $\{R_{n;i,j}\}_{i\in \mathcal{I}_{n}, j \in \mathcal{I}_{i}}$, where $\mathcal{I}_{i}\subset \mathcal{I}_{n}$ is a set of indexes defined for each $i \in \mathcal{I}_{n}$, which can be empty, equal to the union of individual $i$'s connections in the networks $\mathcal{G}_{n}$ and $\mathcal{G}_{n,0}$, or equal to $\mathcal{P}_{n}(i , 1)$. Then, as $n \to \infty$,

$$\left\|\frac{1}{n} \sum_{i \in \mathcal{I}_{n}}\sum_{j\in \mathcal{I}_{i}}w_{i,j}^{\ast}\left(R_{n;i, j}-\mathbb{E}\left[R_{n;i,j}\right]\right)\right\|_{1}{\longrightarrow} 0\text{.}$$
\end{lemma}

\begin{proof}
Using the same approach as \cite{Jenish2009} and \cite{Kojevnikov2020}, let the censoring function $\varphi_{k}(x)=(-K) \vee (K \wedge x)$ be such that, for some $k>0$,

$$
R_{n;i, j}=R_{n;i, j}^{(k)}+\tilde{R}_{n;i, j}^{(k)},
$$

\noindent where $R_{n;i, j}^{(k)}=\varphi_{k}(R_{n;i, j})$ and $\tilde{R}_{n;i,j}^{(k)}=R_{n;i, j}-\varphi_{k}(R_{n;i, j})=(R_{n;i, j}-\operatorname{sgn}(R_{n;i, j}) k) \mathds{1}\{|R_{n;i,j}|>k\}$. Let $\|X\|_{k}=(\mathbb{E}[|X|^{k})^{1/k}$ for $k \in [1, \infty)$. Therefore, following the previous definition, we apply the triangle inequality to get

$\begin{aligned}\left\|\frac{1}{n} \sum_{i \in \mathcal{I}_{n}}\sum_{j\in \mathcal{I}_{i}}w_{i,j}^{\ast}\left(R_{n;i, j}-\mathbb{E}\left[R_{n;i, j} \right]\right)\right\|_{1} \leq &\left\|\frac{1}{n} \sum_{i \in \mathcal{I}_{n}}\sum_{j\in \mathcal{I}_{i}}w_{i,j}^{\ast}\left(R_{n;i, j}^{(k)}-\mathbb{E}\left[R_{n;i, j}^{(k)} \right]\right)\right\|_{1} \\ &+\left\|\frac{1}{n} \sum_{i \in \mathcal{I}_{n}}\sum_{j\in \mathcal{I}_{i}}w_{i,j}^{\ast}\left(\tilde{R}_{n;i,j}^{(k)}-\mathbb{E}\left[\tilde{R}_{n;i, j}^{(k)}\right]\right)\right\|_{1} . \end{aligned}$

\noindent From Assumption \ref{epmoments}, note that the expectation on the second term of the previous equation is bounded by $\mathbb{E}[|\tilde{R}_{n;i,j}^{(k)}| ]=\mathbb{E}[|\tilde{R}_{n;i,j}^{(k)}| \mathds{1}\{|R_{n;i,j}|>k\}] \leq 2 \mathbb{E}[|R_{n;i,j}| \mathds{1}\{|R_{n;i,j}|>k\}]$. Following the arguments as in \cite{Kojevnikov2020}, the second component of the right-hand side of the above equation is bounded by $\sup _{n \geq 1} \max _{i \in \mathcal{I}_{n}} \mathbb{E}\left[\left|R_{n;i, j}\right| \mathds{1}\left\{\left|R_{n;i}\right|>k\right\} \right]$, where $\lim _{k \rightarrow \infty} \sup _{n \geq 1} \max _{i \in \mathcal{I}_{n}} \mathbb{E}\left[\left|R_{n;i, j}\right| \mathds{1}\left\{\left|R_{n;i, j}\right|>k\right\} \right]=0$. Focusing on the first component of the right-hand side, by Lyapunov's inequality, it follows that 

\begin{equation}
\begin{split}
\left\|\frac{1}{n} \sum_{i \in \mathcal{I}_{n}}\sum_{j\in \mathcal{I}_{i}}w_{i,j}^{\ast}\left(R_{n;i, j}^{(k)}-\mathbb{E}\left[R_{n;i,j}^{(k)}\right]\right)\right\|_{1} & \leq  \left\|\frac{1}{n} \sum_{i \in \mathcal{I}_{n}}\sum_{j\in \mathcal{I}_{i}}w_{i,j}^{\ast}\left(R_{n;i, j}^{(k)}-\mathbb{E}\left[R_{n;i,j}^{(k)}\right]\right)\right\|_{2} \\
& = \frac{1}{n}\sqrt{\text{var}\left(\sum_{i \in \mathcal{I}_{n}}\sum_{j\in \mathcal{I}_{i}}w_{i,j}^{\ast}R_{n;i, j}^{(k)}\right)}\text{,}
\end{split}
\label{rijsd}
\end{equation}
 
\noindent where \eqref{rijsd} is an expression for the standard deviation of $\sum_{i \in \mathcal{I}_{n}}\sum_{j\in \mathcal{I}_{i}}w_{i,j}^{\ast}R_{n;i, j}^{(k)}$. Note that 

\begin{small}
\begin{equation}
\nonumber
\text{var}\left(\sum_{i \in \mathcal{I}_{n}}\sum_{j\in \mathcal{I}_{i}}w_{i,j}^{\ast}R_{n;i, j}^{(k)}\right) = \sum_{i \in \mathcal{I}_{n}}\text{var}\left(\sum_{j\in \mathcal{I}_{i}} w_{i,j}^{\ast}R_{n;i, j}^{(k)}\right) + \sum_{i\neq h \in \mathcal{I}_{n}}\text{cov}\left(\sum_{j\in \mathcal{I}_{i}}  w_{i,j}^{\ast}R_{n;i, j}^{(k)}, \sum_{s\in \mathcal{I}_{h}}  w_{h,s}^{\ast}R_{n; h, s}^{(k)}\right).
\end{equation}
\end{small}

\noindent The variance part of the previous equation can be further expressed as

\begin{align}
\label{vark}
\text{var}\left(\sum_{j\in \mathcal{I}_{i}} w_{i,j}^{\ast}R_{n;i, j}^{(k)}\right) & = \sum_{j\in \mathcal{I}_{i}}w_{i,j}^{\ast 2}\text{var}(R_{n;i,j}^{(k)}) + \sum_{j \neq s  \in \mathcal{I}_{i}}w_{i,j}^{\ast}w_{i,s}^{\ast}\text{cov}(R_{n;i,j}^{(k)}, R_{n;i,s}^{(k)}), \\
\nonumber
&\leq C\sum_{j\in \mathcal{I}_{i}}w_{i,j}^{\ast 2} + \sum_{j\in \mathcal{I}_{i}}\sum_{d\geq 1}\sum_{s\in \mathcal{P}_{n}(j,d)\cap \mathcal{I}_{i}}|\text{cov}(R_{n; i,j}^{(k)}R_{n; i,s}^{(k)})|, \\
\nonumber
& \leq C\sum_{j\in \mathcal{I}_{i}}w_{i,j}^{\ast 2}+\psi_{1,1}\left(\varphi_{k}, \varphi_{k}\right) \sum_{d \geq 1} \lambda_{n, d} \sum_{j \in \mathcal{I}_{i}}\left|\mathcal{P}_{n}(j , d)\right|,
\end{align}

\noindent where the second inequality follows from $w_{i,j}^{\ast},w_{i,s}^{\ast}\in[0,1]$. In the first term of the second inequality, $C$ represents any generic constant due to the fact that after the initial partition of $R_{n;i, j}$, the variance of $R_{n;i, j}^{(k)}$ is bounded. The last inequality follows from two reasons. First, from Lemma \ref{bounds} under Assumptions \ref{depas} and \ref{moments}, $|\text{cov}(R_{n;i,j}^{(k)}R_{n;i,s}^{(k)})| \leq \psi_{1,1}\left(\varphi_{k}, \varphi_{k}\right) \lambda_{n, d}$ for $d_{n}(i, j)=d$ and $\varphi_{k}$ is a bounded function with $\operatorname{Lip}\left(\psi_{k}\right)=1$. Second, the set of indexes $\mathcal{P}_{n}(j,d)$ is such that $\mathcal{P}_{n}(j,d)\cap \mathcal{I}_{i}\subset \mathcal{P}_{n}(j,d)$. The covariance component can be written as 

\begin{align}
\label{covark}
\text{cov}\left(\sum_{j\in \mathcal{I}_{i}}  w_{i,j}^{\ast}R_{n;i, j}^{(k)}, \sum_{s\in \mathcal{I}_{h}}  w_{h,s}^{\ast}R_{n;h, s}^{(k)}\right) &= \sum_{j\in \mathcal{I}_{i}}\sum_{s\in \mathcal{I}_{h}} w_{i,j}^{\ast}w_{h,s}^{\ast}\text{cov}(R_{n;i,j}^{(k)}, R_{n;h,s}^{(k)}), \\
\nonumber
& \leq \sum_{j\in \mathcal{I}_{i}}\sum_{d\geq 1}\sum_{s\in \mathcal{P}_{n}(j,d)\cap \mathcal{I}_{h}}|\text{cov}(R_{n;i,j}^{(k)}R_{n;h,s}^{(k)})|, \\
\nonumber
& \leq \psi_{1,1}\left(\varphi_{k}, \varphi_{k}\right) \sum_{d \geq 1} \lambda_{n, d} \sum_{j \in \mathcal{I}_{i}}\left|\mathcal{P}_{n}(j , d)\right|,
\end{align}

\noindent where the second and third inequalities follow from the same principles already discussed in the previous paragraph. It follows from Equations \eqref{vark} and \eqref{covark} that the total variance of $\sum_{j\in \mathcal{I}_{i}} w_{i,j}^{\ast}R_{n;i, j}^{(k)}$ can be bounded by

\begin{align}
\label{fvar}
\text{var}\left(\sum_{i \in \mathcal{I}_{n}}\sum_{j\in \mathcal{I}_{i}}w_{i,j}^{\ast}R_{n;i, j}^{(k)}\right)& = C\sum_{i \in \mathcal{I}_{n}}\sum_{j\in \mathcal{I}_{i}}w_{i,j}^{\ast 2} + 2\psi_{1,1}\left(\varphi_{k}, \varphi_{k}\right) \sum_{i \in \mathcal{I}_{n}}\sum_{d \geq 1} \lambda_{n, d} \sum_{j \in \mathcal{I}_{i}}\left|\mathcal{P}_{n}(j , d)\right|, \\
\nonumber
& = C\sum_{i \in \mathcal{I}_{n}}\sum_{j\in \mathcal{I}_{i}}w_{i,j}^{\ast 2} + 2\psi_{1,1}\left(\varphi_{k}, \varphi_{k}\right)\sum_{d \geq 1} \lambda_{n, d} \sum_{i \in \mathcal{I}_{n}}\left|\mathcal{P}_{n}(j , d)\right|, \\
\nonumber
& \leq n\left( C\bar{\mathcal{I}}_{n} + 2\psi_{1,1}\left(\varphi_{k}, \varphi_{k}\right)\sum_{d\geq 1}\bar{D}_{n}(d)\lambda_{n,d}\right),
\end{align}

\noindent where $\bar{\mathcal{I}}_{n}=n^{-1}\sum_{i\in\mathcal{I}_{n}}|\mathcal{I}_{i}|$ and the inequality follows because $w_{i,j}^{\ast 2}\in[0,1]$. The set $\mathcal{I}_{i}$ can either be empty, equal to the union of individual $i$'s connections in the networks $\mathcal{G}_{n}$ and $\mathcal{G}_{n,0}$, or equal to $\mathcal{P}_{n}(i , 1)$ (individual $i$'s connections in network $\mathcal{G}_{n}$). Note that, for any of the three cases, $|\mathcal{I}_{i}|\leq |\mathcal{P}_{n}(i,1)|$ for all $i$. Also, $\sum_{i\in\mathcal{I}_{n}}|\mathcal{P}_{n}(i,1)|\lambda_{n,1}\leq \sum_{d\geq 1}\bar{D}_{n}(d)\lambda_{n,d}$, which converges in probability to zero by Assumption \ref{decay}. It follows that $n^{-1}\bar{\mathcal{I}}_{n}\stackrel{\text { p }}{\longrightarrow}0$. Therefore,

\begin{equation}
\label{finite_var}
\left\|\frac{1}{n} \sum_{i \in \mathcal{I}_{n,m}}\sum_{j\in \mathcal{I}_{i}}w_{i,j}^{\ast}\left(R_{n;i, j}^{(k)}-\mathbb{E}\left[R_{n;i,j}^{(k)}\right]\right)\right\|_{1} \leq  \left( n^{-1}C\bar{\mathcal{I}}_{n} + 2\psi_{1,1}n^{-1}\sum_{d\geq 1}\bar{D}_{n}(d)\lambda_{n,d}\right)^{1/2}.
\end{equation}

\noindent The result follows from $n^{-1}\bar{\mathcal{I}}_{n}\stackrel{\text { p }}{\longrightarrow}0$ and $n^{-1}\sum_{d\geq 1}\bar{D}_{n}(d)\lambda_{n,d}\stackrel{\text { p }}{\longrightarrow}0$ under Assumption \ref{decay}.\hfill\end{proof}

\begin{corollary}[LLN for Instruments and Regressors]
\label{int_reg}
Let Assumptions \ref{depas} to \ref{epmoments} hold. Then,
\begin{equation*}
\left\|\frac{1}{n}\sum_{i\in\mathcal{I}_{n}}\left(\mathbf{z}_{n;i}\mathbf{d}_{n;i}^{\top} - \mathbb{E}[\mathbf{z}_{n;i}\mathbf{d}_{n;i}^{\top}]\right)\right\|_{1} {\longrightarrow} 0.
\end{equation*}
\end{corollary}

\begin{proof}
There are four different types of components in the matrix $\mathbf{Z}_{n}^{\top}\mathbf{D}_{n}$ formed by summation of products of: (1) Non-network regressors of the form $x_{n;i,q}x_{n;i,\ell}$; (2) Network regressors of the form $\mathbf{w}_{n, 0; i} \mathbf{x}_{n, q}\mathbf{w}_{n;i} \mathbf{x}_{n;\ell}$; (3) Network and non-network regressors of the form $\mathbf{w}_{n, 0; i} \mathbf{x}_{n, q}\mathbf{x}_{n;i,\ell}$; and (4) Network regressors and network outcomes of the form $\mathbf{w}_{n,0;i}\mathbf{x}_{n,q}\mathbf{w}_{n;i}\mathbf{y}_{n}$ [and the versions of (2) and (3) with $\mathbf{w}_{n, 0; i}^{p}$ instead of $\mathbf{w}_{n, 0; i}$]. The LLN follows from Lemma \ref{lln_main} by choosing $\mathcal{I}_{i}=\{\emptyset\}$ for (1), $\mathcal{I}_{i}$ as the union of individual $i$'s connections in the networks $\mathcal{G}_{n}$ and $\mathcal{G}_{n,0}$ in (2), and $\mathcal{I}_{i}=\mathcal{P}_{n}(i , 1)$ for (3). For (4), note that 

\begin{equation}
\label{rexp}
\mathbb{E}[\mathbf{W}_{N} \mathbf{y}]= \gamma_{0} \mathbf{W}_{N} \mathbb{E}[\mathbf{x}_{N}]+(\gamma_{0} \beta_{0}+\delta_{0}) \sum_{p=0}^{\infty} \beta_{0}^{k} \mathbf{W}_{N}^{p+2} \mathbb{E}[\mathbf{x}_{N}],
\end{equation}

\noindent where, again, the expectation is taken conditional on $\mathcal{G}_{N}$. By choosing $\mathcal{I}_{i}$ to be the union of individual $i$'s connections in the network $\mathcal{G}$ and the set of individuals at distance $p$ from $i$ (for all $p\in \mathbb{R}_{+}$), Lemma \ref{lln_main} applies for all values in the infinite sum formed by $\mathbf{w}_{n,0;i}\mathbf{x}_{n,q}\mathbf{w}_{n;i}\mathbf{y}_{n}$ after replacing $\mathbf{w}_{n;i}\mathbf{y}_{n}$ from Equation \eqref{rexp} [the same argument holds for (2) and (3) when using $\mathbf{w}_{n, 0; i}^{p}$ instead of $\mathbf{w}_{n, 0; i}$]. Given that each component of the sum converges to a finite expectation, the infinite sum of finite expectations is also finite given the restriction on the parameters $\beta_{0}$ from Assumption \ref{excrest}, thus completing the proof. \hfill\end{proof}

\begin{corollary}[LLN for Instruments and Errors] 
\label{int_err}
Let Assumptions \ref{depas} to \ref{epmoments} hold, then

\begin{equation*}
\left\|\frac{1}{n}\sum_{i\in\mathcal{I}_{n}}\left(\mathbf{z}_{n;i}\varepsilon_{n;i}^{\top} - \mathbb{E}[\mathbf{z}_{n;i}\varepsilon_{n;i}^{\top}]\right)\right\|_{1} {\longrightarrow} 0.
\end{equation*}

\end{corollary}

\begin{proof}
Given that $\mathbf{r}_{n;i}=[\mathbf{x}_{n;i}, \varepsilon_{n;i}]$ and $\mathbf{z}_{n;i}$ can be divided into both network and nonnetwork components, the proof of this result is analogous to that of Corollary \ref{int_reg} (1) and (3). \hfill\end{proof}

\begin{corollary}[Finite Variance]
\label{finitevar}
Define $\mathbf{S}_{n} = \mathbf{Z}_{n}^{\top}\boldsymbol{\varepsilon}_{n}$ and $\mathbf{\Omega}_{n} = \emph{\text{var}}(n^{-1/2}\mathbf{S}_{n})$ and let Assumptions \ref{depas} to \ref{epmoments} hold, then as $n \to \infty$, $\mathbf{\Omega}_{n} \longrightarrow \mathbf{\Omega}_{N} <\infty$.
\end{corollary}

\begin{proof}
As before, $n^{-1/2}\mathbf{S}_{n} \equiv n^{-1/2}\sum_{i=1}^{n} \mathbf{z}_{n;i}\varepsilon_{n;i}$. The bounded covariance assumptions of Lemma \ref{bounds} combined with the arguments of Lemma \ref{lln_main} guarantee that the following limit $\lim_{n \to\infty} n^{-1} \text{var}\left(\sum_{i=1}^{n} \mathbf{z}_{n;i}\varepsilon_{n;i}\right)$ is finite. In particular, from Equation \eqref{finite_var}, using the appropriate values for $R_{n;i,j}$ and $\mathcal{I}_{i}$ (see Corollary \ref{int_reg}), it follows that $\text{var}(\sum_{i=1}^{n} \mathbf{z}_{n;i}\varepsilon_{n;i}) = O_{p}(1)$. Given that $\mathbf{\Omega}_{n}$ converges to a finite quantity, it follows that $\mathbf{\Omega}_{n} {\longrightarrow} \mathbf{\Omega}_{N}$, where

\begin{equation*}
\mathbf{\Omega}_{N} = \lim_{n \to\infty} n^{-1}\left[\sum_{i=1}^{n} \text{var}(\mathbf{z}_{n;i}\varepsilon_{n;i}) + \sum_{i \neq j}\text{cov}(\mathbf{z}_{n;i}\varepsilon_{n;i}, \mathbf{z}_{n;j}\varepsilon_{n;j})\right] < \infty.
\end{equation*} \hfill\end{proof}

\begin{lemma}[Central Limit Theorem]
\label{clt}
Let Assumptions \ref{id} and \ref{depas}-\ref{av_sparsity} hold and define $S_{n}\equiv\sum_{i \in \mathcal{I}_{n}} z_{n;i,q}\varepsilon_{n;i}$, where $z_{n;i,q}$ is the $q$\emph{th} entrance of the vector $\mathbf{z}_{n;i}$. Then, by the definition of $\mathbf{z}_{n;i}$ and Assumption \ref{id}, $\mathbb{E}[z_{n;i,q}\varepsilon_{n;i}]=0$. As $n\to\infty$, 

$$
\sup _{t \in \mathbf{R}}\left|\mathbf{P}\left\{\frac{S_{n}}{\sigma_{n}} \leq t\right\}-\Phi(t)\right| {\longrightarrow}0,
$$
 
\noindent where $\sigma_{n}\equiv\text{var}(S_{n})$ and $\Phi(\cdot)$ denotes the cumulative distribution function of a standard normal random variable.
\end{lemma}

\begin{proof}
Let $Y_{n;i}=z_{n;i,q}\varepsilon_{n;i}$. From Lemma \ref{bounds}, the covariance of any two $Y_{n;i}$ and $Y_{n;j}$ is bounded. The proof then follows from applying the unconditional version of Lemmas A.2 and A.3 in \citet[][Appendix A, pp. 899-907]{Kojevnikov2020} to $Y_{n;i}$ and $S_{n}/\sigma_{n}$, respectively.   \hfill\end{proof}

\begin{lemma}[Multivariate Central Limit Theorem]
\label{mclt}
Let Assumptions \ref{id} and \ref{depas}-\ref{av_sparsity} hold. Then, as $n\to \infty$, $n^{-1/2}\sum_{i=1}^{n} \mathbf{z}_{n;i}\varepsilon_{n;i} \stackrel{\text { d }}{\longrightarrow} \mathcal{N}(0,\Omega_{N})$.
\end{lemma}

\begin{proof}
From Lemma \ref{clt}, it follows that $n^{-1/2}\sum_{i=1}^{n}z_{n;i,q}\varepsilon_{n;i}\stackrel{d}{\longrightarrow} \mathcal{N}(0,\sigma_{n}^{2})$, while from Lemma \ref{finitevar}, it follows that $\mathbf{\Omega}_{N}$ exists. Therefore, the result follows from an application of the Cram\'{e}r-Wold device.  \hfill\end{proof}

%% file: supp_C.tex
\subsection{Data Description }\label{descrip}

Our data set was collected between March and May 2011 as part of the Hong Kong Secondary Education Survey in Hong Kong (SESHK). The survey was conducted in the second semester before the final exams and involved three secondary schools with 868 students participating. The sample includes 7th-grade students from all three schools and 8th- and 9th-grade students from one school (g $\in\{7,8,9\}$). Each grade within a school is made up of five different sections (cl $\in\{1,\ldots,5\}$).

Table \ref{tables:table1} shows the summary statistics for the variables we use in our empirical application. \texttt{Math test} corresponds to the first math exam score for each student $i$. The data set also includes information on a cognitive ability test on five personality measures: \texttt{Agreeableness}, \texttt{Conscientiousness}, \texttt{Extraversion}, \texttt{Neuroticism} and \texttt{Openness}.\footnote{The numbers in parenthesis in Scale column of Table \ref{tables:table1} represent the scales for the tests.} In our empirical application, we also include the following variables: \texttt{Male} that equals 1 if the student is male, and 0 otherwise. The \texttt{Height} for each student is measured in centimeters (cm) and the \texttt{Weight} in kilograms (kg). Both the \texttt{number of elder and younger siblings} are count variables, and \texttt{ commute by car / taxi} equals 1 when a student goes to school by car or by taxi. 

We also include indicator variables capturing students' engagement at school. For example, the indicator variables \texttt{Siblings' Help} and \texttt{Parents' Help} take the value of 1 if students receive help from their siblings or parents. We label those variables as 0 otherwise. To capture extracurricular school-related activities, we include the indicator variable \texttt{Music}, which equals 1 if students play music, and 0 otherwise.

\subsection*{Seatmate and Student Partner Networks}

In the survey, students were asked to write lists of up to ten peers from among their schoolmates within the same grade with whom they discussed their problems with schoolwork and who sat next to them in class during the first semester. We use this information to build the study partner and seatmate networks using the following reciprocal peer rule: If students $i$ and $j$ named each other as study partners in the survey, we record an edge in the study partner network. We follow the same process for the seatmate network. Table \ref{tables:table2} reports the summary statistics of the network among all students by school.

Seat assignments in the classrooms change several times over a semester, and the changes are decided by the class teacher. Unlike study partners, the seatmate network is based on proximity and imposed by the school. Therefore, the seatmate network is an excellent candidate to be used as the instrumental network $\mathbf{W}_{n,0}$, in our analysis. On the other hand, students can freely choose with whom they study and this decision could be based on unobservable characteristics that also affect exam performance, making the study partner network, $\mathbf{W}_{n}$, likely endogenous. 
                                               
As mentioned in Section \ref{emp} of the main manuscript, we follow the Linear-in-Means model for social effects. To improve readability, we rewrite the equation \ref{emp_eq1} below, 

\begin{eqnarray*}
\texttt{math}_{i,\text{s}\times\text{g}\times\text{cl}} &=& \alpha + \beta \sum_{j\neq i}^n w_{n;i,j}\texttt{math}_{j,\text{s}\times\text{g}\times\text{cl}} \\ \nonumber
       & & {} +  \sum_{j\neq i}^n w_{n;i,j}\texttt{characteristics}_{j,\text{s}\times\text{g}\times\text{cl}}^{\prime}\delta_{\texttt{characteristics}} \\ \nonumber
       & & {} +  \sum_{j\neq i}^n w_{n;i,j}\texttt{personality}_{j,\text{s}\times\text{g}\times\text{cl}}^{\prime}\delta_{\texttt{personality}} \\ \nonumber
       & & {} + \texttt{characteristics}_{i,\text{s}\times\text{g}\times\text{cl}}^{\prime}\gamma_{\texttt{characteristics}}
       +\texttt{personality}_{i,\text{s}\times\text{g}\times\text{cl}}^{\prime}\gamma_{\texttt{personality}} \\ \nonumber
       & & {} + \sum_{\text{s}=1}^{3}\sum_{\text{g}=7}^{9}\sum_{\text{cl}=1}^{5}f_{\text{s}\times\text{g}\times\text{cl}}\times\mathbb{I}\{i\in{\text{s}\times\text{g}\times\text{cl}}\} + \epsilon_i\text{,} 
\end{eqnarray*}

Estimations results are shown in Table \ref{tables:table3}, in Section \ref{emp} in the main text, and Table \ref{tables:table4} (in this section). As mentioned in the manuscript,  estimators that do not control for network endogeneity tend to underestimate peer effects.  In the context of our simulation results, one possible explanation for the negative bias is the existence of unobserved homophily. If the unobserved variables driving the choice of study partners are negatively correlated with the outcome, we expect the endogeneity bias to underestimate the actual peer effects value. Students can choose to study with others they find fun for reasons other than learning the test material. If students select study partners that can distract them from schoolwork, estimators that do not take that sorting process into account can be downward biased. These results suggest that policies that strengthen collaboration between students within and outside the classroom can generate benefits that have the potential to generate positive social multipliers. All results are qualitatively robust to different choices of $p$ and $D_n$; see Section \ref{supp_estimates_output} below.

\subsection{Supplementary Estimation Results}\label{supp_estimates_output}

For robustness purposes, Tables \ref{tables:table15}-\ref{tables:table20} show the empirical estimation of model (6.1) with the kernel Tukey-Hanning, constant $C \in\{1.5,1.6,1.7\}$, and $p \in \{3,4,5\}$.

\subsection{Assessing Assumptions}\label{validate_test}

To validate the Assumption \ref{excrest}, we perform a Least Square (LS) regression that includes all variables in equation \ref{emp_eq1} and our proposed instruments.  Namely, we estimate the following equation.

\begin{eqnarray}\label{eqn_exrest}
   \texttt{math} & = & \beta\textbf{W}_n\texttt{math}\text{ }  + \text{ }\texttt{personality}\text{ }\boldsymbol{\gamma}_{p} + \text{ }\texttt{characteristics}\text{ }\boldsymbol{\gamma}_{ch}\text{ } \\ \nonumber
    & &{}+{}\textbf{W}_n\text{ }\texttt{characteristics}\text{ }\boldsymbol{\delta} + \textbf{W}_{n,0}\text{ }\texttt{characteristics}\text{ }\boldsymbol{\delta}_1\text{ }\\ \nonumber 
    & & {}+{}\textbf{W}^2_{n,0}\text{ }\texttt{characteristics}\text{ }\boldsymbol{\delta}_{2}+\text{error,}
\end{eqnarray}

\noindent where $\texttt{characteristics}$, $\texttt{personality}$, and the adjacency matrices $\textbf{W}_n$, $\textbf{W}_{n,0}$ are defined above. The estimation results for equation \eqref{eqn_exrest} are shown in Table \ref{tables:table14}. They suggest no correlation between the output variable $\texttt{math}$ and the vectors $\textbf{W}_{n,0}\text{ }\texttt{characteristics}$ and $\textbf{W}^2_{n,0}\text{ }\texttt{characteristics}$, which are the proposed instruments. The estimated coefficients $\boldsymbol{\delta}_1$ and $\boldsymbol{\delta}_2$ are statistically insignificant\footnote{The estimated coefficient of $\textbf{W}^2_{n,0}$$\text{ln(Weight)}$  is significant at 10\%.}, which allow us to use the exogenous variation embodied in $\textbf{W}_{n,0}$ to identify the parameters of the linear model \eqref{emp_eq1}. 

To validate Assumption \ref{relevance}, we run a series of LS regressions in which the outcome is the endogenous variable $\textbf{W}_n\texttt{math}$. Namely, we estimate the following models.

\begin{eqnarray}\label{eqn_WY1}
\textbf{W}_n\texttt{math} &=& \textbf{W}_{n,0}\text{ }\texttt{characteristics}\text{ }\boldsymbol{\delta}+\textbf{W}_{n,0}\text{ }\texttt{personality}\text{ }\boldsymbol{\gamma}+\text{error,}\\
\textbf{W}_n\texttt{math} &=& \textbf{W}^2_{n,0}\text{ }\texttt{characteristics}\text{ }\boldsymbol{\delta}+\textbf{W}^2_{n,0}\text{ }\texttt{personality}\text{ }\boldsymbol{\gamma}+\text{error,}\\
\textbf{W}_n\texttt{math} &=& \textbf{W}_{n,0}\text{ }\texttt{characteristics}\text{ }\boldsymbol{\delta}+\textbf{W}_{n,0}\text{ }\texttt{personality}\text{ }\boldsymbol{\gamma}\\ \nonumber 
& & {}+\textbf{W}^2_{n,0}\text{ }\texttt{characteristics}\text{ }\boldsymbol{\delta}+\textbf{W}^2_{n,0}\text{ }\texttt{personality}\text{ }\boldsymbol{\gamma}+\text{error.}
\end{eqnarray}

The estimation results for these specifications are shown in Table \ref{tables:table5}. $F$-statistics suggest that our proposed instruments are relevant to describe the endogenous variable. Furthermore, we also perform LS regressions in which the dependent variable $\textbf{W}_n{\textbf{x}}$ is the average of characteristic s$\textbf{x}$ between the study partners, which can be any of the characteristics mentioned above, for example height, weight, siblings help, parents help, commute to school by car or taxi; playing music; and whether the student is male. 

\begin{eqnarray}\label{eqn_WY2}
\textbf{W}_n\textbf{x} &=& \textbf{W}_{n,0}\text{ }\texttt{characteristics}_{-\textbf{x}}\text{ }\boldsymbol{\delta}+\textbf{W}_{n,0}\text{ }\texttt{personality}\text{ }\boldsymbol{\gamma}+\text{error,}\\
\textbf{W}_n\textbf{x} &=& \textbf{W}^2_{n,0}\text{ }\texttt{characteristics}_{-\textbf{x}}\text{ }\boldsymbol{\delta}+\textbf{W}^2_{n,0}\text{ }\texttt{personality}\text{ }\boldsymbol{\gamma}+\text{error,}\\
\textbf{W}_n\textbf{x} &=& \textbf{W}_{n,0}\text{ }\texttt{characteristics}_{-\textbf{x}}\text{ }\boldsymbol{\delta}+\textbf{W}_{n,0}\text{ }\texttt{personality}\text{ }\boldsymbol{\gamma}+\text{error,}\\ \nonumber 
& & {}+\textbf{W}^2_{n,0}\text{ }\texttt{characteristics}_{-\textbf{x}}\text{ }\boldsymbol{\delta}+\textbf{W}^2_{n,0}\text{ }\texttt{personality}\text{ }\boldsymbol{\gamma}+\text{error.}
\end{eqnarray}

Here, $\texttt{characteristics}_{-\textbf{x}}$ means that all characteristics have been included except $\textbf{x}$, used in $\textbf{W}_n\textbf{x}$. Tables \ref{tables:table7}-\ref{tables:table13} show OLS estimates, and $F$-statistics also suggest that our proposed instruments are relevant to describe endogenous variables $\textbf{W}_n\textbf{x}$.

Finally, we present a network architecture that integrates students based on study partnerships and shared seating arrangements.  Specifically, we establish a new set of connections, denoted $\textbf{W}_{n,*}$, between students $i$ and $k$ following a defined rule: If students $i$ and $j$ mutually identified each other as study partners in the survey and students $j$ and $k$ reciprocated as seatmates, we establish an edge between $i$ and $k$. This connection is formed when neither students $i$ and $k$ are study partners, nor are students $j$ and $k$. Table \ref{tables:table3} reports the summary statistics of this new network in the column named Extra. Based on this new set of connections, we rewrite the previous equations as follows.

\begin{eqnarray}\label{eqn_WY3}
\textbf{W}_n\texttt{math} &=& \textbf{W}_{n,*}\text{ }\texttt{characteristics}\text{ }\boldsymbol{\delta}+\textbf{W}_{n,*}\text{ }\texttt{personality}\text{ }\boldsymbol{\gamma}+\text{error,}\\
\textbf{W}_n\texttt{math} &=& \textbf{W}^2_{n,*}\text{ }\texttt{characteristics}\text{ }\boldsymbol{\delta}+\textbf{W}^2_{n,*}\text{ }\texttt{personality}\text{ }\boldsymbol{\gamma}+\text{error,}\\
\textbf{W}_n\texttt{math} &=& \textbf{W}_{n,*}\text{ }\texttt{characteristics}\text{ }\boldsymbol{\delta}+\textbf{W}_{n,*}\text{ }\texttt{personality}\text{ }\boldsymbol{\gamma}+\text{error,}\\ \nonumber 
& & {}+\textbf{W}^2_{n,*}\text{ }\texttt{characteristics}\text{ }\boldsymbol{\delta}+\textbf{W}^2_{n,*}\text{ }\texttt{personality}\text{ }\boldsymbol{\gamma}+\text{error.}
\end{eqnarray}

The LS estimates for all these three specifications are shown in Table \ref{tables:table6}.  $F$-Statistics suggest that our proposed instruments are relevant to describe the endogenous variable.

\begingroup
\setstretch{1.8}
\begin{landscape}
\begin{table}[h!]
\vspace{-6.5em}
\setlength{\tabcolsep}{3pt} 
\centering
    \begin{threeparttable}
        \caption{Summary Statistics}
        \vspace{0.1em}
        \small
        \input{tables/tables_rev/table1}
        \label{tables:table1}
            \begin{tablenotes}
               \linespread{1}\footnotesize
                \item Note: Descriptive statistics such as sample mean (Mean), standard deviation (SD), minimum (Min), maximum (Max) and sample size ($n$) are presented here for all variables and each school. Course grades, personality trait measures, and cognitive ability tests are scored on the scale indicated.
            \end{tablenotes}
    \end{threeparttable}
    
\vspace{-1.5em}

    \begin{threeparttable}
        \caption{Summary Network Statistics}
        \vspace{0.1em}
        \small
        \input{tables/tables_rev/table2}
        \label{tables:table2}
            \begin{tablenotes}
                \linespread{1}\footnotesize
                \item Note: The degree is multiplied by 100 to increase the scale.
            \end{tablenotes}
    \end{threeparttable}
\end{table}
\end{landscape}
\endgroup

\newpage

\setstretch{1.8}

\newpage
\begingroup
\setstretch{1.8}
\begin{landscape}
\begin{table}[h!]
 \vspace{-5.5em}
\def\arraystretch{1.00}
\setlength{\tabcolsep}{15pt} 
\centering
    \begin{threeparttable}
        \caption{Estimations results, cont.} 
        \vspace{0.1em}
        \footnotesize
        \input{tables//tables_rev/table4}
        \label{tables:table4}
        \begin{tablenotes}[para,flushleft]
            \footnotesize
            \raggedright
            \item{Note:} \text{*} \(p<0.10\), \text{**}   \(p<0.05\), \text{***}  \(p<0.01\)
        \end{tablenotes}
    \end{threeparttable}
\end{table}
\end{landscape}
\endgroup

\newpage
\begingroup
\setstretch{1.8}

\begin{table}[h!]
 \vspace{-5.5em}
\def\arraystretch{1.00}
\setlength{\tabcolsep}{15pt} 
\centering
    \begin{threeparttable}
        \caption{Estimations results with $p=3$} 
        \vspace{0.1em}
        \footnotesize
        \input{tables//tables_rev/table15}
        \label{tables:table15}
        \begin{tablenotes}[para,flushleft]
            \footnotesize
            \raggedright
            \item{Note:} (i) \text{*} \(p<0.10\), \text{**}   \(p<0.05\), \text{***}  \(p<0.01\); (ii) Standard errors are in parentheses. (iii) $\dagger$ These regressors are measured as the deviation of students' personality from their peers' average.
        \end{tablenotes}
    \end{threeparttable}
\end{table}
\endgroup

\newpage
\begingroup
\setstretch{1.8}
\begin{landscape}
\begin{table}[h!]
 \vspace{-5.5em}
\def\arraystretch{1.00}
\setlength{\tabcolsep}{15pt} 
\centering
    \begin{threeparttable}
        \caption{Estimations results with $p=3$, cont.} 
        \vspace{0.1em}
        \footnotesize
        \input{tables//tables_rev/table16}
        \label{tables:table16}
        \begin{tablenotes}[para,flushleft]
            \footnotesize
            \raggedright
            \item{Note:} \text{*} \(p<0.10\), \text{**}   \(p<0.05\), \text{***}  \(p<0.01\)
        \end{tablenotes}
    \end{threeparttable}
\end{table}
\end{landscape}
\endgroup

\newpage
\begingroup
\setstretch{1.8}

\begin{table}[h!]
 \vspace{-5.5em}
\def\arraystretch{1.00}
\setlength{\tabcolsep}{15pt} 
\centering
    \begin{threeparttable}
        \caption{Estimations results with $p=4$} 
        \vspace{0.1em}
        \footnotesize
        \input{tables//tables_rev/table17}
        \label{tables:table17}
        \begin{tablenotes}[para,flushleft]
            \footnotesize
            \raggedright
            \item{Note:} (i) \text{*} \(p<0.10\), \text{**}   \(p<0.05\), \text{***}  \(p<0.01\); (ii) Standard errors are in parentheses. (iii) $\dagger$ These regressors are measured as the deviation of students' personality from their peers' average.
        \end{tablenotes}
    \end{threeparttable}
\end{table}
\endgroup

\newpage
\begingroup
\setstretch{1.8}
\begin{landscape}
\begin{table}[h!]
 \vspace{-5.5em}
\def\arraystretch{1.00}
\setlength{\tabcolsep}{15pt} 
\centering
    \begin{threeparttable}
        \caption{Estimations results with $p=4$, cont.} 
        \vspace{0.1em}
        \footnotesize
        \input{tables//tables_rev/table18}
        \label{tables:table18}
        \begin{tablenotes}[para,flushleft]
            \footnotesize
            \raggedright
            \item{Note:} \text{*} \(p<0.10\), \text{**}   \(p<0.05\), \text{***}  \(p<0.01\)
        \end{tablenotes}
    \end{threeparttable}
\end{table}
\end{landscape}
\endgroup

\newpage
\begingroup
\setstretch{1.8}

\begin{table}[h!]
 \vspace{-5.5em}
\def\arraystretch{1.00}
\setlength{\tabcolsep}{15pt} 
\centering
    \begin{threeparttable}
        \caption{Estimations results with $p=5$} 
        \vspace{0.1em}
        \footnotesize
        \input{tables//tables_rev/table19}
        \label{tables:table19}
        \begin{tablenotes}[para,flushleft]
            \footnotesize
            \raggedright
            \item{Note:} (i) \text{*} \(p<0.10\), \text{**}   \(p<0.05\), \text{***}  \(p<0.01\); (ii) Standard errors are in parentheses. (iii) $\dagger$ These regressors are measured as the deviation of students' personality from their peers' average.
        \end{tablenotes}
    \end{threeparttable}
\end{table}
\endgroup

\newpage
\begingroup
\setstretch{1.8}
\begin{landscape}
\begin{table}[h!]
 \vspace{-5.5em}
\def\arraystretch{1.00}
\setlength{\tabcolsep}{15pt} 
\centering
    \begin{threeparttable}
        \caption{Estimations results with $p=5$, cont.} 
        \vspace{0.1em}
        \footnotesize
        \input{tables//tables_rev/table20}
        \label{tables:table20}
        \begin{tablenotes}[para,flushleft]
            \footnotesize
            \raggedright
            \item{Note:} \text{*} \(p<0.10\), \text{**}   \(p<0.05\), \text{***}  \(p<0.01\)
        \end{tablenotes}
    \end{threeparttable}
\end{table}
\end{landscape}
\endgroup

\newpage
\begingroup
\setstretch{1.8}
\begin{landscape}
\begin{table}[h!]
\label{exclusion_ev}
 \vspace{-5.5em}
\def\arraystretch{1}
\setlength{\tabcolsep}{15pt} 
\centering
    \begin{threeparttable}
        \caption{LS Estimation results of the model $y = \beta\mathbf{W}y + \boldsymbol{\delta}X + \boldsymbol{\gamma}\mathbf{W}X + \boldsymbol{\theta}_1\mathbf{W}_0 X + \boldsymbol{\theta}_2\mathbf{W}_0^2 X$}
        \vspace{0.1em}
        \footnotesize
        \input{tables//tables_rev/table14}
        \label{tables:table14}
        \begin{tablenotes}[para,flushleft]
            \footnotesize
            \raggedright
            \item{Note:} \text{*} \(p<0.10\), \text{**}   \(p<0.05\), \text{***}  \(p<0.01\)
        \end{tablenotes}
    \end{threeparttable}
\end{table}
\end{landscape}
\endgroup

\newpage
\begingroup
\setstretch{1.8}
\begin{landscape}
\begin{table}[h!]
 \vspace{-5.5em}
\def\arraystretch{1.25}
\setlength{\tabcolsep}{15pt} 
\centering
    \begin{threeparttable}
        \caption{LS Estimations using $\mathbf{W}y$ as dependent variable} 
        \vspace{0.1em}
        \footnotesize
        \input{tables//tables_rev/table5}
        \label{tables:table5}
        \begin{tablenotes}[para,flushleft]
            \footnotesize
            \raggedright
            \item{Note:} \text{*} \(p<0.10\), \text{**}   \(p<0.05\), \text{***}  \(p<0.01\)
        \end{tablenotes}
    \end{threeparttable}
\end{table}
\end{landscape}
\endgroup

\newpage
\begingroup
\setstretch{1.8}
\begin{landscape}
\begin{table}[h!]
 \vspace{-5.5em}
\def\arraystretch{1.25}
\setlength{\tabcolsep}{15pt} 
\centering
    \begin{threeparttable}
        \caption{LS Estimations using $\textbf{W}\times\text{Male}$ as dependent variable}
        \vspace{0.1em}
        \footnotesize
        \input{tables//tables_rev/table7}
        \label{tables:table7}
        \begin{tablenotes}[para,flushleft]
            \footnotesize
            \raggedright
            \item{Note:} \text{*} \(p<0.10\), \text{**}   \(p<0.05\), \text{***}  \(p<0.01\)
        \end{tablenotes}
    \end{threeparttable}
\end{table}
\end{landscape}
\endgroup

\newpage
\begingroup
\setstretch{1.8}
\begin{landscape}
\begin{table}[h!]
 \vspace{-5.5em}
\def\arraystretch{1.25}
\setlength{\tabcolsep}{15pt} 
\centering
    \begin{threeparttable}
        \caption{LS Estimations using $\textbf{W}\times\text{ln(Height)}$ as dependent variable}
        \vspace{0.1em}
        \footnotesize
        \input{tables//tables_rev/table8}
        \label{tables:table8}
        \begin{tablenotes}[para,flushleft]
            \footnotesize
            \raggedright
            \item{Note:} \text{*} \(p<0.10\), \text{**}   \(p<0.05\), \text{***}  \(p<0.01\)
        \end{tablenotes}
    \end{threeparttable}
\end{table}
\end{landscape}
\endgroup

\newpage
\begingroup
\setstretch{1.8}
\begin{landscape}
\begin{table}[h!]
 \vspace{-5.5em}
\def\arraystretch{1.25}
\setlength{\tabcolsep}{15pt} 
\centering
    \begin{threeparttable}
        \caption{LS Estimations using $\textbf{W}\times\text{ln(Weight)}$ as dependent variable}
        \vspace{0.1em}
        \footnotesize
        \input{tables//tables_rev/table9}
        \label{tables:table9}
        \begin{tablenotes}[para,flushleft]
            \footnotesize
            \raggedright
            \item{Note:} \text{*} \(p<0.10\), \text{**}   \(p<0.05\), \text{***}  \(p<0.01\)
        \end{tablenotes}
    \end{threeparttable}
\end{table}
\end{landscape}
\endgroup

\newpage
\begingroup
\setstretch{1.8}
\begin{landscape}
\begin{table}[h!]
 \vspace{-5.5em}
\def\arraystretch{1.25}
\setlength{\tabcolsep}{15pt} 
\centering
    \begin{threeparttable}
        \caption{LS Estimations using $\textbf{W}\times\text{Siblings Help}$ as dependent variable}
        \vspace{0.1em}
        \footnotesize
        \input{tables//tables_rev/table10}
        \label{tables:table10}
        \begin{tablenotes}[para,flushleft]
            \footnotesize
            \raggedright
            \item{Note:} \text{*} \(p<0.10\), \text{**}   \(p<0.05\), \text{***}  \(p<0.01\)
        \end{tablenotes}
    \end{threeparttable}
\end{table}
\end{landscape}
\endgroup

\newpage
\begingroup
\setstretch{1.8}
\begin{landscape}
\begin{table}[h!]
 \vspace{-5.5em}
\def\arraystretch{1.25}
\setlength{\tabcolsep}{15pt} 
\centering
    \begin{threeparttable}
        \caption{LS Estimations using $\textbf{W}\times\text{Parents Help}$ as dependent variable}
        \vspace{0.1em}
        \footnotesize
        \input{tables//tables_rev/table11}
        \label{tables:table11}
        \begin{tablenotes}[para,flushleft]
            \footnotesize
            \raggedright
            \item{Note:} \text{*} \(p<0.10\), \text{**}   \(p<0.05\), \text{***}  \(p<0.01\)
        \end{tablenotes}
    \end{threeparttable}
\end{table}
\end{landscape}
\endgroup

\newpage
\begingroup
\setstretch{1.8}
\begin{landscape}
\begin{table}[h!]
 \vspace{-5.5em}
\def\arraystretch{1.25}
\setlength{\tabcolsep}{15pt} 
\centering
    \begin{threeparttable}
        \caption{LS Estimations using $\textbf{W}\times\text{Commute by Car/Taxi}$ as dependent variable}
        \vspace{0.1em}
        \footnotesize
        \input{tables//tables_rev/table12}
        \label{tables:table12}
        \begin{tablenotes}[para,flushleft]
            \footnotesize
            \raggedright
            \item{Note:} \text{*} \(p<0.10\), \text{**}   \(p<0.05\), \text{***}  \(p<0.01\)
        \end{tablenotes}
    \end{threeparttable}
\end{table}
\end{landscape}
\endgroup

\newpage
\begingroup
\setstretch{1.8}
\begin{landscape}
\begin{table}[h!]
 \vspace{-5.5em}
\def\arraystretch{1.25}
\setlength{\tabcolsep}{15pt} 
\centering
    \begin{threeparttable}
        \caption{LS Estimations using $\textbf{W}\times\text{Music}$ as dependent variable}
        \vspace{0.1em}
        \footnotesize
        \input{tables//tables_rev/table13}
        \label{tables:table13}
        \begin{tablenotes}[para,flushleft]
            \footnotesize
            \raggedright
            \item{Note:} \text{*} \(p<0.10\), \text{**}   \(p<0.05\), \text{***}  \(p<0.01\)
        \end{tablenotes}
    \end{threeparttable}
\end{table}
\end{landscape}
\endgroup

\newpage
\begingroup
\setstretch{1.8}
\begin{landscape}
\begin{table}[h!]
 \vspace{-5.5em}
\def\arraystretch{1.25}
\setlength{\tabcolsep}{15pt} 
\centering
    \begin{threeparttable}
        \caption{LS Estimations using $\mathbf{W}y$ as dependent variable}
        \vspace{0.1em}
        \footnotesize
        \input{tables//tables_rev/table6}
        \label{tables:table6}
        \begin{tablenotes}[para,flushleft]
            \footnotesize
            \raggedright
            \item{Note:} \text{*} \(p<0.10\), \text{**}   \(p<0.05\), \text{***}  \(p<0.01\)
        \end{tablenotes}
    \end{threeparttable}
\end{table}
\end{landscape}
\endgroup

%% file: tables/tables_rev/table1.tex
\begin{tabular}{lcccccccccccccccc}
\toprule
                               Variables &   Scale &  \multicolumn{1}{c}{} &   Mean &    SD &    Min &    Max &  \multicolumn{1}{l}{} &   Mean &    SD &    Min &    Max &  \multicolumn{1}{r}{} &   Mean &    SD &    Min &    Max \\
\midrule
      \textbf{Student-related variables} &         &                       &        &       &        &        &                       &        &       &        &        &                       &        &       &        &        \\
                  \hspace{3mm} Math Test & [0,100] &                       &  61.88 & 13.56 &  33.57 &  92.33 &                       &  61.44 & 13.12 &  31.00 &  92.00 &                       &  68.38 & 14.01 &  24.35 & 100.00 \\
                       \hspace{3mm} Male &         &                       &   0.37 &  0.48 &   0.00 &   1.00 &                       &   0.58 &  0.50 &   0.00 &   1.00 &                       &   0.42 &  0.49 &   0.00 &   1.00 \\
                \hspace{3mm} Height (cm) &         &                       & 156.50 &  7.48 & 139.00 & 176.00 &                       & 157.41 &  7.80 & 133.00 & 175.00 &                       & 161.05 &  9.18 & 100.00 & 208.30 \\
                \hspace{3mm} Weight (kg) &         &                       &  46.56 &  9.15 &  27.00 &  72.00 &                       &  46.86 & 11.27 &  28.00 &  99.00 &                       &  48.33 & 10.64 &  26.30 & 130.00 \\
              \hspace{3mm} Siblings Help &         &                       &   0.47 &  0.50 &   0.00 &   1.00 &                       &   0.43 &  0.50 &   0.00 &   1.00 &                       &   0.46 &  0.50 &   0.00 &   1.00 \\
               \hspace{3mm} Parents Help &         &                       &   0.61 &  0.49 &   0.00 &   1.00 &                       &   0.63 &  0.48 &   0.00 &   1.00 &                       &   0.65 &  0.48 &   0.00 &   1.00 \\
                      \hspace{3mm} Music &         &                       &   0.53 &  0.58 &   0.00 &   2.00 &                       &   0.66 &  0.62 &   0.00 &   3.00 &                       &   0.85 &  0.56 &   0.00 &   3.00 \\
        \hspace{3mm} Commute by car/taxi &         &                       &   0.68 &  0.47 &   0.00 &   1.00 &                       &   0.53 &  0.50 &   0.00 &   1.00 &                       &   0.58 &  0.49 &   0.00 &   1.00 \\
\textbf{Cognitive and Personality Tests} &         &                       &        &       &        &        &                       &        &       &        &        &                       &        &       &        &        \\
                  \hspace{3mm} Cognitive &  [0,16] &                       &   7.80 &  1.66 &   3.00 &  12.00 &                       &   7.94 &  1.78 &   4.00 &  12.00 &                       &   8.92 &  1.88 &   2.00 &  14.00 \\
              \hspace{3mm} Agreeableness &  [9,40] &                       &  27.12 &  4.26 &  14.00 &  39.00 &                       &  27.09 &  3.91 &  15.00 &  37.00 &                       &  27.04 &  3.99 &  12.00 &  40.00 \\
          \hspace{3mm} Conscientiousness &  [9,45] &                       &  26.71 &  5.87 &  14.00 &  40.00 &                       &  27.90 &  4.97 &  18.00 &  43.00 &                       &  25.88 &  5.47 &  12.00 &  45.00 \\
               \hspace{3mm} Extraversion &  [8,40] &                       &  27.65 &  4.86 &  16.00 &  38.00 &                       &  27.62 &  4.85 &  16.00 &  38.00 &                       &  26.35 &  5.10 &  10.00 &  39.00 \\
                \hspace{3mm} Neuroticism &  [8,40] &                       &  22.31 &  5.83 &   9.00 &  36.00 &                       &  21.95 &  5.36 &   8.00 &  35.00 &                       &  23.45 &  5.57 &   9.00 &  38.00 \\
                   \hspace{3mm} Openness & [10,55] &                       &  37.45 &  5.45 &  24.00 &  50.00 &                       &  36.65 &  5.09 &  19.00 &  51.00 &                       &  35.26 &  5.48 &  18.00 &  51.00 \\
\bottomrule
\end{tabular}

%% file: tables/tables_rev/table2.tex
\begin{tabular}{lcccccccccccc}
\toprule
&   \multicolumn{3}{c}{School 1} &  &\multicolumn{3}{c}{School 2}&   &\multicolumn{3}{c}{School 3} \\
\cmidrule(lr){2-4} \cmidrule(lr){6-8} \cmidrule(lr){10-12}

              Variables & Studymates & Seatmates & Extra &  \multicolumn{1}{c}{} & Studymates & Seatmates & Extra &  \multicolumn{1}{l}{} & Studymates & Seatmates &  Extra \\
\midrule
        Number of nodes &        133 &       133 &   133 &                       &        171 &       171 &   171 &                       &        564 &       564 &    564 \\
        Number of edges &        171 &       175 &   454 &                       &        228 &       275 &   748 &                       &        819 &       799 &   2974 \\
   Density $\times$ 100 &      1.948 &     1.994 & 5.172 &                       &      1.569 &     1.892 & 5.146 &                       &      0.516 &     0.503 &  1.873 \\
         Average degree &      2.571 &     2.632 & 6.827 &                       &      2.667 &     3.216 & 8.749 &                       &      2.904 &     2.833 & 10.546 \\
     Average clustering &      0.213 &     0.068 & 0.209 &                       &      0.153 &     0.066 & 0.228 &                       &      0.129 &     0.063 &  0.160 \\
  Assortativity measure &      0.190 &     0.266 & 0.187 &                       &      0.039 &     0.194 & 0.256 &                       &      0.177 &     0.177 &  0.067 \\
Number of isolated node &         15 &         7 &     3 &                       &         19 &         3 &     5 &                       &         69 &        13 &      5 \\
     Number of Subgraph &         21 &        14 &     4 &                       &         27 &         9 &     6 &                       &         76 &        30 &      8 \\
           Transitivity &      0.281 &     0.094 & 0.244 &                       &      0.207 &     0.092 & 0.234 &                       &      0.180 &     0.081 &  0.153 \\
\bottomrule
\end{tabular}

%% file: tables/tables_rev/table4.tex
\begin{tabular}{lccccccccc}
\toprule
&   &\multicolumn{2}{c}{OLS} & &\multicolumn{2}{c}{G2SLS}&  & \multicolumn{2}{c}{GMM}\\
  \cmidrule(lr){3-4} \cmidrule(lr){6-7} \cmidrule(l){9-10}  
                          Variables &  \multicolumn{1}{c}{} &     Coef. &       SE &  \multicolumn{1}{l}{} &      Coef. &       SE &  \multicolumn{1}{r}{} &      Coef. &       SE \\
\midrule
                              Music &                       & -0.0196** & (0.0099) &                       &    -0.0031 & (0.0130) &                       &    -0.0019 & (0.0126) \\
                Elder Siblings Help &                       &    0.0087 & (0.0116) &                       &    -0.0030 & (0.0126) &                       &     0.0001 & (0.0138) \\
              Younger Siblings Help &                       & 0.0435*** & (0.0159) &                       &   0.0403** & (0.0201) &                       &    0.0382* & (0.0212) \\
        Male $\times$ ln(Cognitive) &                       & 0.1275*** & (0.0486) &                       &   0.1528** & (0.0680) &                       &  0.2008*** & (0.0576) \\
    Male $\times$ ln(Agreeableness) &                       & -0.1318** & (0.0668) &                       &    -0.0318 & (0.0622) &                       &    -0.0407 & (0.0667) \\
Male $\times$ ln(Conscientiousness) &                       &  0.1216** & (0.0520) &                       &   0.1612** & (0.0639) &                       &   0.1582** & (0.0727) \\
     Male $\times$ ln(Extraversion) &                       &    0.0266 & (0.0404) &                       &     0.0565 & (0.0559) &                       &     0.0434 & (0.0641) \\
      Male $\times$ ln(Neuroticism) &                       &    0.0441 & (0.0518) &                       &     0.1033 & (0.0630) &                       &     0.1003 & (0.0658) \\
         Male $\times$ ln(Openness) &                       &   -0.0157 & (0.0566) &                       &    -0.0968 & (0.0805) &                       &    -0.1089 & (0.0945) \\
         School 1, grade 7, class 1 &                       &  0.1081** & (0.0516) &                       &  0.0417*** & (0.0155) &                       &    0.0293* & (0.0176) \\
         School 1, grade 7, class 2 &                       &    0.0243 & (0.0533) &                       & -0.0682*** & (0.0247) &                       & -0.1288*** & (0.0238) \\
         School 1, grade 7, class 3 &                       & 0.0854*** & (0.0283) &                       &  0.0658*** & (0.0106) &                       &  0.0779*** & (0.0111) \\
         School 1, grade 7, class 4 &                       & 0.1583*** & (0.0525) &                       &  0.0889*** & (0.0148) &                       &  0.1177*** & (0.0132) \\
         School 1, grade 7, class 5 &                       & 0.1897*** & (0.0493) &                       &  0.1225*** & (0.0154) &                       &  0.1394*** & (0.0124) \\
         School 2, grade 7, class 1 &                       & 0.0938*** & (0.0275) &                       &     0.0204 & (0.0125) &                       &     0.0162 & (0.0135) \\
         School 2, grade 7, class 2 &                       &    0.0516 & (0.0326) &                       &    -0.0120 & (0.0132) &                       &    -0.0172 & (0.0159) \\
         School 2, grade 7, class 3 &                       & 0.1350*** & (0.0428) &                       &  0.0538*** & (0.0153) &                       &  0.0649*** & (0.0155) \\
         School 2, grade 7, class 4 &                       &  0.0720** & (0.0316) &                       &  0.0472*** & (0.0141) &                       &   0.0390** & (0.0161) \\
         School 2, grade 7, class 5 &                       & 0.1312*** & (0.0449) &                       &  0.1126*** & (0.0180) &                       &  0.1480*** & (0.0164) \\
         School 3, grade 7, class 1 &                       &  0.0948** & (0.0433) &                       &  0.1207*** & (0.0159) &                       &  0.1541*** & (0.0131) \\
         School 3, grade 7, class 2 &                       &  0.0983** & (0.0434) &                       &  0.1335*** & (0.0165) &                       &  0.1755*** & (0.0117) \\
         School 3, grade 7, class 3 &                       &  0.0900** & (0.0438) &                       &  0.1488*** & (0.0172) &                       &  0.1916*** & (0.0170) \\
         School 3, grade 7, class 4 &                       &   0.0719* & (0.0385) &                       &  0.1034*** & (0.0153) &                       &  0.1294*** & (0.0129) \\
         School 3, grade 7, class 5 &                       &    0.0612 & (0.0453) &                       &  0.0944*** & (0.0202) &                       &  0.1384*** & (0.0151) \\
         School 3, grade 8, class 1 &                       & 0.1364*** & (0.0379) &                       &  0.1232*** & (0.0144) &                       &  0.1547*** & (0.0096) \\
         School 3, grade 8, class 2 &                       &  0.1223** & (0.0520) &                       &  0.1482*** & (0.0235) &                       &  0.2016*** & (0.0162) \\
         School 3, grade 8, class 3 &                       & 0.0982*** & (0.0334) &                       &  0.1308*** & (0.0150) &                       &  0.1619*** & (0.0116) \\
         School 3, grade 8, class 4 &                       &  0.1069** & (0.0484) &                       &  0.1493*** & (0.0219) &                       &  0.1997*** & (0.0172) \\
         School 3, grade 8, class 5 &                       &   0.0770* & (0.0411) &                       &  0.1087*** & (0.0175) &                       &  0.1545*** & (0.0115) \\
         School 3, grade 9, class 1 &                       &   0.0409* & (0.0246) &                       &  0.0782*** & (0.0075) &                       &  0.0693*** & (0.0072) \\
         School 3, grade 9, class 2 &                       &    0.0384 & (0.0263) &                       &  0.0557*** & (0.0066) &                       &  0.0575*** & (0.0090) \\
         School 3, grade 9, class 3 &                       &    0.0374 & (0.0295) &                       &  0.0543*** & (0.0082) &                       &  0.0518*** & (0.0080) \\
         School 3, grade 9, class 4 &                       &    0.0382 & (0.0291) &                       &  0.0798*** & (0.0072) &                       &  0.0819*** & (0.0084) \\
                           Constant &                       &  -4.5872* & (2.5808) &                       &  2.9464*** & (0.9048) &                       &  3.8210*** & (1.2635) \\
\bottomrule
\end{tabular}

%% file: tables/tables_rev/table15.tex
\begin{tabular}{lccc}
\toprule
                  Variables &   $C=1.5$ &    $C=1.6$ &    $C=1.7$ \\
\midrule
       \textbf{Peer effect} &           &            &            \\
              ln(Math Test) &   0.7002* &  0.8889*** &  0.8919*** \\
                            &  (0.3737) &   (0.3157) &   (0.3069) \\
                   \midrule 
\textbf{Contextual effects} &           &            &            \\
                       Male &   -0.2958 &   -0.2395* &   -0.2403* \\
                            &  (0.1836) &   (0.1436) &   (0.1417) \\
                 ln(Height) &   3.8548* &   2.2241** &   2.2012** \\
                            &  (2.0045) &   (1.0261) &   (1.0108) \\
                 ln(Weight) &    0.0245 &    -0.0452 &    -0.0558 \\
                            &  (0.4489) &   (0.3267) &   (0.3123) \\
              Siblings Help &   -0.0049 &     0.0674 &     0.0630 \\
                            &  (0.1390) &   (0.0833) &   (0.0812) \\
               Parents Help &   -0.0303 &    -0.1154 &    -0.1114 \\
                            &  (0.1479) &   (0.0841) &   (0.0823) \\
        Commute by Car/Taxi &    0.0857 &     0.1395 &     0.1515 \\
                            &  (0.1307) &   (0.1126) &   (0.1110) \\
                      Music &    0.0575 &     0.0842 &     0.0782 \\
                            &  (0.0989) &   (0.0841) &   (0.0822) \\
                   \midrule 
   \textbf{$\dagger$} &            &            &            \\                 
              ln(Cognitive) &  0.1119** &  0.1318*** &  0.1322*** \\
                            &  (0.0494) &   (0.0437) &   (0.0421) \\
          ln(Agreeableness) &   -0.0953 &  -0.1254** &   -0.1204* \\
                            &  (0.0720) &   (0.0632) &   (0.0618) \\
      ln(Conscientiousness) &    0.0738 &     0.0649 &     0.0647 \\
                            &  (0.0740) &   (0.0718) &   (0.0713) \\
           ln(Extraversion) & -0.1552** &  -0.1290** &  -0.1259** \\
                            &  (0.0768) &   (0.0614) &   (0.0608) \\
            ln(Neuroticism) &   -0.0003 &    -0.0333 &    -0.0331 \\
                            &  (0.0447) &   (0.0396) &   (0.0389) \\
               ln(Openness) &    0.0229 &     0.0115 &     0.0068 \\
                            &  (0.0741) &   (0.0654) &   (0.0645) \\            
                   \midrule 
    \textbf{Direct effects} &           &            &            \\
                       Male &   -0.2902 &    -0.6800 &    -0.6580 \\
                            &  (0.8509) &   (0.8366) &   (0.8331) \\
                 ln(Height) & -1.3506** & -0.9101*** & -0.9102*** \\
                            &  (0.5345) &   (0.3046) &   (0.2987) \\
                 ln(Weight) &   -0.0692 &    -0.0238 &    -0.0187 \\
                            &  (0.0966) &   (0.0759) &   (0.0725) \\
              Siblings Help &   -0.0407 &  -0.0640** &  -0.0627** \\
                            &  (0.0444) &   (0.0282) &   (0.0273) \\
               Parents Help &    0.0229 &    0.0378* &    0.0371* \\
                            &  (0.0295) &   (0.0203) &   (0.0200) \\
        Commute by Car/Taxi &   -0.0229 &   -0.0315* &   -0.0335* \\
                            &  (0.0190) &   (0.0176) &   (0.0171) \\
                     Degree & 0.0258*** &  0.0209*** &  0.0208*** \\
                            &  (0.0067) &   (0.0058) &   (0.0057) \\
           Isolate Students &    0.3446 &     0.4858 &     0.4755 \\
                            &  (0.4991) &   (0.4822) &   (0.4766) \\
                   \midrule 
                        $n$ &       868 &        868 &        868 \\
             Adjusted $R^2$ &    0.2100 &     0.2617 &     0.2615 \\
                       RMSE &    0.2205 &     0.2092 &     0.2094 \\
\bottomrule
\end{tabular}

%% file: tables/tables_rev/table16.tex
\begin{tabular}{lccccccccc}
\toprule
&   &\multicolumn{2}{c}{$C=1.5$} & &\multicolumn{2}{c}{$C=1.6$}&  & \multicolumn{2}{c}{$C=1.7$}\\
   \cmidrule(lr){3-4} \cmidrule(lr){6-7} \cmidrule(l){9-10}  
   \cmidrule(lr){3-4} \cmidrule(lr){6-7} \cmidrule(l){9-10} 
                           Variables &  \multicolumn{1}{c}{} &     Coef. &       SE &  \multicolumn{1}{l}{} &     Coef. &       SE &  \multicolumn{1}{r}{} &     Coef. &       SE \\
\midrule
                               Music &                       &   -0.0164 & (0.0189) &                       &   -0.0115 & (0.0163) &                       &   -0.0102 & (0.0158) \\
                 Elder Siblings Help &                       &    0.0171 & (0.0207) &                       &    0.0086 & (0.0173) &                       &    0.0073 & (0.0169) \\
               Younger Siblings Help &                       &    0.0357 & (0.0296) &                       &  0.0470** & (0.0218) &                       &  0.0470** & (0.0216) \\
         Male $\times$ ln(Cognitive) &                       &    0.0823 & (0.0951) &                       &    0.0802 & (0.0902) &                       &    0.0788 & (0.0875) \\
     Male $\times$ ln(Agreeableness) &                       &   -0.0804 & (0.1019) &                       &   -0.0349 & (0.1017) &                       &   -0.0354 & (0.0997) \\
 Male $\times$ ln(Conscientiousness) &                       &    0.1095 & (0.0821) &                       &   0.1334* & (0.0751) &                       &   0.1324* & (0.0736) \\
      Male $\times$ ln(Extraversion) &                       &   -0.0066 & (0.0830) &                       &    0.0313 & (0.0632) &                       &    0.0286 & (0.0626) \\
       Male $\times$ ln(Neuroticism) &                       &    0.0384 & (0.0846) &                       &    0.0700 & (0.0777) &                       &    0.0660 & (0.0776) \\
          Male $\times$ ln(Openness) &                       &    0.0417 & (0.0925) &                       &    0.0239 & (0.0821) &                       &    0.0261 & (0.0798) \\
          School 1, grade 7, class 1 &                       &  0.1649** & (0.0771) &                       &  0.1201** & (0.0604) &                       &   0.1136* & (0.0592) \\
          School 1, grade 7, class 2 &                       &    0.0582 & (0.1200) &                       &    0.0671 & (0.1028) &                       &    0.0607 & (0.0998) \\
          School 1, grade 7, class 3 &                       &  0.1181** & (0.0460) &                       &   0.0711* & (0.0367) &                       &   0.0687* & (0.0354) \\
          School 1, grade 7, class 4 &                       &  0.2019** & (0.0891) &                       &    0.1097 & (0.0687) &                       &    0.1033 & (0.0671) \\
          School 1, grade 7, class 5 &                       & 0.2346*** & (0.0808) &                       &  0.1595** & (0.0639) &                       &  0.1525** & (0.0624) \\
          School 2, grade 7, class 1 &                       &  0.1464** & (0.0568) &                       & 0.1190*** & (0.0414) &                       & 0.1146*** & (0.0392) \\
          School 2, grade 7, class 2 &                       &    0.1105 & (0.0674) &                       &    0.0621 & (0.0532) &                       &    0.0609 & (0.0509) \\
          School 2, grade 7, class 3 &                       & 0.1861*** & (0.0704) &                       &  0.1174** & (0.0564) &                       &  0.1107** & (0.0544) \\
          School 2, grade 7, class 4 &                       &   0.1138* & (0.0600) &                       &  0.1154** & (0.0463) &                       &  0.1113** & (0.0444) \\
          School 2, grade 7, class 5 &                       &  0.1701** & (0.0767) &                       &    0.1049 & (0.0644) &                       &    0.1016 & (0.0624) \\
          School 3, grade 7, class 1 &                       &  0.1646** & (0.0769) &                       &    0.0778 & (0.0569) &                       &    0.0737 & (0.0558) \\
          School 3, grade 7, class 2 &                       &   0.1504* & (0.0808) &                       &    0.0727 & (0.0623) &                       &    0.0688 & (0.0611) \\
          School 3, grade 7, class 3 &                       &   0.1408* & (0.0744) &                       &    0.0634 & (0.0600) &                       &    0.0609 & (0.0587) \\
          School 3, grade 7, class 4 &                       &   0.1212* & (0.0631) &                       &    0.0701 & (0.0488) &                       &    0.0665 & (0.0475) \\
          School 3, grade 7, class 5 &                       &    0.0904 & (0.0741) &                       &    0.0268 & (0.0611) &                       &    0.0242 & (0.0600) \\
          School 3, grade 8, class 1 &                       &  0.1560** & (0.0616) &                       &  0.1077** & (0.0508) &                       &  0.1055** & (0.0496) \\
          School 3, grade 8, class 2 &                       &    0.1342 & (0.0906) &                       &    0.0715 & (0.0770) &                       &    0.0690 & (0.0751) \\
          School 3, grade 8, class 3 &                       &    0.0943 & (0.0590) &                       &    0.0611 & (0.0492) &                       &    0.0586 & (0.0480) \\
          School 3, grade 8, class 4 &                       &    0.1327 & (0.0845) &                       &    0.0601 & (0.0714) &                       &    0.0568 & (0.0700) \\
          School 3, grade 8, class 5 &                       &    0.0750 & (0.0640) &                       &    0.0438 & (0.0592) &                       &    0.0407 & (0.0574) \\
          School 3, grade 9, class 1 &                       &    0.0446 & (0.0355) &                       &   0.0515* & (0.0301) &                       &    0.0486 & (0.0298) \\
          School 3, grade 9, class 2 &                       &    0.0319 & (0.0427) &                       &    0.0269 & (0.0338) &                       &    0.0238 & (0.0331) \\
          School 3, grade 9, class 3 &                       &    0.0239 & (0.0554) &                       &    0.0308 & (0.0394) &                       &    0.0270 & (0.0383) \\
          School 3, grade 9, class 4 &                       &    0.0428 & (0.0380) &                       &    0.0500 & (0.0336) &                       &    0.0488 & (0.0334) \\
                            Constant &                       & -11.5128* & (6.3535) &                       &  -6.1593* & (3.3041) &                       &  -6.0328* & (3.2568) \\
\bottomrule
\end{tabular}

%% file: tables/tables_rev/table17.tex
\begin{tabular}{lccc}
\toprule
                  Variables &    $C=1.5$ &    $C=1.6$ &    $C=1.7$ \\
\midrule
       \textbf{Peer effect} &            &            &            \\
              ln(Math Test) &  0.6729*** &  0.6762*** &  0.6791*** \\
                            &   (0.2236) &   (0.2159) &   (0.2090) \\
                   \midrule 
\textbf{Contextual effects} &            &            &            \\
                       Male &   -0.2057* &   -0.2111* &  -0.2161** \\
                            &   (0.1115) &   (0.1080) &   (0.1046) \\
                 ln(Height) &  2.2455*** &  2.2006*** &  2.1608*** \\
                            &   (0.8568) &   (0.8471) &   (0.8359) \\
                 ln(Weight) &     0.0443 &     0.0460 &     0.0486 \\
                            &   (0.2703) &   (0.2653) &   (0.2607) \\
              Siblings Help &     0.0387 &     0.0399 &     0.0419 \\
                            &   (0.0567) &   (0.0546) &   (0.0525) \\
               Parents Help &    -0.0205 &    -0.0133 &    -0.0072 \\
                            &   (0.0658) &   (0.0639) &   (0.0620) \\
        Commute by Car/Taxi &     0.0649 &     0.0733 &     0.0799 \\
                            &   (0.0688) &   (0.0669) &   (0.0653) \\
                      Music &    0.1161* &   0.1143** &   0.1131** \\
                            &   (0.0600) &   (0.0582) &   (0.0565) \\
                   \midrule 
    \textbf{$\dagger$} &           &            &            \\
              ln(Cognitive) &  0.1098*** &  0.1107*** &  0.1118*** \\
                            &   (0.0378) &   (0.0361) &   (0.0345) \\
          ln(Agreeableness) &   -0.0963* &   -0.0902* &   -0.0844* \\
                            &   (0.0500) &   (0.0483) &   (0.0466) \\
      ln(Conscientiousness) &    0.0915* &    0.0930* &    0.0946* \\
                            &   (0.0517) &   (0.0501) &   (0.0487) \\
           ln(Extraversion) & -0.1227*** &  -0.1206** &  -0.1185** \\
                            &   (0.0474) &   (0.0469) &   (0.0462) \\
            ln(Neuroticism) &    -0.0150 &    -0.0136 &    -0.0122 \\
                            &   (0.0293) &   (0.0284) &   (0.0275) \\
               ln(Openness) &     0.0221 &     0.0179 &     0.0138 \\
                            &   (0.0493) &   (0.0475) &   (0.0458) \\
                   \midrule 
    \textbf{Direct effects} &           &            &            \\
                       Male &    -0.4717 &    -0.4442 &    -0.4138 \\
                            &   (0.6411) &   (0.6298) &   (0.6167) \\
                 ln(Height) & -0.9817*** & -0.9767*** & -0.9699*** \\
                            &   (0.2697) &   (0.2625) &   (0.2553) \\
                 ln(Weight) &    -0.0685 &    -0.0668 &    -0.0656 \\
                            &   (0.0675) &   (0.0654) &   (0.0637) \\
              Siblings Help &    -0.0347 &    -0.0351 &   -0.0360* \\
                            &   (0.0235) &   (0.0227) &   (0.0218) \\
               Parents Help &     0.0210 &     0.0186 &     0.0168 \\
                            &   (0.0168) &   (0.0163) &   (0.0158) \\
        Commute by Car/Taxi &    -0.0160 &    -0.0175 &    -0.0188 \\
                            &   (0.0156) &   (0.0153) &   (0.0150) \\
                     Degree &  0.0248*** &  0.0247*** &  0.0247*** \\
                            &   (0.0044) &   (0.0042) &   (0.0041) \\
           Isolate Students &     0.2403 &     0.2186 &     0.1963 \\
                            &   (0.3495) &   (0.3417) &   (0.3327) \\
                   \midrule 
                        $n$ &        868 &        868 &        868 \\
             Adjusted $R^2$ &     0.2800 &     0.2813 &     0.2821 \\
                       RMSE &     0.1980 &     0.1979 &     0.1980 \\
\bottomrule
\end{tabular}

%% file: tables/tables_rev/table18.tex
\begin{tabular}{lccccccccc}
\toprule
&   &\multicolumn{2}{c}{$C=1.5$} & &\multicolumn{2}{c}{$C=1.6$}&  & \multicolumn{2}{c}{$C=1.7$}\\
   \cmidrule(lr){3-4} \cmidrule(lr){6-7} \cmidrule(l){9-10}  
                           Variables &  \multicolumn{1}{c}{} &     Coef. &       SE &  \multicolumn{1}{l}{} &     Coef. &       SE &  \multicolumn{1}{r}{} &     Coef. &       SE \\
\midrule
                               Music &                       &  -0.0232* & (0.0119) &                       &  -0.0225* & (0.0115) &                       &  -0.0219* & (0.0112) \\
                 Elder Siblings Help &                       &    0.0052 & (0.0129) &                       &    0.0055 & (0.0125) &                       &    0.0060 & (0.0121) \\
               Younger Siblings Help &                       &   0.0372* & (0.0192) &                       &  0.0380** & (0.0190) &                       &  0.0389** & (0.0188) \\
         Male $\times$ ln(Cognitive) &                       &    0.1055 & (0.0675) &                       &    0.1049 & (0.0640) &                       &   0.1044* & (0.0609) \\
     Male $\times$ ln(Agreeableness) &                       &   -0.0848 & (0.0879) &                       &   -0.0866 & (0.0857) &                       &   -0.0885 & (0.0836) \\
 Male $\times$ ln(Conscientiousness) &                       &  0.1311** & (0.0634) &                       &  0.1297** & (0.0614) &                       &  0.1276** & (0.0599) \\
      Male $\times$ ln(Extraversion) &                       &    0.0488 & (0.0488) &                       &    0.0473 & (0.0461) &                       &    0.0460 & (0.0437) \\
       Male $\times$ ln(Neuroticism) &                       &    0.0720 & (0.0620) &                       &    0.0703 & (0.0610) &                       &    0.0682 & (0.0600) \\
          Male $\times$ ln(Openness) &                       &   -0.0239 & (0.0706) &                       &   -0.0242 & (0.0672) &                       &   -0.0246 & (0.0641) \\
          School 1, grade 7, class 1 &                       &  0.1284** & (0.0519) &                       &  0.1232** & (0.0510) &                       &  0.1189** & (0.0502) \\
          School 1, grade 7, class 2 &                       &    0.0504 & (0.0677) &                       &    0.0469 & (0.0656) &                       &    0.0439 & (0.0636) \\
          School 1, grade 7, class 3 &                       & 0.0897*** & (0.0324) &                       & 0.0879*** & (0.0312) &                       & 0.0867*** & (0.0301) \\
          School 1, grade 7, class 4 &                       & 0.1606*** & (0.0563) &                       & 0.1568*** & (0.0555) &                       & 0.1540*** & (0.0547) \\
          School 1, grade 7, class 5 &                       & 0.1954*** & (0.0528) &                       & 0.1902*** & (0.0518) &                       & 0.1859*** & (0.0507) \\
          School 2, grade 7, class 1 &                       & 0.0985*** & (0.0343) &                       & 0.0960*** & (0.0320) &                       & 0.0935*** & (0.0299) \\
          School 2, grade 7, class 2 &                       &    0.0584 & (0.0450) &                       &    0.0571 & (0.0430) &                       &    0.0570 & (0.0412) \\
          School 2, grade 7, class 3 &                       & 0.1376*** & (0.0457) &                       & 0.1335*** & (0.0443) &                       & 0.1304*** & (0.0429) \\
          School 2, grade 7, class 4 &                       &  0.0852** & (0.0343) &                       &  0.0823** & (0.0328) &                       &  0.0802** & (0.0312) \\
          School 2, grade 7, class 5 &                       &  0.1242** & (0.0494) &                       &  0.1226** & (0.0479) &                       & 0.1219*** & (0.0465) \\
          School 3, grade 7, class 1 &                       &  0.1041** & (0.0477) &                       &  0.0995** & (0.0468) &                       &  0.0955** & (0.0460) \\
          School 3, grade 7, class 2 &                       &  0.1009** & (0.0498) &                       &  0.0988** & (0.0486) &                       &  0.0968** & (0.0476) \\
          School 3, grade 7, class 3 &                       &   0.0940* & (0.0497) &                       &   0.0909* & (0.0483) &                       &   0.0882* & (0.0470) \\
          School 3, grade 7, class 4 &                       &   0.0791* & (0.0431) &                       &   0.0757* & (0.0416) &                       &   0.0729* & (0.0401) \\
          School 3, grade 7, class 5 &                       &    0.0560 & (0.0507) &                       &    0.0542 & (0.0497) &                       &    0.0527 & (0.0488) \\
          School 3, grade 8, class 1 &                       & 0.1302*** & (0.0396) &                       & 0.1290*** & (0.0383) &                       & 0.1280*** & (0.0373) \\
          School 3, grade 8, class 2 &                       &   0.1121* & (0.0590) &                       &   0.1101* & (0.0572) &                       &   0.1085* & (0.0555) \\
          School 3, grade 8, class 3 &                       &  0.0903** & (0.0382) &                       &  0.0879** & (0.0371) &                       &  0.0858** & (0.0361) \\
          School 3, grade 8, class 4 &                       &   0.1058* & (0.0550) &                       &   0.1025* & (0.0536) &                       &   0.0996* & (0.0523) \\
          School 3, grade 8, class 5 &                       &    0.0644 & (0.0465) &                       &    0.0635 & (0.0450) &                       &    0.0629 & (0.0437) \\
          School 3, grade 9, class 1 &                       &   0.0475* & (0.0256) &                       &   0.0442* & (0.0256) &                       &    0.0413 & (0.0253) \\
          School 3, grade 9, class 2 &                       &    0.0424 & (0.0297) &                       &    0.0390 & (0.0292) &                       &    0.0359 & (0.0289) \\
          School 3, grade 9, class 3 &                       &    0.0483 & (0.0342) &                       &    0.0434 & (0.0332) &                       &    0.0394 & (0.0323) \\
          School 3, grade 9, class 4 &                       &    0.0358 & (0.0309) &                       &    0.0339 & (0.0307) &                       &    0.0319 & (0.0302) \\
                            Constant &                       &  -5.2394* & (2.7824) &                       &  -5.0676* & (2.7405) &                       &  -4.9306* & (2.6995) \\
\bottomrule
\end{tabular}

%% file: tables/tables_rev/table19.tex
\begin{tabular}{lccc}
\toprule
                  Variables &    $C=1.5$ &    $C=1.6$ &    $C=1.7$ \\
\midrule
       \textbf{Peer effect} &            &            &            \\
              ln(Math Test) &  0.5943*** &  0.5987*** &  0.6065*** \\
                            &   (0.1923) &   (0.1836) &   (0.1755) \\
                   \midrule 
\textbf{Contextual effects} &            &            &            \\
                       Male & -0.2494*** & -0.2466*** & -0.2455*** \\
                            &   (0.0949) &   (0.0909) &   (0.0873) \\
                 ln(Height) &  2.3137*** &  2.2570*** &  2.1682*** \\
                            &   (0.7917) &   (0.7660) &   (0.7400) \\
                 ln(Weight) &     0.0491 &     0.0399 &     0.0329 \\
                            &   (0.2123) &   (0.2031) &   (0.1942) \\
              Siblings Help &     0.0263 &     0.0274 &     0.0294 \\
                            &   (0.0521) &   (0.0505) &   (0.0491) \\
               Parents Help &    -0.0326 &    -0.0207 &    -0.0106 \\
                            &   (0.0590) &   (0.0554) &   (0.0517) \\
        Commute by Car/Taxi &    0.1147* &    0.1205* &   0.1255** \\
                            &   (0.0683) &   (0.0659) &   (0.0636) \\
                      Music &   0.1243** &   0.1206** &   0.1187** \\
                            &   (0.0531) &   (0.0505) &   (0.0480) \\
    \midrule 
    \textbf{$\dagger$} &           &            &            \\
              ln(Cognitive) &  0.1093*** &  0.1087*** &  0.1086*** \\
                            &   (0.0317) &   (0.0300) &   (0.0286) \\
          ln(Agreeableness) &   -0.0866* &    -0.0787 &    -0.0712 \\
                            &   (0.0501) &   (0.0479) &   (0.0457) \\
      ln(Conscientiousness) &    0.0752* &    0.0774* &   0.0799** \\
                            &   (0.0435) &   (0.0418) &   (0.0406) \\
           ln(Extraversion) & -0.1229*** & -0.1202*** & -0.1164*** \\
                            &   (0.0420) &   (0.0412) &   (0.0405) \\
            ln(Neuroticism) &    -0.0132 &    -0.0124 &    -0.0113 \\
                            &   (0.0278) &   (0.0266) &   (0.0257) \\
               ln(Openness) &     0.0353 &     0.0311 &     0.0272 \\
                            &   (0.0456) &   (0.0435) &   (0.0414) \\
                   \midrule 
    \textbf{Direct effects} &            &            &            \\
                       Male &    -0.2089 &    -0.1901 &    -0.1648 \\
                            &   (0.5776) &   (0.5609) &   (0.5417) \\
                 ln(Height) & -1.0302*** & -1.0141*** & -0.9904*** \\
                            &   (0.2436) &   (0.2343) &   (0.2248) \\
                 ln(Weight) &    -0.0578 &    -0.0537 &    -0.0494 \\
                            &   (0.0516) &   (0.0483) &   (0.0453) \\
              Siblings Help &   -0.0394* &   -0.0391* &   -0.0395* \\
                            &   (0.0225) &   (0.0214) &   (0.0204) \\
               Parents Help &     0.0224 &     0.0194 &     0.0167 \\
                            &   (0.0155) &   (0.0147) &   (0.0139) \\
        Commute by Car/Taxi &    -0.0216 &   -0.0234* &   -0.0251* \\
                            &   (0.0145) &   (0.0140) &   (0.0136) \\
                     Degree &  0.0251*** &  0.0250*** &  0.0248*** \\
                            &   (0.0037) &   (0.0036) &   (0.0034) \\
           Isolate Students &     0.2024 &     0.1752 &     0.1425 \\
                            &   (0.3118) &   (0.3022) &   (0.2935) \\
                   \midrule 
                        $n$ &        868 &        868 &        868 \\
             Adjusted $R^2$ &     0.2593 &     0.2636 &     0.2675 \\
                       RMSE &     0.2020 &     0.2010 &     0.2002 \\
\bottomrule
\end{tabular}

%% file: tables/tables_rev/table20.tex
\begin{tabular}{lccccccccc}
\toprule
&   &\multicolumn{2}{c}{$C=1.5$} & &\multicolumn{2}{c}{$C=1.6$}&  & \multicolumn{2}{c}{$C=1.7$}\\
   \cmidrule(lr){3-4} \cmidrule(lr){6-7} \cmidrule(l){9-10}  
                           Variables &  \multicolumn{1}{c}{} &     Coef. &       SE &  \multicolumn{1}{l}{} &     Coef. &       SE &  \multicolumn{1}{r}{} &     Coef. &       SE \\
\midrule
                               Music &                       &  -0.0204* & (0.0110) &                       &  -0.0199* & (0.0105) &                       & -0.0196** & (0.0099) \\
                 Elder Siblings Help &                       &    0.0088 & (0.0125) &                       &    0.0087 & (0.0121) &                       &    0.0087 & (0.0116) \\
               Younger Siblings Help &                       &  0.0410** & (0.0169) &                       &  0.0421** & (0.0164) &                       & 0.0435*** & (0.0159) \\
         Male $\times$ ln(Cognitive) &                       &  0.1349** & (0.0538) &                       &  0.1313** & (0.0510) &                       & 0.1275*** & (0.0486) \\
     Male $\times$ ln(Agreeableness) &                       &  -0.1317* & (0.0730) &                       &  -0.1321* & (0.0702) &                       & -0.1318** & (0.0668) \\
 Male $\times$ ln(Conscientiousness) &                       &  0.1256** & (0.0563) &                       &  0.1240** & (0.0539) &                       &  0.1216** & (0.0520) \\
      Male $\times$ ln(Extraversion) &                       &    0.0337 & (0.0472) &                       &    0.0298 & (0.0438) &                       &    0.0266 & (0.0404) \\
       Male $\times$ ln(Neuroticism) &                       &    0.0425 & (0.0538) &                       &    0.0435 & (0.0527) &                       &    0.0441 & (0.0518) \\
          Male $\times$ ln(Openness) &                       &   -0.0163 & (0.0633) &                       &   -0.0154 & (0.0598) &                       &   -0.0157 & (0.0566) \\
          School 1, grade 7, class 1 &                       &  0.1198** & (0.0543) &                       &  0.1142** & (0.0530) &                       &  0.1081** & (0.0516) \\
          School 1, grade 7, class 2 &                       &    0.0304 & (0.0583) &                       &    0.0268 & (0.0557) &                       &    0.0243 & (0.0533) \\
          School 1, grade 7, class 3 &                       & 0.0877*** & (0.0309) &                       & 0.0869*** & (0.0296) &                       & 0.0854*** & (0.0283) \\
          School 1, grade 7, class 4 &                       & 0.1716*** & (0.0540) &                       & 0.1652*** & (0.0533) &                       & 0.1583*** & (0.0525) \\
          School 1, grade 7, class 5 &                       & 0.2042*** & (0.0521) &                       & 0.1973*** & (0.0508) &                       & 0.1897*** & (0.0493) \\
          School 2, grade 7, class 1 &                       & 0.1028*** & (0.0324) &                       & 0.0985*** & (0.0299) &                       & 0.0938*** & (0.0275) \\
          School 2, grade 7, class 2 &                       &    0.0573 & (0.0370) &                       &    0.0541 & (0.0347) &                       &    0.0516 & (0.0326) \\
          School 2, grade 7, class 3 &                       & 0.1482*** & (0.0456) &                       & 0.1418*** & (0.0443) &                       & 0.1350*** & (0.0428) \\
          School 2, grade 7, class 4 &                       &  0.0801** & (0.0356) &                       &  0.0759** & (0.0337) &                       &  0.0720** & (0.0316) \\
          School 2, grade 7, class 5 &                       & 0.1387*** & (0.0478) &                       & 0.1352*** & (0.0463) &                       & 0.1312*** & (0.0449) \\
          School 3, grade 7, class 1 &                       &  0.1047** & (0.0460) &                       &  0.1005** & (0.0446) &                       &  0.0948** & (0.0433) \\
          School 3, grade 7, class 2 &                       &  0.1046** & (0.0458) &                       &  0.1022** & (0.0446) &                       &  0.0983** & (0.0434) \\
          School 3, grade 7, class 3 &                       &  0.0969** & (0.0471) &                       &  0.0940** & (0.0454) &                       &  0.0900** & (0.0438) \\
          School 3, grade 7, class 4 &                       &  0.0816** & (0.0409) &                       &   0.0770* & (0.0397) &                       &   0.0719* & (0.0385) \\
          School 3, grade 7, class 5 &                       &    0.0682 & (0.0468) &                       &    0.0652 & (0.0461) &                       &    0.0612 & (0.0453) \\
          School 3, grade 8, class 1 &                       & 0.1428*** & (0.0403) &                       & 0.1398*** & (0.0391) &                       & 0.1364*** & (0.0379) \\
          School 3, grade 8, class 2 &                       &  0.1314** & (0.0551) &                       &  0.1271** & (0.0535) &                       &  0.1223** & (0.0520) \\
          School 3, grade 8, class 3 &                       & 0.1049*** & (0.0352) &                       & 0.1016*** & (0.0343) &                       & 0.0982*** & (0.0334) \\
          School 3, grade 8, class 4 &                       &  0.1167** & (0.0509) &                       &  0.1124** & (0.0496) &                       &  0.1069** & (0.0484) \\
          School 3, grade 8, class 5 &                       &   0.0818* & (0.0440) &                       &   0.0795* & (0.0425) &                       &   0.0770* & (0.0411) \\
          School 3, grade 9, class 1 &                       &   0.0450* & (0.0254) &                       &   0.0428* & (0.0251) &                       &   0.0409* & (0.0246) \\
          School 3, grade 9, class 2 &                       &    0.0415 & (0.0279) &                       &    0.0401 & (0.0271) &                       &    0.0384 & (0.0263) \\
          School 3, grade 9, class 3 &                       &    0.0476 & (0.0324) &                       &    0.0421 & (0.0310) &                       &    0.0374 & (0.0295) \\
          School 3, grade 9, class 4 &                       &    0.0430 & (0.0306) &                       &    0.0407 & (0.0300) &                       &    0.0382 & (0.0291) \\
                            Constant &                       &  -5.0956* & (2.7565) &                       &  -4.8916* & (2.6662) &                       &  -4.5872* & (2.5808) \\
\bottomrule
\end{tabular}

%% file: tables/tables_rev/table14.tex
\begin{tabular}{lccclcc}
\toprule
                                  Variables &                                              Coef. &       SE &  \multicolumn{1}{c}{} &                                    Variables &     Coef. &      SE \\
\midrule
                              ln(Math Test) &                                          0.2204*** & (0.0784) &                       &                $\textbf{W}^2_{n,0}$ $\times$ Male &    -0.0512 & (0.0466) \\
                                       Male &                                            -0.3938 & (0.4050) &                       &          $\textbf{W}^2_{n,0}$ $\times$ ln(Height) &    -0.3393 & (0.2333) \\
                                 ln(Height) &                                           -0.2541* & (0.1477) &                       &          $\textbf{W}^2_{n,0}$ $\times$ ln(Weight) &    0.1710* & (0.0903) \\
                                 ln(Weight) &                                           -0.0669* & (0.0392) &                       &       $\textbf{W}^2_{n,0}$ $\times$ Siblings Help &     0.0571 & (0.0470) \\
                              Siblings Help &                                            -0.0262 & (0.0223) &                       &        $\textbf{W}^2_{n,0}$ $\times$ Parents Help &    -0.0117 & (0.0557) \\
                               Parents Help &                                             0.0290 & (0.0248) &                       & $\textbf{W}^2_{n,0}$ $\times$ Commute by Car/Taxi &     0.0111 & (0.0508) \\
                        Commute by Car/Taxi &                                            -0.0104 & (0.0200) &                       &               $\textbf{W}^2_{n,0}$ $\times$ Music &     0.0251 & (0.0363) \\
                                      Music &                                            -0.0091 & (0.0132) &                       &                    School 1, grade 7, class 1 &     0.0105 & (0.0232) \\
                                     Degree &                                          0.0281*** & (0.0040) &                       &                    School 1, grade 7, class 2 & -0.1337*** & (0.0251) \\
                           Isolate Students &                                             1.2398 & (1.1350) &                       &                    School 1, grade 7, class 3 &  0.0768*** & (0.0127) \\
                        Elder Siblings Help &                                            -0.0002 & (0.0140) &                       &                    School 1, grade 7, class 4 &  0.1075*** & (0.0206) \\
                      Younger Siblings Help &                                            0.0364* & (0.0191) &                       &                    School 1, grade 7, class 5 &  0.1462*** & (0.0167) \\
                Male $\times$ ln(Cognitive) &                                          0.2563*** & (0.0452) &                       &                    School 2, grade 7, class 1 &     0.0233 & (0.0150) \\
            Male $\times$ ln(Agreeableness) &                                          -0.1395** & (0.0588) &                       &                    School 2, grade 7, class 2 &    -0.0049 & (0.0244) \\
        Male $\times$ ln(Conscientiousness) &                                          0.2086*** & (0.0624) &                       &                    School 2, grade 7, class 3 &  0.0742*** & (0.0175) \\
             Male $\times$ ln(Extraversion) &                                            -0.0559 & (0.0543) &                       &                    School 2, grade 7, class 4 &  0.0487*** & (0.0177) \\
              Male $\times$ ln(Neuroticism) &                                             0.0632 & (0.0537) &                       &                    School 2, grade 7, class 5 &  0.1526*** & (0.0227) \\
                 Male $\times$ ln(Openness) &                                            -0.0795 & (0.0802) &                       &                    School 3, grade 7, class 1 &  0.1593*** & (0.0188) \\
                 $\textbf{W}_n$ $\times$ Male &                                            -0.0330 & (0.0299) &                       &                    School 3, grade 7, class 2 &  0.1710*** & (0.0187) \\
           $\textbf{W}_n$ $\times$ ln(Height) &                                            -0.0563 & (0.2103) &                       &                    School 3, grade 7, class 3 &  0.1997*** & (0.0270) \\
           $\textbf{W}_n$ $\times$ ln(Weight) &                                           0.1383** & (0.0673) &                       &                    School 3, grade 7, class 4 &  0.1355*** & (0.0249) \\
        $\textbf{W}_n$ $\times$ Siblings Help &                                             0.0321 & (0.0259) &                       &                    School 3, grade 7, class 5 &  0.1486*** & (0.0184) \\
         $\textbf{W}_n$ $\times$ Parents Help &                                            0.0438* & (0.0264) &                       &                    School 3, grade 8, class 1 &  0.1516*** & (0.0167) \\
  $\textbf{W}_n$ $\times$ Commute by Car/Taxi &                                             0.0334 & (0.0264) &                       &                    School 3, grade 8, class 2 &  0.2023*** & (0.0183) \\
                $\textbf{W}_n$ $\times$ Music &                                            -0.0236 & (0.0201) &                       &                    School 3, grade 8, class 3 &  0.1563*** & (0.0129) \\
               $\textbf{W}_{n,0}$ $\times$ Male &                                           -0.0248 & (0.0149) &                       &                    School 3, grade 8, class 4 &  0.1964*** & (0.0215) \\
         $\textbf{W}_{n,0}$ $\times$ ln(Height) &                                             0.2446 & (0.2356) &                       &                    School 3, grade 8, class 5 &  0.1676*** & (0.0112) \\
         $\textbf{W}_{n,0}$ $\times$ ln(Weight) &                                            -0.0355 & (0.0528) &                       &                    School 3, grade 9, class 1 &  0.0651*** & (0.0125) \\
      $\textbf{W}_{n,0}$ $\times$ Siblings Help &                                            -0.0151 & (0.0213) &                       &                    School 3, grade 9, class 2 &  0.0449*** & (0.0067) \\
       $\textbf{W}_{n,0}$ $\times$ Parents Help &                                            -0.0106 & (0.0206) &                       &                    School 3, grade 9, class 3 &  0.0508*** & (0.0085) \\
$\textbf{W}_{n,0}$ $\times$ Commute by Car/Taxi &                                            -0.0118 & (0.0245) &                       &                    School 3, grade 9, class 4 &  0.0668*** & (0.0159) \\
              $\textbf{W}_{n,0}$ $\times$ Music &                                            -0.0036 & (0.0169) &                       &                                      Constant &  4.2484*** & (1.2059) \\
              \midrule
                                        $n$ &                                                868 &          &                       &                                               &            &          \\
                                     $F$-Statistics & 12.9000 &          &                       &                                               &            &          \\
                                   $p$-value &                                                0.0000 &          &                       &                                               &            &          \\
\bottomrule
\end{tabular}

%% file: tables/tables_rev/table5.tex
\begin{tabular}{lccccccccc}
\toprule
 &   &\multicolumn{2}{c}{(D-2)} & &\multicolumn{2}{c}{(D-3)}&  & \multicolumn{2}{c}{(D-4)}\\
   \cmidrule(lr){3-4} \cmidrule(lr){6-7} \cmidrule(l){9-10}  
                                       Variables &  \multicolumn{1}{c}{} &     Coef. &       SE &  \multicolumn{1}{l}{} &     Coef. &       SE &  \multicolumn{1}{r}{} &     Coef. &       SE \\
\midrule
                    $\textbf{W}_{n,0}$ $\times$ Male &                       &    0.1277 & (0.1521) &                       &           &          &                       &    0.2244 & (0.1505) \\
              $\textbf{W}_{n,0}$ $\times$ ln(Height) &                       &  1.1821** & (0.5783) &                       &           &          &                       &    1.6369 & (2.0158) \\
              $\textbf{W}_{n,0}$ $\times$ ln(Weight) &                       &   -0.5142 & (0.3457) &                       &           &          &                       &   -0.2787 & (0.4065) \\
           $\textbf{W}_{n,0}$ $\times$ Siblings Help &                       &    0.1925 & (0.1588) &                       &           &          &                       &    0.1315 & (0.1354) \\
            $\textbf{W}_{n,0}$ $\times$ Parents Help &                       &    0.1572 & (0.1427) &                       &           &          &                       &    0.2092 & (0.1405) \\
     $\textbf{W}_{n,0}$ $\times$ Commute by Car/Taxi &                       &    0.0588 & (0.1407) &                       &           &          &                       &    0.0455 & (0.1457) \\
                   $\textbf{W}_{n,0}$ $\times$ Music &                       &  0.2277** & (0.1077) &                       &           &          &                       &    0.0790 & (0.1280) \\
           $\textbf{W}_{n,0}$ $\times$ ln(Cognitive) &                       &   -0.1062 & (0.1635) &                       &           &          &                       &    0.3246 & (0.2969) \\
       $\textbf{W}_{n,0}$ $\times$ ln(Agreeableness) &                       &    0.1843 & (0.2616) &                       &           &          &                       &   -0.0557 & (0.4539) \\
   $\textbf{W}_{n,0}$ $\times$ ln(Conscientiousness) &                       & 0.6115*** & (0.2035) &                       &           &          &                       &    0.0661 & (0.3383) \\
        $\textbf{W}_{n,0}$ $\times$ ln(Extraversion) &                       &    0.1699 & (0.1471) &                       &           &          &                       &    0.2083 & (0.2455) \\
         $\textbf{W}_{n,0}$ $\times$ ln(Neuroticism) &                       &    0.0704 & (0.1952) &                       &           &          &                       &   -0.0816 & (0.2820) \\
            $\textbf{W}_{n,0}$ $\times$ ln(Openness) &                       &    0.2039 & (0.2609) &                       &           &          &                       &    0.2749 & (0.3080) \\
                  $\textbf{W}^2_{n,0}$ $\times$ Male &                       &        &       &                       &   -0.2862 & (0.2443) &                       &   -0.3681 & (0.2556) \\
            $\textbf{W}^2_{n,0}$ $\times$ ln(Height) &                       &        &       &                       &    1.1908 & (0.8774) &                       &   -0.5011 & (1.9841) \\
            $\textbf{W}^2_{n,0}$ $\times$ ln(Weight) &                       &        &       &                       &   -0.2088 & (0.6695) &                       &    0.0118 & (0.7225) \\
         $\textbf{W}^2_{n,0}$ $\times$ Siblings Help &                       &        &       &                       & 0.5321*** & (0.1883) &                       & 0.5106*** & (0.1875) \\
          $\textbf{W}^2_{n,0}$ $\times$ Parents Help &                       &        &       &                       &    0.0411 & (0.2204) &                       &    0.0127 & (0.2210) \\
   $\textbf{W}^2_{n,0}$ $\times$ Commute by Car/Taxi &                       &        &       &                       &   -0.1939 & (0.2112) &                       &   -0.1865 & (0.2059) \\
                 $\textbf{W}^2_{n,0}$ $\times$ Music &                       &        &       &                       & 0.4272*** & (0.1628) &                       &  0.4009** & (0.1825) \\
         $\textbf{W}^2_{n,0}$ $\times$ ln(Cognitive) &                       &        &       &                       &   -0.3517 & (0.2263) &                       &   -0.7223 & (0.4509) \\
     $\textbf{W}^2_{n,0}$ $\times$ ln(Agreeableness) &                       &        &       &                       &    0.3012 & (0.4596) &                       &    0.4606 & (0.7711) \\
 $\textbf{W}^2_{n,0}$ $\times$ ln(Conscientiousness) &                       &        &       &                       & 1.1246*** & (0.2488) &                       &  1.0246** & (0.4803) \\
      $\textbf{W}^2_{n,0}$ $\times$ ln(Extraversion) &                       &        &       &                       &    0.1388 & (0.3130) &                       &   -0.1316 & (0.5399) \\
       $\textbf{W}^2_{n,0}$ $\times$ ln(Neuroticism) &                       &        &       &                       &    0.0915 & (0.3264) &                       &    0.1898 & (0.5169) \\
          $\textbf{W}^2_{n,0}$ $\times$ ln(Openness) &                       &        &       &                       &    0.1003 & (0.4454) &                       &   -0.2281 & (0.5347) \\
                                        Constant &                       &   -0.7597 & (1.8614) &                       &   -1.8886 & (2.8393) &                       &   -1.6844 & (2.9717) \\
                                        \midrule
                                             $n$ &                       &       868 &          &                       &       868 &          &                       &       868 &          \\
                                          $F$-Statistics &                       &    4.1245 &          &                       &    9.2074 &          &                       &   12.3000 &          \\
                                         $p$-value &                       &    0.0013 &          &                       &    0.0000 &          &                       &       0.0000 &          \\
\bottomrule
\end{tabular}

%% file: tables/tables_rev/table7.tex
\begin{tabular}{lccccccccc}
\toprule
 &   &\multicolumn{2}{c}{(D-5)} & &\multicolumn{2}{c}{(D-6)}&  & \multicolumn{2}{c}{(D-7)}\\
   \cmidrule(lr){3-4} \cmidrule(lr){6-7} \cmidrule(l){9-10}  
                                      Variables &  \multicolumn{1}{c}{} &      Coef. &       SE &  \multicolumn{1}{l}{} &      Coef. &       SE &  \multicolumn{1}{r}{} &      Coef. &       SE \\
\midrule
             $\textbf{W}_{n,0}$ $\times$ ln(Height) &                       &   -0.1938* & (0.1152) &                       &            &          &                       &   -0.8460* & (0.4750) \\
             $\textbf{W}_{n,0}$ $\times$ ln(Weight) &                       &   0.1995** & (0.0950) &                       &            &          &                       &    0.1754* & (0.0922) \\
          $\textbf{W}_{n,0}$ $\times$ Siblings Help &                       &     0.0570 & (0.0516) &                       &            &          &                       &     0.0652 & (0.0486) \\
           $\textbf{W}_{n,0}$ $\times$ Parents Help &                       &    -0.0086 & (0.0494) &                       &            &          &                       &    -0.0142 & (0.0496) \\
    $\textbf{W}_{n,0}$ $\times$ Commute by Car/Taxi &                       &    -0.0495 & (0.0511) &                       &            &          &                       &    -0.0470 & (0.0510) \\
                  $\textbf{W}_{n,0}$ $\times$ Music &                       &    -0.0159 & (0.0496) &                       &            &          &                       &     0.0092 & (0.0491) \\
          $\textbf{W}_{n,0}$ $\times$ ln(Cognitive) &                       &  0.2337*** & (0.0448) &                       &            &          &                       &     0.1023 & (0.0700) \\
      $\textbf{W}_{n,0}$ $\times$ ln(Agreeableness) &                       & -0.1969*** & (0.0622) &                       &            &          &                       &     0.1203 & (0.1030) \\
  $\textbf{W}_{n,0}$ $\times$ ln(Conscientiousness) &                       &     0.1031 & (0.0627) &                       &            &          &                       &     0.0781 & (0.1071) \\
       $\textbf{W}_{n,0}$ $\times$ ln(Extraversion) &                       &     0.0437 & (0.0367) &                       &            &          &                       &     0.0723 & (0.0714) \\
        $\textbf{W}_{n,0}$ $\times$ ln(Neuroticism) &                       & -0.1813*** & (0.0502) &                       &            &          &                       &    -0.0531 & (0.0784) \\
           $\textbf{W}_{n,0}$ $\times$ ln(Openness) &                       &    -0.0243 & (0.0876) &                       &            &          &                       &    -0.0344 & (0.0896) \\
           $\textbf{W}^2_{n,0}$ $\times$ ln(Height) &                       &         &       &                       &   -0.3383* & (0.1797) &                       &     0.3152 & (0.4610) \\
           $\textbf{W}^2_{n,0}$ $\times$ ln(Weight) &                       &         &       &                       &     0.2465 & (0.1653) &                       &     0.2330 & (0.1827) \\
        $\textbf{W}^2_{n,0}$ $\times$ Siblings Help &                       &         &       &                       &     0.0158 & (0.0582) &                       &     0.0296 & (0.0621) \\
         $\textbf{W}^2_{n,0}$ $\times$ Parents Help &                       &         &       &                       &     0.0624 & (0.0618) &                       &     0.0802 & (0.0579) \\
  $\textbf{W}^2_{n,0}$ $\times$ Commute by Car/Taxi &                       &         &       &                       &    -0.0127 & (0.0620) &                       &    -0.0032 & (0.0651) \\
                $\textbf{W}^2_{n,0}$ $\times$ Music &                       &         &       &                       &   -0.0940* & (0.0505) &                       &   -0.0906* & (0.0483) \\
        $\textbf{W}^2_{n,0}$ $\times$ ln(Cognitive) &                       &         &       &                       &  0.4096*** & (0.0722) &                       &   0.2789** & (0.1191) \\
    $\textbf{W}^2_{n,0}$ $\times$ ln(Agreeableness) &                       &         &       &                       & -0.5443*** & (0.1188) &                       & -0.7011*** & (0.1642) \\
$\textbf{W}^2_{n,0}$ $\times$ ln(Conscientiousness) &                       &         &       &                       &   0.2130** & (0.0943) &                       &     0.1059 & (0.1780) \\
     $\textbf{W}^2_{n,0}$ $\times$ ln(Extraversion) &                       &         &       &                       &     0.0224 & (0.0805) &                       &    -0.0703 & (0.1309) \\
      $\textbf{W}^2_{n,0}$ $\times$ ln(Neuroticism) &                       &         &       &                       & -0.2919*** & (0.0801) &                       &   -0.2427* & (0.1313) \\
         $\textbf{W}^2_{n,0}$ $\times$ ln(Openness) &                       &         &       &                       &     0.0448 & (0.1575) &                       &     0.0866 & (0.1824) \\
                                       Constant &                       &     0.5878 & (0.4326) &                       &    1.1516* & (0.6258) &                       &   1.4789** & (0.6419) \\
                                       \midrule
                                            $n$ &                       &        868 &          &                       &        868 &          &                       &        868 &          \\
                                         $F$-Statistics &                       &    10.9248 &          &                       &    19.3011 &          &                       &    126.4900 &          \\
                                        $p$-value &                       &     0.0000 &          &                       &     0.0000 &          &                       &        0.0000 &          \\
\bottomrule
\end{tabular}

%% file: tables/tables_rev/table8.tex
\begin{tabular}{lccccccccc}
\toprule
 &   &\multicolumn{2}{c}{(D-5)} & &\multicolumn{2}{c}{(D-6)}&  & \multicolumn{2}{c}{(D-7)}\\
   \cmidrule(lr){3-4} \cmidrule(lr){6-7} \cmidrule(l){9-10}  
                                      Variables &  \multicolumn{1}{c}{} &     Coef. &       SE &  \multicolumn{1}{l}{} &     Coef. &       SE &  \multicolumn{1}{r}{} &    Coef. &       SE \\
\midrule
                   $\textbf{W}_{n,0}$ $\times$ Male &                       &    0.1180 & (0.1872) &                       &           &          &                       &   0.2803 & (0.1927) \\
             $\textbf{W}_{n,0}$ $\times$ ln(Weight) &                       &    0.2625 & (0.1916) &                       &           &          &                       &   0.0643 & (0.3955) \\
          $\textbf{W}_{n,0}$ $\times$ Siblings Help &                       &    0.2883 & (0.1941) &                       &           &          &                       &   0.1956 & (0.1622) \\
           $\textbf{W}_{n,0}$ $\times$ Parents Help &                       &    0.1967 & (0.1748) &                       &           &          &                       &   0.2461 & (0.1683) \\
    $\textbf{W}_{n,0}$ $\times$ Commute by Car/Taxi &                       &    0.1375 & (0.1625) &                       &           &          &                       &   0.1052 & (0.1700) \\
                  $\textbf{W}_{n,0}$ $\times$ Music &                       &   0.2489* & (0.1306) &                       &           &          &                       &   0.0600 & (0.1596) \\
          $\textbf{W}_{n,0}$ $\times$ ln(Cognitive) &                       &   -0.2424 & (0.1828) &                       &           &          &                       &   0.4494 & (0.3957) \\
      $\textbf{W}_{n,0}$ $\times$ ln(Agreeableness) &                       &   -0.1306 & (0.2771) &                       &           &          &                       &  -0.1050 & (0.5708) \\
  $\textbf{W}_{n,0}$ $\times$ ln(Conscientiousness) &                       &  0.5010** & (0.2017) &                       &           &          &                       &   0.0665 & (0.4104) \\
       $\textbf{W}_{n,0}$ $\times$ ln(Extraversion) &                       &    0.0614 & (0.1853) &                       &           &          &                       &   0.2481 & (0.2555) \\
        $\textbf{W}_{n,0}$ $\times$ ln(Neuroticism) &                       &   -0.1489 & (0.1897) &                       &           &          &                       &  -0.0746 & (0.3510) \\
           $\textbf{W}_{n,0}$ $\times$ ln(Openness) &                       &    0.1441 & (0.3229) &                       &           &          &                       &   0.3960 & (0.3747) \\
                 $\textbf{W}^2_{n,0}$ $\times$ Male &                       &        &       &                       &   -0.4456 & (0.3082) &                       & -0.5613* & (0.3149) \\
           $\textbf{W}^2_{n,0}$ $\times$ ln(Weight) &                       &        &       &                       &    0.6125 & (0.5231) &                       &   0.5680 & (0.6976) \\
        $\textbf{W}^2_{n,0}$ $\times$ Siblings Help &                       &        &       &                       &  0.5837** & (0.2417) &                       & 0.5906** & (0.2335) \\
         $\textbf{W}^2_{n,0}$ $\times$ Parents Help &                       &        &       &                       &    0.0234 & (0.2432) &                       &  -0.0045 & (0.2465) \\
  $\textbf{W}^2_{n,0}$ $\times$ Commute by Car/Taxi &                       &        &       &                       &   -0.2487 & (0.2473) &                       &  -0.2697 & (0.2421) \\
                $\textbf{W}^2_{n,0}$ $\times$ Music &                       &        &       &                       & 0.5045*** & (0.1665) &                       & 0.4822** & (0.1975) \\
        $\textbf{W}^2_{n,0}$ $\times$ ln(Cognitive) &                       &        &       &                       & -0.5994** & (0.2816) &                       & -1.1088* & (0.5745) \\
    $\textbf{W}^2_{n,0}$ $\times$ ln(Agreeableness) &                       &        &       &                       &    0.0437 & (0.4733) &                       &   0.3197 & (0.8356) \\
$\textbf{W}^2_{n,0}$ $\times$ ln(Conscientiousness) &                       &        &       &                       & 1.0995*** & (0.2510) &                       &  1.0402* & (0.5467) \\
     $\textbf{W}^2_{n,0}$ $\times$ ln(Extraversion) &                       &        &       &                       &   -0.0054 & (0.3580) &                       &  -0.2854 & (0.5586) \\
      $\textbf{W}^2_{n,0}$ $\times$ ln(Neuroticism) &                       &        &       &                       &   -0.1436 & (0.3153) &                       &  -0.0041 & (0.5480) \\
         $\textbf{W}^2_{n,0}$ $\times$ ln(Openness) &                       &        &       &                       &   -0.0027 & (0.5465) &                       &  -0.4294 & (0.6521) \\
                                       Constant &                       & 2.8956*** & (0.8051) &                       &    1.8026 & (1.9623) &                       &   1.3428 & (1.7616) \\
                                       \midrule
                                            $n$ &                       &       868 &          &                       &       868 &          &                       &      868 &          \\
                                         $F$-Statistics &                       &    2.8743 &          &                       &    8.3756 &          &                       &  267.8900 &          \\
                                        $p$-value &                       &    0.0134 &          &                       &    0.0000 &          &                       &      0.0000 &          \\
\bottomrule
\end{tabular}

%% file: tables/tables_rev/table9.tex
\begin{tabular}{lccccccccc}
\toprule
 &   &\multicolumn{2}{c}{(D-5)} & &\multicolumn{2}{c}{(D-6)}&  & \multicolumn{2}{c}{(D-7)}\\
   \cmidrule(lr){3-4} \cmidrule(lr){6-7} \cmidrule(l){9-10}  
                                      Variables &  \multicolumn{1}{c}{} &     Coef. &       SE &  \multicolumn{1}{l}{} &     Coef. &       SE &  \multicolumn{1}{r}{} &    Coef. &       SE \\
\midrule
                   $\textbf{W}_{n,0}$ $\times$ Male &                       &    0.0836 & (0.1408) &                       &           &          &                       &   0.1835 & (0.1335) \\
             $\textbf{W}_{n,0}$ $\times$ ln(Height) &                       &  0.7174** & (0.3325) &                       &           &          &                       &   1.2811 & (1.4793) \\
          $\textbf{W}_{n,0}$ $\times$ Siblings Help &                       &    0.1808 & (0.1492) &                       &           &          &                       &   0.1313 & (0.1349) \\
           $\textbf{W}_{n,0}$ $\times$ Parents Help &                       &    0.1202 & (0.1302) &                       &           &          &                       &   0.1585 & (0.1339) \\
    $\textbf{W}_{n,0}$ $\times$ Commute by Car/Taxi &                       &    0.0749 & (0.1223) &                       &           &          &                       &   0.0626 & (0.1324) \\
                  $\textbf{W}_{n,0}$ $\times$ Music &                       &    0.1526 & (0.0957) &                       &           &          &                       &   0.0207 & (0.1181) \\
          $\textbf{W}_{n,0}$ $\times$ ln(Cognitive) &                       &   -0.1177 & (0.1470) &                       &           &          &                       &   0.2401 & (0.2855) \\
      $\textbf{W}_{n,0}$ $\times$ ln(Agreeableness) &                       &    0.1349 & (0.2465) &                       &           &          &                       &  -0.1997 & (0.4257) \\
  $\textbf{W}_{n,0}$ $\times$ ln(Conscientiousness) &                       & 0.5276*** & (0.1917) &                       &           &          &                       &  -0.0289 & (0.3189) \\
       $\textbf{W}_{n,0}$ $\times$ ln(Extraversion) &                       &    0.1552 & (0.1399) &                       &           &          &                       &   0.1207 & (0.2145) \\
        $\textbf{W}_{n,0}$ $\times$ ln(Neuroticism) &                       &    0.0531 & (0.1750) &                       &           &          &                       &  -0.1283 & (0.2683) \\
           $\textbf{W}_{n,0}$ $\times$ ln(Openness) &                       &    0.2478 & (0.2398) &                       &           &          &                       &   0.3701 & (0.2943) \\
                 $\textbf{W}^2_{n,0}$ $\times$ Male &                       &        &       &                       &   -0.2662 & (0.2209) &                       &  -0.3274 & (0.2211) \\
           $\textbf{W}^2_{n,0}$ $\times$ ln(Height) &                       &        &       &                       &  1.1151** & (0.5044) &                       &  -0.2210 & (1.6553) \\
        $\textbf{W}^2_{n,0}$ $\times$ Siblings Help &                       &        &       &                       &  0.3890** & (0.1887) &                       & 0.3742** & (0.1821) \\
         $\textbf{W}^2_{n,0}$ $\times$ Parents Help &                       &        &       &                       &   -0.0328 & (0.1957) &                       &  -0.0597 & (0.1952) \\
  $\textbf{W}^2_{n,0}$ $\times$ Commute by Car/Taxi &                       &        &       &                       &   -0.2155 & (0.1877) &                       &  -0.2338 & (0.1884) \\
                $\textbf{W}^2_{n,0}$ $\times$ Music &                       &        &       &                       &  0.3344** & (0.1344) &                       & 0.3274** & (0.1590) \\
        $\textbf{W}^2_{n,0}$ $\times$ ln(Cognitive) &                       &        &       &                       &   -0.3279 & (0.2107) &                       &  -0.6054 & (0.4380) \\
    $\textbf{W}^2_{n,0}$ $\times$ ln(Agreeableness) &                       &        &       &                       &    0.3474 & (0.4004) &                       &   0.6676 & (0.6866) \\
$\textbf{W}^2_{n,0}$ $\times$ ln(Conscientiousness) &                       &        &       &                       & 1.0425*** & (0.2321) &                       & 1.0599** & (0.4446) \\
     $\textbf{W}^2_{n,0}$ $\times$ ln(Extraversion) &                       &        &       &                       &    0.1596 & (0.2918) &                       &   0.0006 & (0.4765) \\
      $\textbf{W}^2_{n,0}$ $\times$ ln(Neuroticism) &                       &        &       &                       &    0.1233 & (0.2948) &                       &   0.2734 & (0.4803) \\
         $\textbf{W}^2_{n,0}$ $\times$ ln(Openness) &                       &        &       &                       &    0.1558 & (0.4027) &                       &  -0.2917 & (0.4938) \\
                                       Constant &                       &   -0.5989 & (1.6537) &                       &   -2.4319 & (2.4581) &                       &  -2.3797 & (2.4980) \\
                                       \midrule
                                            $n$ &                       &       868 &          &                       &       868 &          &                       &      868 &          \\
                                         $F$-Statistics &                       &    3.2110 &          &                       &    7.3745 &          &                       &  28.3500 &          \\
                                        $p$-value &                       &    0.0072 &          &                       &    0.0000 &          &                       &      0.0000 &          \\
\bottomrule
\end{tabular}

%% file: tables/tables_rev/table10.tex
\begin{tabular}{lccccccccc}
\toprule
 &   &\multicolumn{2}{c}{(D-5)} & &\multicolumn{2}{c}{(D-6)}&  & \multicolumn{2}{c}{(D-7)}\\
   \cmidrule(lr){3-4} \cmidrule(lr){6-7} \cmidrule(l){9-10}  
                                      Variables &  \multicolumn{1}{c}{} &    Coef. &       SE &  \multicolumn{1}{l}{} &     Coef. &       SE &  \multicolumn{1}{r}{} &     Coef. &       SE \\
\midrule
                   $\textbf{W}_{n,0}$ $\times$ Male &                       &  -0.0436 & (0.0379) &                       &           &          &                       &   -0.0275 & (0.0345) \\
             $\textbf{W}_{n,0}$ $\times$ ln(Height) &                       &  -0.0999 & (0.0919) &                       &           &          &                       &   -0.6092 & (0.3855) \\
             $\textbf{W}_{n,0}$ $\times$ ln(Weight) &                       &   0.0392 & (0.0863) &                       &           &          &                       &    0.1094 & (0.0930) \\
           $\textbf{W}_{n,0}$ $\times$ Parents Help &                       &   0.0451 & (0.0332) &                       &           &          &                       &    0.0330 & (0.0301) \\
    $\textbf{W}_{n,0}$ $\times$ Commute by Car/Taxi &                       &  0.0639* & (0.0326) &                       &           &          &                       &   0.0656* & (0.0372) \\
                  $\textbf{W}_{n,0}$ $\times$ Music &                       &   0.0158 & (0.0352) &                       &           &          &                       &   -0.0173 & (0.0335) \\
          $\textbf{W}_{n,0}$ $\times$ ln(Cognitive) &                       &  -0.0248 & (0.0345) &                       &           &          &                       &   -0.0716 & (0.0642) \\
      $\textbf{W}_{n,0}$ $\times$ ln(Agreeableness) &                       &  -0.0670 & (0.0598) &                       &           &          &                       &   -0.0254 & (0.0959) \\
  $\textbf{W}_{n,0}$ $\times$ ln(Conscientiousness) &                       & 0.0854** & (0.0399) &                       &           &          &                       &    0.0177 & (0.0724) \\
       $\textbf{W}_{n,0}$ $\times$ ln(Extraversion) &                       &  -0.0416 & (0.0547) &                       &           &          &                       &    0.0135 & (0.0898) \\
        $\textbf{W}_{n,0}$ $\times$ ln(Neuroticism) &                       &  -0.0215 & (0.0378) &                       &           &          &                       &   -0.0045 & (0.0657) \\
           $\textbf{W}_{n,0}$ $\times$ ln(Openness) &                       &  -0.0335 & (0.0614) &                       &           &          &                       &   -0.0775 & (0.0799) \\
                 $\textbf{W}^2_{n,0}$ $\times$ Male &                       &       &       &                       & -0.1188** & (0.0597) &                       &  -0.1127* & (0.0644) \\
           $\textbf{W}^2_{n,0}$ $\times$ ln(Height) &                       &       &       &                       &    0.0725 & (0.1319) &                       &    0.6017 & (0.4002) \\
           $\textbf{W}^2_{n,0}$ $\times$ ln(Weight) &                       &       &       &                       &   -0.1291 & (0.0842) &                       &  -0.1695* & (0.1014) \\
         $\textbf{W}^2_{n,0}$ $\times$ Parents Help &                       &       &       &                       &    0.0610 & (0.0602) &                       &    0.0652 & (0.0592) \\
  $\textbf{W}^2_{n,0}$ $\times$ Commute by Car/Taxi &                       &       &       &                       &    0.0416 & (0.0510) &                       &    0.0415 & (0.0519) \\
                $\textbf{W}^2_{n,0}$ $\times$ Music &                       &       &       &                       &  0.0947** & (0.0417) &                       & 0.1039*** & (0.0402) \\
        $\textbf{W}^2_{n,0}$ $\times$ ln(Cognitive) &                       &       &       &                       &    0.0496 & (0.0587) &                       &    0.1322 & (0.1107) \\
    $\textbf{W}^2_{n,0}$ $\times$ ln(Agreeableness) &                       &       &       &                       &   -0.1213 & (0.1071) &                       &   -0.1176 & (0.1763) \\
$\textbf{W}^2_{n,0}$ $\times$ ln(Conscientiousness) &                       &       &       &                       & 0.1726*** & (0.0634) &                       &    0.1475 & (0.1139) \\
     $\textbf{W}^2_{n,0}$ $\times$ ln(Extraversion) &                       &       &       &                       &   -0.0354 & (0.0873) &                       &   -0.0551 & (0.1422) \\
      $\textbf{W}^2_{n,0}$ $\times$ ln(Neuroticism) &                       &       &       &                       &   -0.0623 & (0.0651) &                       &   -0.0665 & (0.1185) \\
         $\textbf{W}^2_{n,0}$ $\times$ ln(Openness) &                       &       &       &                       &   -0.0080 & (0.1083) &                       &    0.0776 & (0.1576) \\
                                       Constant &                       & 0.7144** & (0.3385) &                       &    0.4611 & (0.5473) &                       &    0.5524 & (0.6007) \\
                                       \midrule
                                            $n$ &                       &      868 &          &                       &       868 &          &                       &       868 &          \\
                                         $F$-Statistics &                       &   1.7506 &          &                       &    4.3328 &          &                       &   10.4200 &          \\
                                        $p$-value &                       &   0.1174 &          &                       &    0.0011 &          &                       &       0.0000 &          \\
\bottomrule
\end{tabular}

%% file: tables/tables_rev/table11.tex
\begin{tabular}{lccccccccc}
\toprule
 &   &\multicolumn{2}{c}{(D-5)} & &\multicolumn{2}{c}{(D-6)}&  & \multicolumn{2}{c}{(D-7)}\\
   \cmidrule(lr){3-4} \cmidrule(lr){6-7} \cmidrule(l){9-10}  
                                      Variables &  \multicolumn{1}{c}{} &     Coef. &       SE &  \multicolumn{1}{l}{} &     Coef. &       SE &  \multicolumn{1}{r}{} &    Coef. &       SE \\
\midrule
                   $\textbf{W}_{n,0}$ $\times$ Male &                       &    0.0252 & (0.0534) &                       &           &          &                       &   0.0333 & (0.0564) \\
             $\textbf{W}_{n,0}$ $\times$ ln(Height) &                       &    0.2073 & (0.1268) &                       &           &          &                       &   0.5553 & (0.4865) \\
             $\textbf{W}_{n,0}$ $\times$ ln(Weight) &                       & -0.1920** & (0.0948) &                       &           &          &                       & -0.1684* & (0.0909) \\
          $\textbf{W}_{n,0}$ $\times$ Siblings Help &                       &    0.0617 & (0.0435) &                       &           &          &                       &   0.0520 & (0.0438) \\
    $\textbf{W}_{n,0}$ $\times$ Commute by Car/Taxi &                       &    0.0236 & (0.0373) &                       &           &          &                       &   0.0162 & (0.0412) \\
                  $\textbf{W}_{n,0}$ $\times$ Music &                       &    0.0308 & (0.0399) &                       &           &          &                       &  -0.0005 & (0.0379) \\
          $\textbf{W}_{n,0}$ $\times$ ln(Cognitive) &                       &    0.0143 & (0.0390) &                       &           &          &                       &   0.0399 & (0.0814) \\
      $\textbf{W}_{n,0}$ $\times$ ln(Agreeableness) &                       &   -0.0304 & (0.0722) &                       &           &          &                       &  -0.0502 & (0.1190) \\
  $\textbf{W}_{n,0}$ $\times$ ln(Conscientiousness) &                       &  0.1347** & (0.0646) &                       &           &          &                       &   0.0392 & (0.0911) \\
       $\textbf{W}_{n,0}$ $\times$ ln(Extraversion) &                       &    0.0377 & (0.0427) &                       &           &          &                       &   0.1117 & (0.0750) \\
        $\textbf{W}_{n,0}$ $\times$ ln(Neuroticism) &                       &   -0.0500 & (0.0554) &                       &           &          &                       &  -0.1029 & (0.0751) \\
           $\textbf{W}_{n,0}$ $\times$ ln(Openness) &                       &    0.0034 & (0.0716) &                       &           &          &                       &  -0.0387 & (0.0908) \\
                 $\textbf{W}^2_{n,0}$ $\times$ Male &                       &        &       &                       &    0.0146 & (0.0651) &                       &  -0.0097 & (0.0735) \\
           $\textbf{W}^2_{n,0}$ $\times$ ln(Height) &                       &        &       &                       &    0.0980 & (0.1905) &                       &  -0.3670 & (0.5233) \\
           $\textbf{W}^2_{n,0}$ $\times$ ln(Weight) &                       &        &       &                       &   -0.0499 & (0.1486) &                       &   0.0217 & (0.1424) \\
        $\textbf{W}^2_{n,0}$ $\times$ Siblings Help &                       &        &       &                       & 0.1268*** & (0.0433) &                       & 0.0972** & (0.0477) \\
  $\textbf{W}^2_{n,0}$ $\times$ Commute by Car/Taxi &                       &        &       &                       &    0.0172 & (0.0634) &                       &   0.0378 & (0.0602) \\
                $\textbf{W}^2_{n,0}$ $\times$ Music &                       &        &       &                       &  0.0976** & (0.0454) &                       & 0.0989** & (0.0454) \\
        $\textbf{W}^2_{n,0}$ $\times$ ln(Cognitive) &                       &        &       &                       &   -0.0012 & (0.0618) &                       &  -0.0401 & (0.1282) \\
    $\textbf{W}^2_{n,0}$ $\times$ ln(Agreeableness) &                       &        &       &                       &   -0.0425 & (0.1254) &                       &   0.0448 & (0.1980) \\
$\textbf{W}^2_{n,0}$ $\times$ ln(Conscientiousness) &                       &        &       &                       &  0.2338** & (0.1002) &                       &   0.1852 & (0.1538) \\
     $\textbf{W}^2_{n,0}$ $\times$ ln(Extraversion) &                       &        &       &                       &   -0.0281 & (0.0860) &                       &  -0.1573 & (0.1540) \\
      $\textbf{W}^2_{n,0}$ $\times$ ln(Neuroticism) &                       &        &       &                       &   -0.0400 & (0.0898) &                       &   0.0910 & (0.1330) \\
         $\textbf{W}^2_{n,0}$ $\times$ ln(Openness) &                       &        &       &                       &    0.0027 & (0.1347) &                       &   0.0405 & (0.1839) \\
                                       Constant &                       &    0.2101 & (0.4627) &                       &    0.1457 & (0.6395) &                       &   0.0231 & (0.6233) \\
                                       \midrule
                                            $n$ &                       &       868 &          &                       &       868 &          &                       &      868 &          \\
                                         $F$-Statistics &                       &    3.0899 &          &                       &    1.8041 &          &                       &  25.4400&          \\
                                        $p$-value &                       &    0.0090 &          &                       &    0.1057 &          &                       &      0.0000 &          \\
\bottomrule
\end{tabular}

%% file: tables/tables_rev/table12.tex
\begin{tabular}{lccccccccc}
\toprule
 &   &\multicolumn{2}{c}{(D-5)} & &\multicolumn{2}{c}{(D-6)}&  & \multicolumn{2}{c}{(D-7)}\\
   \cmidrule(lr){3-4} \cmidrule(lr){6-7} \cmidrule(l){9-10}  
                                      Variables &  \multicolumn{1}{c}{} &     Coef. &       SE &  \multicolumn{1}{l}{} &     Coef. &       SE &  \multicolumn{1}{r}{} &     Coef. &       SE \\
\midrule
                   $\textbf{W}_{n,0}$ $\times$ Male &                       &    0.0167 & (0.0482) &                       &           &          &                       &    0.0494 & (0.0512) \\
             $\textbf{W}_{n,0}$ $\times$ ln(Height) &                       &    0.1140 & (0.1299) &                       &           &          &                       &   -0.1766 & (0.3776) \\
             $\textbf{W}_{n,0}$ $\times$ ln(Weight) &                       &    0.0206 & (0.0859) &                       &           &          &                       &    0.1109 & (0.0949) \\
          $\textbf{W}_{n,0}$ $\times$ Siblings Help &                       &    0.0464 & (0.0444) &                       &           &          &                       &    0.0344 & (0.0441) \\
           $\textbf{W}_{n,0}$ $\times$ Parents Help &                       &    0.0028 & (0.0290) &                       &           &          &                       &    0.0043 & (0.0314) \\
                  $\textbf{W}_{n,0}$ $\times$ Music &                       &    0.0211 & (0.0360) &                       &           &          &                       &   -0.0033 & (0.0404) \\
          $\textbf{W}_{n,0}$ $\times$ ln(Cognitive) &                       &   -0.0589 & (0.0596) &                       &           &          &                       &    0.0425 & (0.0971) \\
      $\textbf{W}_{n,0}$ $\times$ ln(Agreeableness) &                       &    0.0527 & (0.0563) &                       &           &          &                       &   -0.0520 & (0.1007) \\
  $\textbf{W}_{n,0}$ $\times$ ln(Conscientiousness) &                       & 0.1208*** & (0.0422) &                       &           &          &                       &   -0.0152 & (0.0892) \\
       $\textbf{W}_{n,0}$ $\times$ ln(Extraversion) &                       &    0.0581 & (0.0519) &                       &           &          &                       &   -0.0253 & (0.0891) \\
        $\textbf{W}_{n,0}$ $\times$ ln(Neuroticism) &                       &    0.0323 & (0.0419) &                       &           &          &                       &   -0.0002 & (0.0748) \\
           $\textbf{W}_{n,0}$ $\times$ ln(Openness) &                       &   -0.0344 & (0.0811) &                       &           &          &                       &    0.0570 & (0.1415) \\
                 $\textbf{W}^2_{n,0}$ $\times$ Male &                       &        &       &                       &  -0.1128* & (0.0684) &                       & -0.1296** & (0.0587) \\
           $\textbf{W}^2_{n,0}$ $\times$ ln(Height) &                       &        &       &                       &   0.3602* & (0.2035) &                       &    0.4488 & (0.4245) \\
           $\textbf{W}^2_{n,0}$ $\times$ ln(Weight) &                       &        &       &                       &   -0.1240 & (0.1436) &                       &   -0.1462 & (0.1429) \\
        $\textbf{W}^2_{n,0}$ $\times$ Siblings Help &                       &        &       &                       & 0.1113*** & (0.0398) &                       & 0.1211*** & (0.0415) \\
         $\textbf{W}^2_{n,0}$ $\times$ Parents Help &                       &        &       &                       &    0.0816 & (0.0527) &                       &    0.0838 & (0.0510) \\
                $\textbf{W}^2_{n,0}$ $\times$ Music &                       &        &       &                       &    0.0342 & (0.0376) &                       &    0.0350 & (0.0381) \\
        $\textbf{W}^2_{n,0}$ $\times$ ln(Cognitive) &                       &        &       &                       &   -0.0991 & (0.0765) &                       &   -0.1547 & (0.1215) \\
    $\textbf{W}^2_{n,0}$ $\times$ ln(Agreeableness) &                       &        &       &                       &    0.1499 & (0.1091) &                       &    0.2124 & (0.1820) \\
$\textbf{W}^2_{n,0}$ $\times$ ln(Conscientiousness) &                       &        &       &                       & 0.2543*** & (0.0743) &                       &   0.2695* & (0.1472) \\
     $\textbf{W}^2_{n,0}$ $\times$ ln(Extraversion) &                       &        &       &                       & 0.1712*** & (0.0663) &                       &   0.1959* & (0.1155) \\
      $\textbf{W}^2_{n,0}$ $\times$ ln(Neuroticism) &                       &        &       &                       &    0.0603 & (0.0730) &                       &    0.0494 & (0.1346) \\
         $\textbf{W}^2_{n,0}$ $\times$ ln(Openness) &                       &        &       &                       &   -0.1452 & (0.0956) &                       &   -0.2129 & (0.1814) \\
                                       Constant &                       &   -0.1795 & (0.4529) &                       &   -0.9086 & (0.6702) &                       &   -0.8397 & (0.6981) \\
                                       \midrule
                                            $n$ &                       &       868 &          &                       &       868 &          &                       &       868 &          \\
                                         $F$-Statistics &                       &    2.8197 &          &                       &    6.3966 &          &                       &  21.4400 &          \\
                                        $p$-value &                       &    0.0148 &          &                       &    0.0001 &          &                       &       0.0000 &          \\
\bottomrule
\end{tabular}

%% file: tables/tables_rev/table13.tex
\begin{tabular}{lccccccccc}
\toprule
 &   &\multicolumn{2}{c}{(D-5)} & &\multicolumn{2}{c}{(D-6)}&  & \multicolumn{2}{c}{(D-7)}\\
   \cmidrule(lr){3-4} \cmidrule(lr){6-7} \cmidrule(l){9-10}  
                                      Variables &  \multicolumn{1}{c}{} &      Coef. &       SE &  \multicolumn{1}{l}{} &      Coef. &       SE &  \multicolumn{1}{r}{} &      Coef. &       SE \\
\midrule
                   $\textbf{W}_{n,0}$ $\times$ Male &                       &    -0.0666 & (0.0557) &                       &            &          &                       &    -0.0043 & (0.0557) \\
             $\textbf{W}_{n,0}$ $\times$ ln(Height) &                       &   0.3127** & (0.1515) &                       &            &          &                       &    -0.1204 & (0.5234) \\
             $\textbf{W}_{n,0}$ $\times$ ln(Weight) &                       & -0.3884*** & (0.1083) &                       &            &          &                       & -0.2920*** & (0.1108) \\
          $\textbf{W}_{n,0}$ $\times$ Siblings Help &                       &     0.0005 & (0.0389) &                       &            &          &                       &    -0.0089 & (0.0356) \\
           $\textbf{W}_{n,0}$ $\times$ Parents Help &                       &     0.0761 & (0.0708) &                       &            &          &                       &     0.0791 & (0.0626) \\
    $\textbf{W}_{n,0}$ $\times$ Commute by Car/Taxi &                       &     0.0057 & (0.0499) &                       &            &          &                       &    -0.0038 & (0.0485) \\
          $\textbf{W}_{n,0}$ $\times$ ln(Cognitive) &                       & -0.1652*** & (0.0637) &                       &            &          &                       &    -0.0564 & (0.0934) \\
      $\textbf{W}_{n,0}$ $\times$ ln(Agreeableness) &                       &     0.0980 & (0.0797) &                       &            &          &                       &     0.0883 & (0.1335) \\
  $\textbf{W}_{n,0}$ $\times$ ln(Conscientiousness) &                       &     0.0190 & (0.0478) &                       &            &          &                       &    -0.0632 & (0.0964) \\
       $\textbf{W}_{n,0}$ $\times$ ln(Extraversion) &                       &    -0.0369 & (0.0520) &                       &            &          &                       &    -0.1098 & (0.0681) \\
        $\textbf{W}_{n,0}$ $\times$ ln(Neuroticism) &                       &     0.0477 & (0.0569) &                       &            &          &                       &    -0.0444 & (0.0774) \\
           $\textbf{W}_{n,0}$ $\times$ ln(Openness) &                       &    -0.0273 & (0.0824) &                       &            &          &                       &    -0.0120 & (0.1017) \\
                 $\textbf{W}^2_{n,0}$ $\times$ Male &                       &         &       &                       & -0.2865*** & (0.0669) &                       & -0.2900*** & (0.0804) \\
           $\textbf{W}^2_{n,0}$ $\times$ ln(Height) &                       &         &       &                       &     0.3415 & (0.2258) &                       &     0.6147 & (0.5420) \\
           $\textbf{W}^2_{n,0}$ $\times$ ln(Weight) &                       &         &       &                       &   -0.2853* & (0.1620) &                       &    -0.1608 & (0.2271) \\
        $\textbf{W}^2_{n,0}$ $\times$ Siblings Help &                       &         &       &                       &    0.1136* & (0.0687) &                       &   0.0990** & (0.0445) \\
         $\textbf{W}^2_{n,0}$ $\times$ Parents Help &                       &         &       &                       &     0.0645 & (0.0680) &                       &     0.0614 & (0.0753) \\
  $\textbf{W}^2_{n,0}$ $\times$ Commute by Car/Taxi &                       &         &       &                       &    -0.0844 & (0.0652) &                       &    -0.0823 & (0.0755) \\
        $\textbf{W}^2_{n,0}$ $\times$ ln(Cognitive) &                       &         &       &                       &  -0.2039** & (0.0905) &                       &    -0.1334 & (0.1251) \\
    $\textbf{W}^2_{n,0}$ $\times$ ln(Agreeableness) &                       &         &       &                       &     0.0838 & (0.1486) &                       &    -0.0304 & (0.2709) \\
$\textbf{W}^2_{n,0}$ $\times$ ln(Conscientiousness) &                       &         &       &                       &     0.0887 & (0.1102) &                       &     0.1801 & (0.1672) \\
     $\textbf{W}^2_{n,0}$ $\times$ ln(Extraversion) &                       &         &       &                       &     0.0247 & (0.1097) &                       &     0.1649 & (0.1593) \\
      $\textbf{W}^2_{n,0}$ $\times$ ln(Neuroticism) &                       &         &       &                       &     0.0749 & (0.0914) &                       &     0.1379 & (0.1367) \\
         $\textbf{W}^2_{n,0}$ $\times$ ln(Openness) &                       &         &       &                       &    -0.0119 & (0.1453) &                       &    -0.0084 & (0.2091) \\
                                       Constant &                       &     0.5647 & (0.5552) &                       &     0.1232 & (0.8391) &                       &    -0.0420 & (0.8874) \\
                                       \midrule
                                            $n$ &                       &        868 &          &                       &        868 &          &                       &        868 &          \\
                                         $F$-Statistics &                       &     3.7731 &          &                       &   3.65 &          &                       &   13.4900 &          \\
                                        $p$-value &                       &     0.0027 &          &                       &     0.0000 &          &                       &        0.0000 &          \\
\bottomrule
\end{tabular}

%% file: tables/tables_rev/table6.tex
\begin{tabular}{lccccccccc}
\toprule
 &   &\multicolumn{2}{c}{(D-8)} & &\multicolumn{2}{c}{(D-9)}&  & \multicolumn{2}{c}{(D-10)}\\
   \cmidrule(lr){3-4} \cmidrule(lr){6-7} \cmidrule(l){9-10}  
                                     Variables &  \multicolumn{1}{c}{} &     Coef. &       SE &  \multicolumn{1}{l}{} &     Coef. &       SE &  \multicolumn{1}{r}{} &      Coef. &       SE \\
\midrule
                $\textbf{W}_{n,*}$ $\times$ Male &                       &   -0.0504 & (0.0487) &                       &           &          &                       &    -0.0496 & (0.0413) \\
          $\textbf{W}_{n,*}$ $\times$ ln(Height) &                       & 0.8742*** & (0.2639) &                       &           &          &                       &   0.3636** & (0.1728) \\
          $\textbf{W}_{n,*}$ $\times$ ln(Weight) &                       &   -0.1477 & (0.1316) &                       &           &          &                       &     0.0612 & (0.0988) \\
       $\textbf{W}_{n,*}$ $\times$ Siblings Help &                       &    0.0403 & (0.0444) &                       &           &          &                       &     0.0544 & (0.0394) \\
        $\textbf{W}_{n,*}$ $\times$ Parents Help &                       &    0.0065 & (0.0348) &                       &           &          &                       &    -0.0557 & (0.0400) \\
    $\textbf{W}_{n,*}$ $\times$ Commute Car/Taxi &                       &    0.0122 & (0.0539) &                       &           &          &                       &    -0.0089 & (0.0474) \\
               $\textbf{W}_{*}$ $\times$ Music &                       &    0.0651 & (0.0436) &                       &           &          &                       &    -0.0012 & (0.0365) \\
       $\textbf{W}_{n,*}$ $\times$ ln(Cognitive) &                       & -0.2676** & (0.1081) &                       &           &          &                       & -0.2059*** & (0.0729) \\
   $\textbf{W}_{n,*}$ $\times$ ln(Agreeableness) &                       &   -0.0154 & (0.2255) &                       &           &          &                       &   -0.2280* & (0.1287) \\
   $\textbf{W}_{n,*}$ $\times$ ln(Conscientious) &                       &  0.3370** & (0.1402) &                       &           &          &                       &    0.1333* & (0.0731) \\
    $\textbf{W}_{n,*}$ $\times$ ln(Extraversion) &                       &    0.0306 & (0.1665) &                       &           &          &                       &    0.1601* & (0.0826) \\
     $\textbf{W}_{n,*}$ $\times$ ln(Neuroticism) &                       &   -0.0098 & (0.1254) &                       &           &          &                       &     0.0122 & (0.0625) \\
        $\textbf{W}_{n,*}$ $\times$ ln(Openness) &                       &    0.2343 & (0.2068) &                       &           &          &                       &    -0.0649 & (0.1187) \\
              $\textbf{W}^2_{n,*}$ $\times$ Male &                       &        &       &                       &   -0.2182 & (0.1743) &                       &    -0.0939 & (0.1271) \\
        $\textbf{W}^2_{n,*}$ $\times$ ln(Height) &                       &        &       &                       &   0.9572* & (0.4962) &                       &     0.5604 & (0.4129) \\
        $\textbf{W}^2_{n,*}$ $\times$ ln(Weight) &                       &        &       &                       &   -0.0553 & (0.4378) &                       &    -0.1974 & (0.2979) \\
    $\textbf{W}^2_{n,*}$  $\times$ Siblings Help &                       &        &       &                       &    0.2420 & (0.2819) &                       &     0.2458 & (0.1603) \\
     $\textbf{W}^2_{n,*}$  $\times$ Parents Help &                       &        &       &                       &    0.0493 & (0.1802) &                       &     0.0326 & (0.0824) \\
 $\textbf{W}^2_{n,*}$  $\times$ Commute Car/Taxi &                       &        &       &                       & 0.5050*** & (0.1750) &                       &     0.0522 & (0.1192) \\
             $\textbf{W}^2_{n,*}$ $\times$ Music &                       &        &       &                       &   -0.1076 & (0.0977) &                       &  0.1880*** & (0.0484) \\
     $\textbf{W}^2_{n,*}$ $\times$ ln(Cognitive) &                       &        &       &                       & -0.3588** & (0.1464) &                       &    -0.0596 & (0.1343) \\
 $\textbf{W}^2_{n,*}$ $\times$ ln(Agreeableness) &                       &        &       &                       &    0.0570 & (0.2684) &                       &     0.2458 & (0.2486) \\
 $\textbf{W}^2_{n,*}$ $\times$ ln(Conscientious) &                       &        &       &                       &  0.4525** & (0.2023) &                       &    0.2621* & (0.1535) \\
  $\textbf{W}^2_{n,*}$ $\times$ ln(Extraversion) &                       &        &       &                       &   -0.0010 & (0.1893) &                       &    -0.1465 & (0.1558) \\
   $\textbf{W}^2_{n,*}$ $\times$ ln(Neuroticism) &                       &        &       &                       &    0.0569 & (0.1904) &                       &    -0.0245 & (0.1524) \\
      $\textbf{W}^2_{n,*}$ $\times$ ln(Openness) &                       &        &       &                       &    0.4369 & (0.2845) &                       &    0.3608* & (0.1930) \\
                                      Constant &                       &    0.2689 & (1.1172) &                       &   -0.3995 & (1.6556) &                       &    -0.1877 & (1.4016) \\
                                      \midrule
                                           $n$ &                       &       868 &          &                       &       868 &          &                       &        868 &          \\
                                        $F$-Statistics &                       &   20.2968 &          &                       &   22.4297 &          &                       &    40.1600 &          \\
                                       $p$-value &                       &    0.0000 &          &                       &    0.0000 &          &                       &        0.0000 &          \\
\bottomrule
\end{tabular}

%% file: main.bbl
\begin{thebibliography}{43}
\newcommand{\enquote}[1]{``#1''}
\providecommand{\natexlab}[1]{#1}
\providecommand{\url}[1]{\texttt{#1}}
\providecommand{\urlprefix}{URL }

\bibitem[{Alexander et~al.(2020)Alexander, Piazza, Mekos, and
  Valente}]{L_i_M_public_health}
Alexander, Cheryl, Marina Piazza, Debra Mekos, and Thomas Valente. 2020.
\newblock \enquote{Peers, Schools, and Adolescent Cigarette Smoking.}
\newblock \emph{Journal of Adolescent Health} 29:22--30.

\bibitem[{Athey and Imbens(2017)}]{Athey2017}
Athey, Susan and Guido~W Imbens. 2017.
\newblock \enquote{The Econometrics of Randomized Experiments.}
\newblock In \emph{Handbook of economic field experiments}, vol.~1. Elsevier,
  73--140.

\bibitem[{Auerbach(2022)}]{Auerbach2022}
Auerbach, Eric. 2022.
\newblock \enquote{Identification and Estimation of a Partially Linear
  Regression Model Using Network Data.}
\newblock \emph{Econometrica} 90~(1):347--365.

\bibitem[{Blume et~al.(2015)Blume, Brock, Durlauf, and Jayaraman}]{Blume2015}
Blume, Lawrence~E., William~A. Brock, Steven~N. Durlauf, and Rajshri Jayaraman.
  2015.
\newblock \enquote{Linear Social Interactions Models.}
\newblock \emph{Journal of Political Economy} 123~(2):444--496.

\bibitem[{Boccaletti et~al.(2014)Boccaletti, Bianconi, Criado, del Genio,
  G{\'{o}}mez-Garde{\~{n}}es, Romance, Sendi{\~{n}}a-Nadal, Wang, and
  Zanin}]{Boccaletti2014}
Boccaletti, S., G.~Bianconi, R.~Criado, C.~I. del Genio,
  J.~G{\'{o}}mez-Garde{\~{n}}es, M.~Romance, I.~Sendi{\~{n}}a-Nadal, Z.~Wang,
  and M.~Zanin. 2014.
\newblock \enquote{The Structure and Dynamics of Multilayer Networks.}
\newblock \emph{Physics Reports} 544:1--122.

\bibitem[{Bramoull{\'{e}}, Djebbari, and Fortin(2009)}]{Bramoulle2009}
Bramoull{\'{e}}, Yann, Habiba Djebbari, and Bernard Fortin. 2009.
\newblock \enquote{Identification of Peer Effects through Social Networks.}
\newblock \emph{Journal of Econometrics} 150~(1):41--55.

\bibitem[{Cai and Szeidl(2018)}]{Cai2018}
Cai, Jing and Adam Szeidl. 2018.
\newblock \enquote{Interfirm Relationships and Business Performance.}
\newblock \emph{The Quarterly Journal of Economics} 133~(3):1229--1282.

\bibitem[{Carrell, Fullerton, and West(2009)}]{Carrell2009}
Carrell, Scott~E., Richard~L. Fullerton, and James~E. West. 2009.
\newblock \enquote{Does Your Cohort Matter? Measuring Peer Effects in College
  Achievement.}
\newblock \emph{Journal of Labor Economics} 27~(3):439--464.

\bibitem[{Carrell, Sacerdote, and West(2013)}]{Carrell2013}
Carrell, Scott~E., Bruce~I. Sacerdote, and James~E. West. 2013.
\newblock \enquote{From Natural Variation to Optimal Policy? The Importance of
  Endogenous Peer Group Formation.}
\newblock \emph{Econometrica} 81~(3):855--882.

\bibitem[{Conley, Hansen, and Rossi(2012)}]{Conley2012}
Conley, Timothy~G, Christian~B Hansen, and Peter~E Rossi. 2012.
\newblock \enquote{Plausibly Exogenous.}
\newblock \emph{Review of Economics and Statistics} 94~(1):260--272.

\bibitem[{De~Giorgi, Pellizzari, and Redaelli(2010)}]{Degiorgi2010}
De~Giorgi, Giacomo, Michele Pellizzari, and Silvia Redaelli. 2010.
\newblock \enquote{Identification of social interactions through partially
  overlapping peer groups.}
\newblock \emph{American Economic Journal: Applied Economics} 2~(2):241--75.

\bibitem[{de~Paula(2017)}]{Paula2017}
de~Paula, {\'{A}}ureo. 2017.
\newblock \emph{Econometrics of Network Models}, \emph{Econometric Society
  Monographs}, vol.~1, chap.~8.
\newblock Cambridge University Press, 268--323.

\bibitem[{Doukhan and Louhichi(1999)}]{Doukhan1999}
Doukhan, Paul and Sana Louhichi. 1999.
\newblock \enquote{A New Weak Dependence Condition and Applications to Moment
  Inequalities.}
\newblock \emph{Stochastic Processes and Their Applications} 84~(2):313--342.

\bibitem[{Erd\"{o}s and R\'{e}nyi(1959)}]{Erdos1959}
Erd\"{o}s, P and A~R\'{e}nyi. 1959.
\newblock \enquote{On Random Graphs.}
\newblock \emph{Publicationes Mathematicae Debrecen} 6:290--297.

\bibitem[{Estrada et~al.(2024)Estrada, Estrada, Huynh, Jacho-Ch\'{a}vez, and
  S\'{a}nchez-Arag\'{o}n}]{netivreg}
Estrada, Pablo, Juan Estrada, Kim~P. Huynh, David~T. Jacho-Ch\'{a}vez, and
  Leonardo S\'{a}nchez-Arag\'{o}n. 2024.
\newblock \enquote{\protect{\texttt{netivreg}}: Instrumental Variable
  Estimation in the Linear-in-Means Model.}
\newblock Unpublished Manuscript.

\bibitem[{Gargiulo and Benassi(2000)}]{Gargiulo2000}
Gargiulo, Martin and Mario Benassi. 2000.
\newblock \enquote{Trapped in Your Own Net? Network Cohesion, Structural Holes,
  and the Adaptation of Social Capital.}
\newblock \emph{Organization Science} 11~(2):183--196.

\bibitem[{Goldsmith-Pinkham and Imbens(2013)}]{Goldsmith-Pinkham2013}
Goldsmith-Pinkham, Paul and Guido~W. Imbens. 2013.
\newblock \enquote{Social Networks and the Identification of Peer Effects.}
\newblock \emph{Journal of Business and Economic Statistics} 31~(3):253--264.

\bibitem[{Graham(2020)}]{Graham2020}
Graham, Bryan~S. 2020.
\newblock \enquote{Network Data.}
\newblock In \emph{Handbook of Econometrics}, vol.~7. Elsevier, 111--218.

\bibitem[{Granovetter(1973)}]{Granovetter1973}
Granovetter, Mark~S. 1973.
\newblock \enquote{The Strength of Weak Ties.}
\newblock \emph{American Journal of Sociology} 78~(6):1360--1380.

\bibitem[{Hasan and Koning(2019)}]{Hasan2019}
Hasan, Sharique and Rembrand Koning. 2019.
\newblock \enquote{Prior Ties and the Limits of Peer Effects on Startup Team
  Performance.}
\newblock \emph{Strategic Management Journal} 40~(9):1394--1416.

\bibitem[{Heckman and Rubinstein(2001)}]{Heckman2001}
Heckman, James~J and Yona Rubinstein. 2001.
\newblock \enquote{The Importance of Noncognitive Skills: Lessons from the GED
  Testing Program.}
\newblock \emph{American Economic Review} 91~(2):145--149.

\bibitem[{Hjort(2014)}]{Hjort2014}
Hjort, Jonas. 2014.
\newblock \enquote{Ethnic Divisions and Production in Firms.}
\newblock \emph{The Quarterly Journal of Economics} 129~(4):1899--1946.

\bibitem[{Johnsson and Moon(2019)}]{Johnsson2019}
Johnsson, Ida and Hyungsik~Roger Moon. 2019.
\newblock \enquote{{Estimation of Peer Effects in Endogenous Social Networks:
  Control Function Approach}.}
\newblock \emph{The Review of Economics and Statistics} :1--51.

\bibitem[{Kato and Shu(2016)}]{Kato2016}
Kato, Takao and Pian Shu. 2016.
\newblock \enquote{Competition and Social Identity in the Workplace: Evidence
  from a Chinese Textile Firm.}
\newblock \emph{Journal of Economic Behavior \& Organization} 131:37--50.

\bibitem[{Kelejian and Prucha(1998)}]{Kelejian1998}
Kelejian, Harry~H. and Ingmar~R. Prucha. 1998.
\newblock \enquote{A Generalized Spatial Two-Stage Least Squares Procedure for
  Estimating a Spatial Autoregressive Model with Autoregressive Disturbances.}
\newblock \emph{Journal of Real Estate Finance and Economics} 17~(1):99--121.

\bibitem[{Kelejian and Prucha(1999)}]{Kelejian_Prucha_1999_ER}
---{}---{}---. 1999.
\newblock \enquote{A Generalized Moments Estimator for the Autoregressive
  Parameter in a Spatial Model.}
\newblock \emph{International Economic Review} 40~(2):509--533.

\bibitem[{Kim, Oh, and Swaminathan(2006)}]{Kim2006}
Kim, Tai-Young, Hongseok Oh, and Anand Swaminathan. 2006.
\newblock \enquote{Framing interorganizational network change: A network
  inertia perspective.}
\newblock \emph{Academy of Management Review} 31~(3):704--720.

\bibitem[{Kivela et~al.(2014)Kivela, Arenas, Barthelemy, Gleeson, Moreno, and
  Porter}]{Kivela_multilayer_network_2014}
Kivela, M., A.~Arenas, M.~Barthelemy, J.~P. Gleeson, Y.~Moreno, and M.~A.
  Porter. 2014.
\newblock \enquote{Multilayer Networks.}
\newblock \emph{Journal of Complex Networks} 2~(3):203–271.

\bibitem[{Kojevnikov, Marmer, and Song(2021)}]{Kojevnikov2020}
Kojevnikov, Denis, Vadim Marmer, and Kyungchul Song. 2021.
\newblock \enquote{Limit Theorems for Network Dependent Random Variables.}
\newblock \emph{Journal of Econometrics} 222~(2):882--908.

\bibitem[{Kreager, Rulison, and Moody(2020)}]{L_i_M_criminology}
Kreager, Derek~a., Kelly Rulison, and James Moody. 2020.
\newblock \enquote{Delinquency and the Structure of Adolescent Peer Groups.}
\newblock \emph{Criminology} 49:95--127.

\bibitem[{Lee(2003)}]{Lee2003}
Lee, Lung~Fei. 2003.
\newblock \enquote{Best Spatial Two-Stage Least Squares Estimators for a
  Spatial Autoregressive Model with Autoregressive Disturbances.}
\newblock \emph{Econometric Reviews} 22~(4):307--335.

\bibitem[{Lewbel, Qu, and Tang(2023)}]{Lewbel_Qu_Tang}
Lewbel, Arthur, Xi~Qu, and Xun Tang. 2023.
\newblock \enquote{Social Networks with Unobserved Links.}
\newblock \emph{Journal of Political Economy} 131~(4):898--946.

\bibitem[{Liu, Patacchini, and Zenou(2014)}]{liu2014}
Liu, Xiaodong, Eleonora Patacchini, and Yves Zenou. 2014.
\newblock \enquote{Endogenous Peer Effects: Local Aggregate or Local Average?}
\newblock \emph{Journal of Economic Behavior and Organization} 103:39 -- 59.

\bibitem[{Manski(1993)}]{manski1993}
Manski, Charles~F. 1993.
\newblock \enquote{Identification of Endogenous Social Effects: The Reflection
  Problem.}
\newblock \emph{Review of Economic Studies} 60~(3):531--542.

\bibitem[{Manta et~al.(2022)Manta, Ho, Huynh, and Jacho-Chavez}]{manta2021}
Manta, Alexandra, Anson~T.Y. Ho, Kim~P. Huynh, and David~T. Jacho-Chavez. 2022.
\newblock \enquote{Estimating Social Effects in a Multilayered Linear-in-Means
  Model with Network Data.}
\newblock \emph{Statistics \& Probability Letters} 183:109331.

\bibitem[{Mas and Moretti(2009)}]{Mas2009}
Mas, Alexandre and Enrico Moretti. 2009.
\newblock \enquote{Peers at Work.}
\newblock \emph{American Economic Review} 99~(1):112--145.

\bibitem[{Mizruchi and Neuman(2008)}]{Mizruchi2008}
Mizruchi, Mark~S and Eric~J Neuman. 2008.
\newblock \enquote{The Effect of Density on the Level of Bias in the Network
  Autocorrelation Model.}
\newblock \emph{Social Networks} 30~(3):190--200.

\bibitem[{Moffitt(2001)}]{Moffitt2000}
Moffitt, Robert~A. 2001.
\newblock \enquote{Policy Interventions, Low-Level Equilibria And Social
  Interactions.}
\newblock In \emph{Social Dynamics}, edited by Steven Durlauf and Peyton Young.
  MIT Press, 45--82.

\bibitem[{Neuman and Mizruchi(2010)}]{Neuman2010}
Neuman, Eric~J. and Mark~S. Mizruchi. 2010.
\newblock \enquote{Structure and Bias in the Network Autocorrelation Model.}
\newblock \emph{Social Networks} 32~(4):290--300.

\bibitem[{Qu and Lee(2015)}]{Qu2015}
Qu, Xi and Lung~Fei Lee. 2015.
\newblock \enquote{Estimating a Spatial Autoregressive Model with an Endogenous
  Spatial Weight Matrix.}
\newblock \emph{Journal of Econometrics} 184~(2):209--232.

\bibitem[{Qu, Lee, and Yang(2021)}]{Qu2021}
Qu, Xi, Lung-fei Lee, and Chao Yang. 2021.
\newblock \enquote{Estimation of a SAR model with endogenous spatial weights
  constructed by bilateral variables.}
\newblock \emph{Journal of Econometrics} 221~(1):180--197.

\bibitem[{Sacerdote(2001)}]{Sacerdote_QJE}
Sacerdote, Bruce. 2001.
\newblock \enquote{Peer Effects with Random Assignment: Results for Dartmouth
  Roommates.}
\newblock \emph{Quaterly Journal of Economics} 116~(2):681--704.

\bibitem[{Salmivalli(2020)}]{L_i_M_sociology}
Salmivalli, Christina. 2020.
\newblock \enquote{Bullying and the Peer Group: A Review.}
\newblock \emph{Aggression and Violent Behavior} 15:112--120.

\end{thebibliography}


\begin{thebibliography}{2}
\newcommand{\enquote}[1]{``#1''}
\providecommand{\natexlab}[1]{#1}
\providecommand{\url}[1]{\texttt{#1}}
\providecommand{\urlprefix}{URL }

\bibitem[{Jenish and Prucha(2009)}]{Jenish2009}
Jenish, Nazgul and Ingmar~R Prucha. 2009.
\newblock \enquote{Central limit theorems and uniform laws of large numbers for
  arrays of random fields.}
\newblock \emph{Journal of econometrics} 150~(1):86--98.

\bibitem[{Kojevnikov, Marmer, and Song(2021)}]{Kojevnikov2020}
Kojevnikov, Denis, Vadim Marmer, and Kyungchul Song. 2021.
\newblock \enquote{Limit Theorems for Network Dependent Random Variables.}
\newblock \emph{Journal of Econometrics} 222~(2):882--908.

\end{thebibliography}
